\newif\ifready\readytrue

\documentclass[11pt,letterpaper]{article}
\usepackage[margin=1in]{geometry}
\usepackage{times}
\usepackage{amsmath,amsthm,amssymb}
\usepackage{thm-restate}
\usepackage{dsfont}
\usepackage{physics}
\usepackage{mathtools}
\usepackage{mathrsfs}
\usepackage{nicefrac}
\usepackage[colorlinks=true,allcolors=blue]{hyperref}
\usepackage{bm,color,comment}

\usepackage{algpseudocode}
\usepackage[linesnumbered,ruled,vlined,algosection]{algorithm2e}

\usepackage[capitalise,noabbrev,nameinlink]{cleveref}
\usepackage{booktabs,multirow}

\usepackage{typed-checklist}
\usepackage[shortlabels]{enumitem}

\usepackage[column=O]{cellspace}
\usepackage{makecell}

\usepackage[backref=true, doi=false, backend=biber, style=alphabetic, sortcites=true, sorting=alphabeticlabel,
minbibnames=3, maxbibnames=999, mincitenames=3, maxcitenames=3, minalphanames=3, maxalphanames=3]{biblatex}

\DeclareSortingTemplate{alphabeticlabel}{
  \sort[final]{\field{labelalpha}} %
  \sort{
    \field{year}   %
  }
  \sort{
    \field{title}  %
  }
}

\AtBeginRefsection{\GenRefcontextData{sorting=ynt}}
\AtEveryCite{\localrefcontext[sorting=ynt]}
\crefname{LP}{LP}{LPs}
\crefname{condition}{condition}{conditions}
\Crefname{condition}{Condition}{Conditions}
\crefname{inequality}{inequality}{inequalities}
\Crefname{inequality}{Inequality}{Inequalities}
\crefname{part}{part}{parts}
\Crefname{part}{Part}{Parts}

\definecolor{mydarkblue}{rgb}{0,0.08,0.45}
\hypersetup{ %
    colorlinks=true,
    linkcolor=mydarkblue,
    citecolor=mydarkblue,
    filecolor=mydarkblue,
    urlcolor=mydarkblue,
}

\newcommand{\ie}{\emph{i.e.}}
\newcommand{\eg}{\emph{e.g.}}
\newcommand{\iid}{\emph{i.i.d.}}

\newcommand{\R}{\mathbb{R}}

\newcommand{\E}{\mathop{\mathbb{E}}}

\newcommand{\sset}{\subseteq}

\newcommand{\mcal}{\mathcal}
\newcommand{\mfrak}{\mathfrak}

\newcommand{\ones}{\mathds{1}}

\DeclarePairedDelimiter{\card}{\lvert}{\rvert}
\let\abs\relax
\DeclarePairedDelimiter{\abs}{\lvert}{\rvert}
\let\norm\relax
\DeclarePairedDelimiter{\norm}{\lVert}{\rVert}
\DeclarePairedDelimiter{\set}{\lbrace}{\rbrace}
\DeclarePairedDelimiter{\iprod}{\langle}{\rangle}

\DeclareMathOperator{\cut}{cut}

\DeclareMathOperator{\poly}{poly}
\DeclareMathOperator{\prox}{prox}
\DeclareMathOperator{\argmin}{argmin}

\DeclareMathOperator{\TV}{TV}
\newcommand{\OPT}{\mathrm{OPT}}
\newcommand{\OPTLP}{\mathrm{OPT}_{\mathrm{LP}}}

\newcommand{\expect}{\mathbb{E}}
\newcommand{\eps}{\varepsilon}
\DeclareMathOperator{\EMD}{EMD}
\DeclareMathOperator{\PoBin}{PoBin}

\NewDocumentEnvironment{pf}{o}
  {\IfNoValueTF{#1}{\begin{proof}}{\begin{proof}[Proof of #1.]}}
  {\IfNoValueTF{#1}{\end{proof}}{\end{proof}}}

\ifready
\else

\definecolor{amethyst}{rgb}{0.6, 0.4, 0.8}

\fi

\SetKwProg{Procedure}{Procedure}{}{}
\allowdisplaybreaks

\newtheorem{theorem}{Theorem}[section]
\newtheorem{lemma}[theorem]{Lemma}
\newtheorem{corollary}[theorem]{Corollary}
\newtheorem{remark}[theorem]{Remark}
\newtheorem{proposition}[theorem]{Proposition}

\theoremstyle{definition}
\newtheorem{definition}[theorem]{Definition}
\newtheorem{assumption}[theorem]{Assumption}

\title{Pointwise Lipschitz Continuous Graph Algorithms}
\date{}

\author{%
    \begin{tabular}{cc}	
        \begin{tabular}{c}
            Quanquan C. Liu\\
            Yale University\\
            \texttt{quanquan.liu@yale.edu}
        \end{tabular}
        & 	
        \begin{tabular}{c}
            Grigoris Velegkas\\
            Yale University\\
            \texttt{grigoris.velegkas@yale.edu}
        \end{tabular}
        \\
        \\
        \begin{tabular}{c}
            Yuichi Yoshida\\
            National Institute of Informatics\\
            \texttt{yyoshida@nii.ac.jp}
        \end{tabular}
         & 
         \begin{tabular}{c}
             Felix Zhou\\
            Yale University\\
            \texttt{felix.zhou@yale.edu}
        \end{tabular}
    \end{tabular}
}

\begin{document}
\pagenumbering{gobble}

\maketitle

\begin{abstract}
\noindent In many real-world applications, it is 
undesirable to drastically change the problem solution after a small
perturbation in the input
as unstable outputs can lead to costly transaction fees,
privacy and security concerns,
reduced user trust, and lack of replicability.
Despite the widespread application of graph algorithms,
many classical algorithms are not robust to small input disturbances.
Towards addressing this issue,
we study the \emph{{pointwise} Lipschitz continuity} of graph algorithms,
a notion of stability introduced by \textcite[FOCS'23]{kumabe2022lipschitz}
and further studied in {related} settings \cite[ICALP'24]{kumabe2024allocations}, \cite[SODA'25]{kumabe2023lipschitz,}, \cite[ESA'25]{gima2025courcelle}.

Our main result is a linear programming (LP) based minimum $S$-$T$ cut algorithm 
with a provably optimal Lipschitz constant,
as witnessed by an accompanying lower bound.
As a direct corollary,
we give the \emph{first} dynamic minimum $S$-$T$ cut algorithm with non-trivial recourse bound.
{At the core of our techniques is a novel framework for analyzing the Lipschitz constant of regularized LP relaxations.}
{Our framework} crucially unlocks the use of weighted regularizers,
which could not be analyzed through previous methods
and leads to polynomial improvements in the Lipschitz constant compared to what is achievable through previous techniques.
To demonstrate the flexibility of our methods,
we also design an LP-based $\bm b$-matching algorithm
that improves on the state-of-the-art~\cite{kumabe2022lipschitz} Lipschitz constant {in certain input regimes} when $\bm b\equiv 1$.
Moreover,
our algorithm cleanly {extends} to the general case when $\bm b\geq 1$,
whereas \cite{kumabe2022lipschitz} is specialized to the case of $\bm b\equiv 1$.
\end{abstract}

\clearpage
\setcounter{tocdepth}{2}
\tableofcontents

\clearpage
\pagenumbering{arabic}

\section{Introduction}\sloppy
The stability of algorithmic solutions
is an important property for many real-world applications,
where stability is generally
defined as how much the solution to a problem changes when 
small perturbations to the input occur. 
This topic has long been studied by the algorithmic community
in various domains including privacy~\cite{watson2020stability,wang2016learning}, replicability~\cite{impagliazzo2022reproducibility}, 
perturbation resilience~\cite{awasthi2012center,angelidakis2017algorithms,balcan2016clustering,bilu2012stable}, 
bounded recourse algorithms~\cite{davis2006online,gupta2017online,cohen2019fully,bernstein2025matching}, average sensitivity~\cite{Yoshida2021,Kumabe22,varma2021average,kumabe_et_al:LIPIcs.ESA.2022.75,Peng2020}, and
Lipschitz continuous graph algorithms~\cite{kumabe2022lipschitz,kumabe2023lipschitz,kumabe2024allocations,gima2025courcelle}. 

As a concrete example, 
consider the 
task of identifying bottlenecks in a traffic network,
where streets are modeled as edges
with the maximum amount of traffic as edge capacities,
and intersections are modeled as nodes.
Hence,
the underlying task is to find a \emph{minimum $s$-$t$ cut} between two locations $s, t$. 
In this setting,
stability of the solution is crucial because minor changes in capacity (edge weights),
due to \eg{}, weather conditions,
should not drastically alter the output cut 
{as this could disrupt operations and incur costs}.

In this work,
{our goal is to design graph algorithms that satisfy two desiderata: accuracy and stability. Notice that, without taking into account the quality of the
solution, one can provide trivial algorithms
that are perfectly stable under any perturbation of
the input.
{However, algorithms that achieve perfect stability 
through overly simplistic methods (\eg{}, always outputting the same value) often suffer from poor accuracy.}

Towards this end, 
we study \emph{pointwise Lipschitz continuity},
a notion of algorithmic stability introduced by \textcite{kumabe2022lipschitz}}.\footnote{To be more precise, \cite{kumabe2022lipschitz} introduced \emph{two} related notions of Lipschitz continuity, the \emph{weighted} and \emph{unweighted} one. They focused mostly on designing algorithms that satisfy the weighted version, whereas we extensively study the unweighted version; See \Cref{sec:related-work} for a detailed comparison.}
Informally speaking, the Lipschitz constant of an algorithm
is defined by considering the expected symmetric difference ($\ell_1$-distance) of the outputs of the algorithm 
under sufficiently small perturbations of edge-weights, 
where the expectation is taken with respect to the internal randomness of the algorithm. 
Specifically,
we design pointwise Lipschitz algorithms for the minimum $S$-$T$ cut and $\bm b$-matching problems.

Our main algorithmic tool is regularized linear programming relaxations with natural choices of \emph{weighted} regularizers
to ensure the stability of the optimal fractional solution.
Paired with an appropriate stable rounding scheme, 
this approach can yield \emph{optimal}
stability-approximation trade-offs, 
as illustrated by a tight lower bound we provide for the minimum $S$-$T$ cut problem. 
As a direct corollary of our minimum $S$-$T$ cut algorithm,
we achieve the first dynamic minimum $S$-$T$ cut algorithm with non-trivial recourse.

Prior to our work, \cite{kumabe2022lipschitz} also leveraged regularized LP-relaxations to design a pointwise Lipschitz algorithm for maximum bipartite matching using an \emph{unweighted} entropy regularizer~\cite{cuturi2013sinkhorn}.
Their stability analysis of the optimal fractional solution,
based on first-order optimality conditions for strongly convex functions,
heavily relies on the fact that the regularizer is independent of the edge weights.
Unfortunately,
the use of unweighted regularizers can lead to suboptimal stability-approximation tradeoffs by polynomial factors, 
as illustrated in \Cref{sec:min s-t cut prev work limits}.
Moreover,
this seems to be an inherent limitation of their technique and not just an artifact of the specific choice of unweighted regularizer (see again \Cref{sec:min s-t cut prev work limits}).

Similar to \cite{kumabe2022lipschitz},
we also add a smooth, strongly convex regularizer to the natural LP relaxation of the underlying graph problem.
However,
we depart from the direct application of first-order optimality conditions when analyzing the stability of optimal solutions
and instead analyze the stability of iterates generated by the proximal gradient method,
a first-order optimization method that is known to converge.
This leads to a general framework for analyzing the Lipschitz constant of regularized LP relaxations
we call the \emph{proximal gradient trajectory analysis} (PGTA).
{Importantly,} our bounds hold \textbf{regardless of the specific optimization algorithm employed to find the solution},
but we choose proximal gradient iterates as it is suitable for composite functions (\eg{}, an objective function plus a regularizer).
While the stability of solutions to strongly convex programs under small perturbations may be intuitive, 
establishing tight bounds on the $\ell_1$ distance between optimal solutions 
as a function of the input perturbation in $\ell_1$ distance,
is a challenging problem.
One of the reasons is that pointwise Lipschitz continuity requires us to measure both the input and output perturbations using the $\ell_1$-norm,
rather than the dual $\ell_0$-$\ell_1$ norm pairing (or vice versa).

Beyond leading to improved bounds for the minimum $S$-$T$ cut problem,
PGTA unifies the analysis of the Lipschitz constant of convex relaxations for many problems under one clean and easy-to-use framework.
For example,
we design a bipartite matching algorithm that improves the stability-approximation tradeoffs 
compared to the \cite{kumabe2022lipschitz} bipartite matching algorithm
for {certain input regimes}.
More importantly,
our algorithm and its analysis cleanly generalizes to the case of bipartite $\bm b$-matchings when $\bm b\geq 1$,
whereas the \cite{kumabe2022lipschitz} algorithm is hardwired for $1$-matchings (see \Cref{sec:mwm prev work limits}).
{To further demonstrate the flexibility of our approach,
we also design a pointwise Lipschitz algorithm for packing integer programs.}

We note that on top of applying PGTA to obtain stable fractional solutions,
designing end-to-end algorithms requires designing new stable rounding techniques that are problem-specific and technically involved.
Such problem-dependent rounding is necessary as rounding fractional LP solutions even without Lipschitz constraints is in general a difficult problem.

\subsection{Our Results}\label{sec:results}
Before we discuss our results, we give 
an informal description of
the pointwise Lipschitz constant of a graph algorithm \cite{kumabe2022lipschitz}. For the formal definition,
we refer the reader to \Cref{sec:pointwise-lipschitz}. 
We say that a randomized algorithm $\mcal A$ that operates
on an $n$-vertex $m$-edge graph $G=(V,E)$ with weights $\bm w \in \R_{\geq 0}^{E}$ is pointwise Lipschitz continuous with pointwise Lipschitz constant
$P_{G, \bm w}$ if the earth mover's distance (EMD) 
between the output distributions of $\mcal A$ 
on inputs $(G, \bm w)$ and $(G, \tilde{\bm w})$ is at most $P_{G, \bm w}\cdot \norm{\bm w - \tilde{\bm w}}_1.$ 
Here, $\tilde{\bm w}$ is assumed
to be in a sufficiently small neighborhood of $\bm w$. Crucially,
the pointwise Lipschitz constant is allowed to depend
on the structure of the graph $G$ and the particular weight vector $\bm w$.
The reason for this becomes more apparent in our
discussion of the formal definition {in~\cref{sec:def-remarks}}.

Our algorithms produce $(\zeta, \alpha)$-bicriteria 
approximations where $\zeta$ is the multiplicative
factor in the approximation and $\alpha$ is the additive error.
We write $\zeta$-approximation to indicate a $(\zeta, 0)$-approximation algorithm as a shorthand.

For graphs undergoing edge insertions/deletions,
the recourse of a dynamic algorithm is defined as the size of the symmetric difference of output sets between consecutive updates,
\eg{}, the change in the output between consecutive updates.

\subsubsection{Minimum Vertex \texorpdfstring{$S$-$T$}{S-T} Cut}
We begin by informally stating a pointwise Lipschitz algorithm for the minimum $S$-$T$ cut problem 
that achieves an asymptotically tight pointwise Lipschitz constant, 
{which we demonstrate via a matching lower bound.}

Let $S, T\sset V$ be disjoint vertex subsets.
We write $\partial A \coloneqq \set{uv\in E: u\in A, v\notin A}$.
The \emph{minimum vertex $S$-$T$ cut} problem asks for a vertex subset $A\supseteq S$ 
that minimizes the cut weight $\bm w(\partial A) = \sum_{e\in \partial A} w_e$
subject to the constraint $A\cap T = \varnothing$.
Note that this is a generalization of the minimum vertex $s$-$t$ cut problem.

\begin{theorem}[Informal; See \Cref{thm:s-t-min-cut-additive}, \Cref{thm:s-t-cut-additive-lower-bound}]\label{thm:s-t-cut-informal}
    There is an efficient $(1+O(\nicefrac1{\sqrt{n}}), O(\lambda_2))$-approximation algorithm for minimum $S$-$T$ cut
    with pointwise Lipschitz constant $O(\nicefrac{n}{\lambda_2})$,
    where $\lambda_2$ is the second smallest eigenvalue of the unnormalized Laplacian matrix of the input graph.
    Moreover,
    any $(O(1), O(\lambda_2))$-approximation algorithm must have pointwise Lipschitz constant $\Omega(\nicefrac{n}{\lambda_2})$.
\end{theorem}

{
It is known that a $(\zeta, \alpha)$-approximate (pointwise) Lipschitz algorithm with (pointwise) Lipschitz constant $K$ can be black-box transformed to a dynamic $(\zeta, \alpha)$-approximation algorithm where the expected recourse is at most $K$~\cite[Theorem 1.6]{kumabe2022lipschitz}.
Thus, as a direct corollary of our minimum $s$-$t$ cut algorithm,
we obtain the first dynamic minimum $s$-$t$ cut algorithm with non-trivial recourse.
\begin{corollary}
    There is a dynamic $(1+O(\nicefrac1{\sqrt{n}}), O(\lambda_2))$-approximation algorithm for minimum $S$-$T$ cut
    with expected recourse $O(\nicefrac{n}{\lambda_2})$,
    where $\lambda_2$ is the second smallest eigenvalue of the unnormalized Laplacian matrix of the current input graph.
\end{corollary}
}

We also consider the minimum $\beta$-balanced vertex $S$-$T$ cut problem,
which is a variant of the minimum vertex $S$-$T$ cut problem
where we additionally restrict the feasible cut sets $A$ to have at least $\beta n$ vertices and at most $(1-\beta)n$ vertices.

\begin{theorem}[Informal; See \Cref{thm:s-t-cut-vertex-set}]%
    There is an efficient $O(\nicefrac1\beta)$-approximation algorithm to the $\beta$-balanced $S$-$T$ cut problem
    with pointwise Lipschitz constant $O\left( \frac{\sqrt{n}}{\beta^2 \lambda_2} \right)$,
    where $\lambda_2$ is the second smallest eigenvalue of the unnormalized Laplacian matrix of the input graph.
\end{theorem}

We summarize our results for the minimum vertex $S$-$T$ cut
in \Cref{tab:cut-results}.
\begin{table}[hbtp!]
    \centering
    \renewcommand{\arraystretch}{1.75}
    \begin{tabular}{llll}
        \toprule
        Problem & Approximation  & \makecell{Pointwise Lipschitz\\Constant}  & Reference \\
        \midrule
        \makecell[l]{Minimum Vertex \\ $S$-$T$ Cut} & $\left( 1+o(1) \right) \cdot \mathrm{OPT}+ O(\lambda_2)$ & $O\left( \frac{n}{\lambda_2} \right)$ & \Cref{thm:s-t-min-cut-additive} \\
        & $O(1) \cdot \OPT + O(\lambda_2) $ & $\Omega\left(\frac{n}{\lambda_2}\right)$ & \Cref{thm:s-t-cut-additive-lower-bound} \\
        \makecell[l]{Minimum $\beta$-Balanced \\ Vertex $S$-$T$ Cut} & $O(\beta^{-1}) \cdot \mathrm{OPT}_\beta$ & $O\left( \frac{\sqrt{n}}{\beta^2 \lambda_2}\right)$ & \Cref{thm:s-t-cut-vertex-set} \\
        \bottomrule
    \end{tabular}
    \caption{Our results. 
    For graph problems, $n$ and $m$ denote the number of vertices and edges, respectively.
    For the minimum balanced $S$-$T$ cut problem, $\OPT_\beta$ denotes the minimum cut weight over vertex subsets of size at least $\beta n$ and at most $(1-\beta)n$.
    $\lambda_2$ denotes the second smallest eigenvalue of the input graph.
    }
    \label{tab:cut-results}
\end{table}

{The tight upper and lower bounds of the pointwise Lipschitz constant for minimum $s$-$t$ cut
suggest that algorithms achieving strong approximation guarantees for this problem must be unstable in general.
This is also supported by strong lower bounds from the differential privacy literature (see \Cref{sec:related-work}).
Our work contributes a fine-grained quantification of the particular instances for which this instability holds, parameterized by $\lambda_2$.}

\subsubsection{Bipartite $\bm b$-Matching}
Recall the maximum weight (bipartite) matching problem
where we are given a (bipartite) graph and want to output an edge set $M\sset E$
such that every vertex is adjacent to at most one edge of $M$
while maximizing the matching weight $\bm w(M) = \sum_{e\in M} w_e$.

\begin{theorem}[Informal; 
See \Cref{thm:mwm-bipartite}
]
    There is an efficient $(2+\eps)$-approximation algorithm for maximum weight bipartite matching
    with pointwise Lipschitz constant $O\left( \frac{\sqrt{m}}{\eps w_{\min}} \right)$.
    Here, $w_{\min}$ is the minimum input edge weight.
\end{theorem}

We note that \cite{kumabe2022lipschitz}
obtain a $(2+\eps)$-approximate maximum weight bipartite matching with pointwise Lipschitz constant 
$O\left(\eps^{-1} n^{3/2} \log(m)/\OPT\right)$. 
In cases where the graph is sparse, \ie{}, $m = \tilde{O}(n)$, 
the number of edges in the optimum matching is small, \eg{}, $o(n)$, 
and the edge weights are similar in magnitude, 
our algorithm improves on their pointwise Lipschitz bounds. 

We also study the maximum weight (bipartite) $\bm b$-matching problem,
a variation of the maximum weight (bipartite) matching problem, 
where each vertex $v$ can be adjacent to at most $b_v$ edges of the matching.
Note the case of $\bm b\equiv 1$ recovers the maximum weight (bipartite) matching problem.

\begin{theorem}[Informal; 
See \Cref{thm:maximum bipartite b-matching}
]
    There is an efficient $\left( \frac{2+\eps}{1-\nicefrac1e} \right)$-approximation algorithm for maximum weight bipartite $\bm b$-matching
    with pointwise Lipschitz constant $O\left( \frac{\sqrt{m}}{\eps w_{\min}} \right)$.
    Here, $w_{\min}$ is the minimum input edge weight.
\end{theorem}

{We remark that our bipartite matching and $\bm b$-matching algorithms
can be adapted to general graphs with the same Lipschitz constants 
through random bipartitions,
but at a $2$-factor cost in the approximation ratio.}

Finally,
we study packing integer programs (PIP),
where given an input matrix $A\in [0, 1]^{p\times m}$ and budget $\bm b\in \R_{\geq 1}^p$,
we ask for a binary vector $\bm y\in \set{0, 1}^m$
maximizing its weight $\bm w^\top \bm y$
subject to the constraints $A\bm y\leq \bm b$.
Note the case where $A$ is the edge-incidence matrix of a graph
and $\bm b$ is integral recovers the maximum weight $\bm b$-matching problem.
{It is important to note that while the generality of PIPs may present challenges in achieving the same tight 
bounds as problem-specific approaches, this result nevertheless highlights the versatility and broad applicability of our framework.}
\begin{theorem}[Informal; See \Cref{thm:PIP}]
    There is an efficient $O(p^{1/B})$-approximation algorithm for packing integer programs
    with pointwise Lipschitz constant $O\left( \frac{\sqrt{m}}{w_{\min} p^{1/B}} \right)$,
    where $w_{\min}$ denotes the minimum input weight
    and $B\coloneqq \min(\bm b)\geq 1$ is the minimum input budget.
\end{theorem}
We summarize our results for the maximum bipartite matching {problem} and its variants in \Cref{tab:matching-results}.
\begin{table}[htbp!]
    \centering
    \renewcommand{\arraystretch}{1.75}
    \begin{tabular}{llll}
        \toprule
        Problem & Approximation  & \makecell{Pointwise Lipschitz\\Constant}  & Reference \\
        \midrule
        Bipartite Matching & $\frac{1}{2+\eps}\cdot \OPT$ & $O\left(\frac{\sqrt{m}}{\eps w_{\min}}\right)$ & \Cref{thm:mwm-bipartite} \\
        Bipartite $\bm b$-Matching & $\frac{1}{3.164+\eps}\cdot \mathrm{OPT}$ & $O\left(\frac{\sqrt{m}}{\eps w_{\min}}\right)$ & \Cref{thm:maximum bipartite b-matching} \\
        \makecell[l]{Packing\\Integer Programs} & $O\left( \frac1{p^{1/B}} \right)\cdot \mathrm{OPT}$ & $O\left(\frac{\sqrt{m}}{w_{\min} p^{1/B}}\right)$ & \Cref{thm:PIP} \\
        \bottomrule
    \end{tabular}
    \caption{Our results for weighted ($\bm b$-)matching and packing integer programs. 
    {$\eps > 0$ denotes the approximation parameter.}
    For graph problems, $n$ and $m$ denote the number of vertices and edges, respectively.
    For the matching problems, $w_{\min}$ denotes the minimum weight of an edge.
    For the packing integer programs, $m$ is the number of variables, $p$ is the number of constraints, $B$ is the minimum bound of a packing constraint, and $w_{\min}$ is the minimum weight of a variable.}
    \label{tab:matching-results}
\end{table}

\section{Technical Overview}
In this section, we provide
an overview of the techniques we use
to obtain our results.
We first detail the proximal gradient trajectory analysis framework in \Cref{sec:overview:pgta}.
Then,
we describe our minimum $S$-$T$ cut algorithm
and the accompanying lower bound in \Cref{sec:overview:min s-t cut}.
Next,
we overview our 
($\bm b$-)matching algorithm in \Cref{sec:overview:bipartite b-matching}.
In \Cref{sec:overview:packing IP},
we describe our algorithm for the broad generalization of packing integer programs.
Finally,
we discuss how to implement our pointwise Lipschitz algorithms using shared randomness in \Cref{sec:overview:shared randomness}.

\subsection{PGTA Framework (\Cref{sec:framework})}\label{sec:overview:pgta}
Our first main contribution is a general framework bounding the Lipschitz constant of regularized convex programs. 
Our approach builds on the LP-based technique of \cite{kumabe2022lipschitz},
who used regularized LP relaxations to obtain stable fractional solutions as an important intermediate step.
{However,
our analysis departs from their direct analysis using first-order optimality conditions
and more subtly controls the Lipschitz constant by analyzing the iterates of a well-behaved process that converges to the optimal solution.}

{Before outlining our approach in \Cref{sec:pgta our approach},
we sketch the techniques from prior works that bound the Lipschitz constant of regularized convex programs in \Cref{sec:pgta prev work comparison}.}

\subsubsection{Prior Works \& Challenges}\label{sec:pgta prev work comparison}
Before proceeding further,
we briefly explain the approach of \cite{kumabe2022lipschitz} for controlling the Lipschitz constant.
Let $h(\bm x) = f(\bm x) + \Lambda g(\bm x)$ where, for simplicity, $g$ is a 1-strongly convex regularizer and $\Lambda > 0$ is the regularization parameter.
We use $\tilde h, \tilde f, \Lambda \tilde g$ to denote the perturbed objective functions.
We write $\bm x^\star, \tilde{\bm x}^\star$ to be the optimal solutions of the original and perturbed instances, respectively.
By directly using 
{first-order optimality conditions as well as}
the definition of strong convexity,
it can be shown that (see \eg{}, \cite[Section 5]{chen2023private})
\[
    \Lambda \norm{\bm x^\star - \tilde{\bm x}^\star}^2
    \leq h(\tilde{\bm x}^\star) - h(\bm x^\star) + \tilde h(\bm x^\star) - \tilde h(\tilde{\bm x}^\star).
\]
When $g = \tilde g$ and $f, \tilde f$ are linear,
\cite{kumabe2022lipschitz} observed that
\begin{align*}
    \Lambda \norm{\bm x^\star - \tilde{\bm x}^\star}^2
    &\leq f(\tilde{\bm x}^\star) - f(\bm x^\star) + \tilde f(\bm x^\star) - \tilde f(\tilde{\bm x}^\star)
    \leq \delta \norm{\bm x^\star - \tilde{\bm x}^\star}\,, \\
    \norm{\bm x^\star - \tilde{\bm x}^\star}
    &\leq \frac\delta\Lambda\,.
\end{align*}
Here $\delta$ is the magnitude of perturbation.
Crucially,
{in order to obtain a bound that scales \emph{linearly} with the weight perturbation $\delta$,}
this analysis relies on using the same regularizer $\tilde g = g$ on the perturbed instance. 
{Unfortunately},
as we show in \Cref{sec:min s-t cut prev work limits},
choosing such an unweighted regularizer yields sub-optimal Lipschitz constants for minimum vertex $S$-$T$ cut.
In comparison,
PGTA allows us to use the weighted $\ell_2$-regularizer which yields the optimal Lipschitz constant for minimum vertex $S$-$T$ cut,
as witnessed by our algorithm and matching 
bound in \Cref{sec:min-cut-through-pgm}.

Moreover, in \Cref{sec:mwm prev work limits},
we show that the coarse analysis above offers less flexibility than PGTA.
Specifically,
this analysis does not allow us to extend the matching algorithm of \cite{kumabe2022lipschitz},
which relies on the entropy regularizer, 
to the natural generalization of $\bm b$-matchings.
On the other hand,
PGTA allows us to use the weighted $\ell_2$-regularizer which does extend to $\bm b$-matchings.

\subsubsection{Our Approach}\label{sec:pgta our approach}
First, we consider an appropriate
convex relaxation of the underlying graph problem on a graph with edge weights $\bm w$.
Then, we add a \emph{$\sigma$-strongly convex} and \emph{$L$-smooth} regularizer to our objective.
We remark that even though the regularizer we add
depends on the problem, the structure of the objective
function and the approximation guarantees we are aiming
for often dictate its choice.
{Subsequently, we show that the optimal solution corresponding to any slightly perturbed edge weights $\tilde{\bm w}$
will be close to the optimal solution corresponding to our original edge weights $\bm w$.}

Our proof is inspired by the work of \textcite{hardt2016train}
and goes through the proximal gradient trajectory analysis (PGTA), 
where we show
that the {distance between the} trajectories of the proximal gradient method (PGM) under $\bm w$ and $\tilde{\bm w}$
{scales linearly with $\norm{\bm w-\tilde{\bm w}}_1$}. 
{While \cite{hardt2016train} also show that various first-order algorithms are stable,
their notion of stability {has a more discrete flavor} and is defined on neighboring datasets that differ in one entry,
whereas we require our algorithms to be continuously stable across instances parameterized by weights $\bm w$ and $\tilde{\bm w}$.
Moreover,
their definition is designed for deriving generalization bounds when using first-order algorithms for statistical learning tasks and requires executing the exact algorithms they analyze.}
On the other hand,
\textbf{PGM is used merely to establish the closeness of the optimal points} of
these two convex programs; 
solving them to a reasonable precision using any algorithm is enough
to guarantee the closeness of the solutions.

{
Roughly speaking,
for given initial points $x^{(0)}$,
PGM generates a sequence of points $x^{(t+1)} = G(x^{(t)})$
that converge to the optimal solution $x^{(t)}\to x^\star$
as $t\to \infty$,
where $G$ is an update operator.
As we will see,
$G, \tilde G$ are contractions
and are sufficiently well-behaved across perturbed instances so that
\begin{align*}
    \norm{G(x^{(t)}) - \tilde G(\tilde x^{(t)})}
    &\leq (1-\nicefrac{\sigma}{L}) \norm{x^{(t)} - \tilde x^{(t)}} + O(\norm{\bm w-\tilde{\bm w}}_1) \\
    &\leq (1-\nicefrac{\sigma}{L})^{t} \norm{x^{(0)} - \tilde x^{(0)}} + \sum_{\tau=1}^t (1-\nicefrac{\sigma}{L})^\tau\cdot  O(\norm{\bm w-\tilde{\bm w}}_1) \\
    &\leq O\left( \frac{\norm{\bm w-\tilde{\bm w}}_1}{\nicefrac{\sigma}{L}} \right)\,.
\end{align*}
Taking the limit as $t\to \infty$ yields a perturbation bound on the optimal fractional solutions $x^\star, \tilde x^\star$.
Crucially,
our analysis relies only on the $\sigma$-strong convexity and $L$-smoothness of the regularizer and does \emph{not} require the regularizer to be identical on perturbed instances.
This opens the door to {using} weighted regularizers,
which are the key to obtaining tight pointwise Lipschitz guarantees for minimum vertex $S$-$T$ cut
and extending our matching algorithm to the more general setting of $\bm b$-matchings.
}

Finally,
we need to round the continuous solution of the relaxed program
to a feasible discrete one in a way that (approximately)
maintains both the quality of the objective
and the pointwise Lipschitz guarantee. The rounding schemes we use are problem-specific
and are usually the most technically involved steps of our approach.

\subsection{Minimum Vertex $S$-$T$ Cut (\Cref{sec:min-cut-through-pgm})}\label{sec:overview:min s-t cut}
Let $G=(V,E)$ be an undirected graph with edge weights $\bm w \in \R_{\geq 0}^{E}$ and $S, T\sset V$ be disjoint vertex subsets.
For $A\sset V$,
we write $\partial A \coloneqq \set{uv\in E: u\in A, v\notin A}$.
The \emph{minimum vertex $S$-$T$ cut} problem asks for a vertex subset $A\supseteq S$ such that $A\cap T = \varnothing$ 
minimizing the cut weight $\bm w(\partial A) = \sum_{e\in \partial A} w_e$.
Note that this is a generalization of the minimum vertex $s$-$t$ cut problem.

{
Before diving into the details of our $S$-$T$ cut results,
we briefly outline existing techniques and their limitations in \Cref{sec:min s-t cut prev work limits}.
Then,
\Cref{sec:min s-t cut our approach} overviews our algorithmic approach.
Finally,
we describe a matching lower bound in \Cref{sec:s-t cut lower bound overview}.
}

\subsubsection{Existing Techniques \& Limitations (\Cref{sec:naive min s-t cut})}\label{sec:min s-t cut prev work limits}

Before summarizing our algorithm,
we work through the approach of prior works.
{As mentioned in \Cref{sec:pgta prev work comparison},
\cite{kumabe2022lipschitz} directly uses the definition of strong convexity to control the change in output
in response to $\ell_1$ input perturbations.
In their coarse analysis,
it is crucial when analyzing two regularized LPs that differ by a small weight perturbation
that the regularizer used is \emph{exactly} the same.
In particular,
their analysis would not apply for the weighted $\ell_2$-regularizer we use below.

Indeed,
if we were to directly apply the analysis of \cite{kumabe2022lipschitz} for the minimum vertex $S$-$T$ cut problem,
we would need to choose an unweighted regularizer,
of which the $\ell_2$-regularizer $\norm{\cdot}_2^2$ is the most canonical.
As shown in \Cref{sec:naive min s-t cut},
following the steps of previous work would yield a fractional Lipschitz algorithm
with $O(\Lambda n)$ additive error (no multiplicative error)
and a Lipschitz constant of $O(n/\Lambda)$ for some regularization parameter $\Lambda$.
{In order to match the additive error of $O(\lambda_2)$ in \Cref{thm:s-t-cut-informal},
we must set the regularization parameter to $\Lambda = \nicefrac{\lambda_2}{n}$.
However,
this yields a Lipschitz constant of $\nicefrac{n^{2}}{\lambda_2}$,
which is strictly worse than the guarantees of \Cref{thm:s-t-cut-informal}.}
The key to obtaining this improved bound is using a weighted regularizer, which cannot be directly analyzed using the coarse technique of \cite{kumabe2022lipschitz}.
In comparison, the analysis of PGTA holds in this setting.

{This bottleneck seems to be inherent to their analysis
rather than a specific limitation of the unweighted $\ell_2$-regularizer.
Indeed,
the analysis of \cite{kumabe2022lipschitz} suffers an additive error proportional to the maximum range of the chosen regularizer.
Meanwhile,
it is known that any $\Lambda$-strongly convex function over $[-1, 1]^n$
(even with respect to the $\ell_\infty$-norm)
must have range $\Omega(\Lambda n)$.
This phenomenon,
informally known as the ``$\ell_\infty$-barrier'',
was responsible for the previous stalled progress on several important optimization problems~\cite{sherman2017area,sidford2018coordinate}.
PGTA sidesteps this barrier as the derived fractional Lipschitz constant does not directly depend on the range of the regularizer.}

See {\Cref{sec:naive min s-t cut}} for a complete algorithm obtained by following the steps taken in \cite{kumabe2022lipschitz} using the unweighted $\ell_2$-regularizer.
}

\subsubsection{Our Approach}\label{sec:min s-t cut our approach}
Our algorithm first solves a standard minimum $S$-$T$ cut LP relaxation regularized by the quadratic form of the unnormalized Laplacian (\Cref{subsec:s-t-cut-fractional}).
We can then directly bound the Lipschitz constant of the optimal fractional solution by applying the PGTA framework.
Let $\lambda_2$ be the \emph{algebraic connectivity} of the input graph, that is, the second smallest eigenvalue of the unnormalized Laplacian of the graph.
Roughly speaking,
the optimal fractional solution is a $(1+\eps)$-approximation
with Lipschitz constant $O(\nicefrac1{\eps \lambda_2})$,
where $\eps > 0$ is a tunable parameter.
Given stable fractional solutions,
we can then apply the classical threshold rounding to obtain integral solutions.

One technical complication we need to handle,
is that in order to apply our framework,
we need to ensure that the objective function is strongly convex. 
Thus, due to the natural choice
of the Laplacian regularizer, we need
to optimize over a feasible region that is orthogonal to the all-ones vector.
{In particular,
the feasible region is a subset of $[-1, 1]^n$ rather than the standard $[0, 1]^n$.}
As a result, 
{we must perform threshold rounding} with a random threshold in $[-1, 1]$,
which returns a {non-trivial vertex} set with probability merely $\nicefrac12$
and we must take additional steps to boost the success probability.
{Notice that, in the absence
of the pointwise Lipschitz continuity requirement,
boosting the success probability
of such an algorithm is straightforward; 
we can just run multiple independent
instances of the algorithm and, with 
high probability, we will obtain
a feasible output. 
{The key difficulty is smoothly selecting the feasible output,
as it might be completely different
in two executions of the algorithm under
perturbed weights.}

We design two different stable boosting algorithms: 
the first utilizes
the exponential mechanism~\cite{varma2021average} (modified from~\cite{mcsherry2007mechanism})
to smoothly select a feasible output from independent executions
and the second
circumvents the selection issue by
taking advantage of the \emph{$k$-way submodularity property}~\cite{harvey2006capacity} of the objective
function
to \emph{combine} solutions.
We believe
that our boosting mechanisms may be of independent interest.
This leads to two algorithms for the minimum vertex $S$-$T$ cut problem.

\paragraph{Exponential Mechanism Boosting (\Cref{subsec:exp-mech-vertex-S-T-cut}).}
As mentioned,
vanilla threshold rounding gives a feasible solution with probability merely $\nicefrac12$.
To achieve a higher success probability, we first compute a fractional solution assuming that the optimal cut (as a vertex set) has size around $\gamma i\cdot n$ for each $i \in \{1,\ldots,1/\gamma\}$, where $\gamma>0$ is a small constant, and then apply the exponential mechanism to select the index $i^*$.
Then, we apply threshold rounding to the fractional solution corresponding to $i^*$, which leads to a success probability of $1-\gamma$.

The optimal choice of $\eps > 0$ above is $\eps = \nicefrac1{\sqrt{n}}$
when rounding using the exponential mechanism.
This leads to a polynomial-time $(1+O(\nicefrac1{\sqrt{n}}), O(\lambda_2))$-approximation algorithm with pointwise Lipschitz constant $O(\nicefrac{n}{\lambda_2})$.
We remark that the additive error is a product of the rounding scheme
and is not present in the fractional solution.
We show that this algorithm is tight in the sense that even an $(O(1),\lambda_2)$-approximation algorithm must have pointwise Lipschitz constant $\Omega(\nicefrac{n}{\lambda_2})$. 
{Thus, our PGTA framework can achieve tight pointwise Lipschitz bounds given a sufficiently stable rounding scheme.}

\paragraph{$k$-Way Submodularity Boosting (\Cref{subsec:vertex-S-T-cut}).}
The second algorithm is a polynomial-time algorithm that, given a graph $G=(V,E)$, a weight vector $\bm w \in \mathbb{R}_{\geq 0}^E$, and a parameter $\beta \in (0,\nicefrac12)$, outputs a vertex set of cut weight at most $O(\beta^{-1}) \cdot \OPT_\beta$ with pointwise Lipschitz constant $O(\beta^{-2}\lambda_2 \sqrt{n})$, where $\OPT_\beta$ is the minimum weight of a cut $A$ with $\beta n \leq |A|\leq (1-\beta) n$.

In contrast to the first boosting algorithm, 
the second algorithm has no additive error, but the multiplicative error is with respect to the $\beta$-balanced cuts. {The reason
why we consider $\beta$-balanced cuts instead of arbitrary cuts
is to ensure that the success probability
of the thresholding-based rounding 
is bounded away from $\nicefrac12$ by at least
$\beta.$}
In this algorithm, we first compute vertex sets $A_1,\ldots,A_k$ by independently applying the threshold rounding multiple times to the (single) fractional solution assuming that the optimal cut (as a vertex set) has size between $\beta n$ and $(1-\beta) n$.
Then, we utilize the $k$-way submodularity~\cite{harvey2006capacity} of the cut function to combine these vertex sets into a single vertex set with a small pointwise Lipschitz constant.
Using the fact that at least half of $A_1,\ldots,A_k$ are feasible (with high probability, because $0 < \beta < \nicefrac12$), we can show that the output set is feasible (with high probability).
The $k$-way submodularity {property} guarantees that the output set is a $O(\beta^{-1})$-approximation.

\subsubsection{Lower Bound (\Cref{subsec:s-t-cut-lower-bound})}\label{sec:s-t cut lower bound overview}
In order to show the optimality of
our framework, we show a matching 
lower bound for the minimum $S$-$T$ cut problem. 
Our approach consists of two main steps: first
we consider a family of graphs, parametrized by the number of vertices $n,$ and two weight vectors $\bm w, \bm w'$ that
are sufficiently far apart. Then, we show that any algorithm that achieves the desired approximation
ratio must output solutions that, with high probability, differ in a large number of vertices under $\bm w, \bm w'$. Lastly, we show this construction implies
a lower bound on the Lipschitz constant of this algorithm by applying a compactness argument on a path from $\bm w$ to $\bm w'$ (see \Cref{thm:finite perturbation}). We believe this compactness argument could find further applications in the design
of lower bounds for Lipschitz continuous algorithms.
We refer the reader to \Cref{subsec:s-t-cut-lower-bound} for the formal construction.

\subsection{Bipartite \texorpdfstring{$\bm b$}{b}-Matching \texorpdfstring{(\Cref{sec:b-matching})}{Section}}\label{sec:overview:bipartite b-matching}
Given a graph $G= (V, E)$ (which may be bipartite) and a weight vector 
$\bm w \in \mathbb{R}_{> 0}^E$, the maximum weight matching problem asks for a subset of edges $M \subseteq E$ 
which forms a matching in $G$ and maximizes the sum of the weights of the edges in $M$. 
The maximum weight $\bm b$-matching problem is defined similarly to the 
{maximum weight matching problem},
except 
each vertex $v$ is associated with a vertex capacity $b_v$ indicating the maximum number of neighbors that can be matched to $v$ in the $\bm b$-matching.
Note that a matching is a $\bm b$-matching for $\bm b \equiv 1$.

It is illustrative to first study the simpler case of $1$-matchings before tackling general $\bm b$-matchings

\subsubsection{Warmup for Bipartite Matching (\Cref{sec:bipartite})}\label{sec:overview:bipartite matching}
We design an algorithm that outputs a set of edges that gives a $2(1+\eps)$-approximate maximum weight bipartite matching with pointwise Lipschitz constant $O\left(\frac{\sqrt{m}}{\eps w_{\min}}\right)$. 
First,
we employ the classic LP for maximum bipartite matching regularized with a {weighted} $\ell_2$-regularizer {$\frac\eps2 \sum_{uv\in E} w_{uv} x_{uv}^2$},
{which is $w_{\min}$-strongly convex}. 
The optimal fractional solution is a $(1+\eps)$-approximation
and its Lipschitz constant of $O\left(\frac{\sqrt{m}}{\eps w_{\min}}\right)$
can be directly bounded using PGTA.
Then,
we use the stable rounding procedure from \cite{kumabe2022lipschitz} to obtain an integral solution.
{
Roughly speaking,
their rounding algorithm treats one side of the bipartition as buyers $u$ who bid on their neighboring sellers (vertices) $v$ with probability proportional to the value of the optimal fractional solution $x_{uv}^\star$.
Then,
each seller $v$ accepts a bid uniformly at random.
}
We note that the rounding scheme preserves the Lipschitz constant of the fractional solution up to constants
but incurs a factor of $2$ in the approximation ratio.
As mentioned in \Cref{sec:results},
our algorithm improves upon the pointwise Lipschitz algorithm of \cite{kumabe2022lipschitz} in certain input regimes.

{Before detailing how we extend these techniques to handle $\bm b$-matchings,
we briefly sketch previous approaches and their limitations in \Cref{sec:mwm prev work limits}.}

\subsubsection{Previous Work \& Limitations}\label{sec:mwm prev work limits}
{We first illustrate how the stability analysis of the unweighted entropy regularizer applied to bipartite matching from \cite{kumabe2022lipschitz}
does not straightforwardly extend to $\bm b$-matchings.
On the other hand,
PGTA allows us to employ the \emph{weighted} $\ell_2$-regularizer that extends directly to the $\bm b$-matching setting
without further analysis.

Let $G = (V, E)$ be a bipartite graph with bipartition $V = U\cup R$.
Rather than the weighted $\ell_2$-regularizer,
\cite{kumabe2022lipschitz} uses the unweighted entropy regularizer
\[
    h(x) 
    \coloneqq \sum_{uv\in E} x_{uv} \log x_{uv}
    = \sum_{u\in U} \sum_{v: uv\in E} x_{uv}\log x_{uv}.
\]
Then, they argue about the stability of the standard matching LP under this regularizer.
One crucial aspect of their analysis is that for any $u\in U$,
$\sum_{v: uv\in E} x_{uv} \leq 1$ so that we can think of a feasible solution $x$
as $\card U$ discrete probability distributions,
each over the edges incident to some $u\in U$.
The entropy regularizer has a small range over the simplex.

In the standard $\bm b$-matching LP however,
this property no longer holds and we can only guarantee that $\sum_{v: uv\in E} x_{uv} \leq b_u$.
Thus, the analysis of \cite{kumabe2022lipschitz} does not straightforwardly extend to $\bm b$-matchings.
On the other hand,
PGTA allows us to extend our analysis of the stability of the weighted $\ell_2$-regularizer from the matching polytope to the $\bm b$-matching polytope without further work.
This is possible as unlike the unweighted entropy regularizer,
the strong convexity property of the weighted $\ell_2$-regularizer does not depend on the matching-specific constraints of the matching polytope.

Another more subtle detail of the Lipschitz continuous bipartite matching algorithm of \cite{kumabe2022lipschitz}
is the need to smoothly estimate the value of $\OPT$
in order to correctly set the regularization parameter.
PGTA allows us to bypass this extra step since the weighted $\ell_2$-regularizer automatically adapts to the input weights.}

\subsubsection{Our Approach}\label{sec:mwbm our approach}
While PGTA applies directly to the weighted $\ell_2$-regularizer and $\bm b$-matching polytope,
we can no longer directly use the rounding scheme given in~\cite{kumabe2022lipschitz}
to obtain an integral solution.
{This is due to the fact that our $b_v$ constraints must be satisfied,
which the rounding scheme given in~\cite{kumabe2022lipschitz} does not necessarily
ensure.}
{Instead, we design a new multi-item auction-based rounding scheme where multiple vertices on the left (buyers) bid on multiple 
vertices on the right (sellers).}
This yields a $\frac{2(1+\eps)}{1-\nicefrac1e}$-approximate maximum weight bipartite $\bm b$-matching with pointwise Lipschitz constant $O\left(\frac{\sqrt{m}}{\eps w_{\min}}\right)$.
{Compared to our bipartite $1$-matching algorithm,
the Lipschitz constant is identical up to constants,
while the approximation ratio degrades by a factor of $(1-\nicefrac1e)$ due to additional complications in the rounding scheme we detail below.}

{
Let $x^\star$ denote the optimal fractional solution.
Loosely speaking,
we think of each seller $v$ now as having $b_v$ items to sell
and each buyer $u$ now selects $b_u$ sellers independently with probability $\frac{x_{uv}^\star}{b_u}$.
For each selected seller,
the buyer bids on one of the seller's items uniformly at random.
Each seller $v$ accumulates the bids for each of its $b_v$ items
and accepts a bid uniformly at random per item.
The main challenge compared to the matching case ($\bm b\equiv 1$) is that an edge $uv$ can be selected at most once in the $\bm b$-matching,
but the buyer $u$ may purchase multiple items from the seller $v$
in the described multi-item auction scheme.
In other words,
the probability of the edge $uv$ being selected in the rounding scheme
is equal to the probability of \emph{at least} one transaction occurring between $u, v$
but does not scale with the \emph{total} number of transactions.
This is the cause of the additional $(1-\nicefrac1e)$ factor in the approximation ratio compared to our bipartite matching algorithm.}

\subsection{Packing Integer Programs (\Cref{sec:packing IP})}\label{sec:overview:packing IP}
A \emph{packing integer program (PIP)} is specified by a matrix of constraints $A \in [0, 1]^{p \times m}$, 
weight vector $\bm w \in \mathbb{R}^{m}_{\geq 0}$, and $\bm b \in \mathbb{R}^p_{\geq 1}$.
The goal is to output an integral vector $y\in \set{0, 1}^m$
maximizing $\bm w^\top \bm y$ subject to $A\bm y\leq \bm b$.
Note that when $A$ is the edge-vertex incidence matrix and $\bm b$ is integral,
this is exactly the $\bm b$-matching problem.
For this problem,
we obtain a $O\left(p^{1/B}\right)$-approximation algorithm with pointwise
Lipschitz constant $O\left(\frac{\sqrt{m}}{w_{\min}p^{1/B}}\right)$ 
where $p$ is the number of constraints and $B \coloneqq \min_j b_j$, so $B \geq 1$.

Similar to our $\bm b$-matching algorithm,
we regularize the standard LP relaxation by adding a weighted
$\ell_2$ regularizer.
Then, we again apply PGTA to bound the Lipschitz constant of the fractional solution.
A stable integral solution can then be obtained by independently rounding each coordinate after down-scaling the fractional solution.

\subsection{Shared Randomness (\texorpdfstring{\Cref{sec:shared-randomness}}{Section})}\label{sec:overview:shared randomness}
Executing a pointwise Lipschitz continuous algorithm independently on two weight instances $\bm w, \tilde{\bm w}$ can still lead to very different results.
Hence, it is practical to consider the case where we require the coupling in the calculation of the earth mover's distance to be the one induced by using the same randomness, which we call the \emph{pointwise Lipschitz constant under shared randomness}.
Our algorithms can be straightforwardly made pointwise Lipschitz continuous even in the shared randomness setting. 
See \Cref{sec:shared-randomness} for details.

\section{Related Work}\label{sec:related-work}
\paragraph{{Weighted} Lipschitz Continuity vs. Pointwise Lipschitz Continuity.}
The definition of Lipschitz continuous
graph algorithms has been given in two settings.
{In the first setting, Lipschitz constants are measured by considering the \emph{weighted} $\ell_1$ 
distance of the characteristic vectors of the outputs under perturbations of the input.
In the second setting, pointwise Lipschitz constants are measured by considering the \emph{unweighted} $\ell_1$ distance of the characteristic vectors of the outputs.}
To make the distinction between the two concepts clear, we will {generally} refer to the former as the {\emph{weighted}} Lipschitz constant
and the latter as the {\emph{unweighted}} or \emph{pointwise} Lipschitz constant.

While weighted Lipschitz continuous graph algorithms with bounded weighted Lipschitz constants have been studied to a greater extent \cite{kumabe2022lipschitz,kumabe2023lipschitz,gima2025courcelle}, pointwise Lipschitz continuous
graph algorithms have largely been unexplored.\footnote{{In general,
there is no translation between the weighted and pointwise Lipschitz constant of an algorithm
due to their qualitative differences. 
See \Cref{sec:def-remarks} for more details.}}
{Indeed,
\cite{kumabe2022lipschitz} developed combinatorial algorithms for minimum spanning tree, shortest path, and maximum weight matching.
\cite{kumabe2023lipschitz} design combinatorial algorithms for vertex cover, set cover, and vertex feedback set
as well as a regularized LP-based algorithm for set cover,
all of which are in the weighted Lipschitz continuous setting.
Lastly,
\cite{gima2025courcelle} studied the weighted Lipschitz definition for the special case of bounded-treewidth graphs.}
On the other hand,
previous work \cite{kumabe2022lipschitz}
studied the pointwise Lipschitz constant for only two problems:\footnote{Although another work \cite{kumabe2024allocations} also studied the unweighted Lipschitz definition,
they do so in the context of fair allocations in cooperative game theory,
a different setting from what we consider.}
the minimum spanning tree problem and the maximum weight bipartite matching problem,
{for which they developed a combinatorial and a regularized LP-based method, respectively}.
{See \Cref{sec:overview:pgta} for more detailed comparison with previous regularized LP-based techniques.}

Although measuring the change in output using weighted distances may be appropriate when the input and output spaces are the same,
\eg{} in the case of weighted edge inputs and edge outputs,
many problems can have different input and output spaces.
One example is the minimum vertex $s$-$t$ cut problem,
where we have weighted edge inputs but vertex subset outputs.
Hence, it {does not} make sense to consider the weighted metric in these scenarios.
Moreover, in many applications, we want to
bound the unweighted distance of the outputs of an algorithm, as it determines the cost of updating the solution.
Thus, designing algorithms with small pointwise Lipschitz constants 
{can} be more desirable than bounding the weighted distance of their outputs.

\paragraph{Online/Dynamic Algorithms \& Recourse.}
Low recourse online/dynamic algorithms (see \eg{}, \cite{arar2018dynamic,henzinger2025contraction,megow2025online,gupta2020submodular,bhattacharya2022efficient,davis2006online,gupta2017online,lattanzi2017consistent} 
and references therein) are algorithms where updates to the graph are given in an online manner,
and the goal is to minimize the 
\emph{recourse} or the change in the output between consecutive days over $T$ total days. 
{For min $s$-$t$ cut,
there are dynamic cut sparsifiers with bounded recourse~\cite{goranci2021expander,chen2020fast}, which may lead to $s$-$t$ cut algorithms with small recourse. 
However, we are not aware of any such explicit recourse bounds.}
For matchings on the other hand,
there are black-box transformations of dynamic matching algorithms that obtain $O(1)$ recourse
while incurring an $O(1)$ increase in update time and approximation ratio~\cite{solomon2021reconfiguration,bernstein2025matching}.
{Interestingly, improving recourse has been shown to improve update time for various problems in the dynamic setting~\cite{goranci2021expander,roghani2022beating}.}

{It is known that}
re-running any $\alpha$-approximate (pointwise) Lipschitz algorithm with (pointwise) Lipschitz constant $K$ each day from scratch
immediately translates to an online $\alpha$-approximation algorithm where the total expected recourse is bounded by $KT$ (over all days)~\cite[Theorem 1.6]{kumabe2022lipschitz}.
{Thus, all of the 
the results we present in this paper also translate to the low-recourse online setting.}
{In fact,
we achieve a stronger guarantee:
the total expected recourse depends only on the \emph{net} change in edges
(note that this is at most $T$).
For example,
if an arbitrary sequence of $n$ edge updates yields a net change in only $1$ edge,
then the total expected recourse is $K$ rather than $Kn$.
It is not surprising that dynamic algorithms can in general achieve smaller recourse bounds than 
{re-running a pointwise Lipschitz algorithm each day}
as the latter satisfies this stronger net-change recourse property.}

{Related to the translation of \cite{kumabe2022lipschitz} is the lazy or deferred update framework
that has been widely applied towards dynamic matching algorithms (see \eg{} \cite{gupta2013fully}).
Among this line of work,
perhaps the most similar to ours is that of
\cite{jambulapati2022decremental,chen2025decremental},
who leverage a regularized convex relaxation to design dynamic matching algorithms that support edge deletions.
While the optimal fractional solutions of these programs are likely stable,
it is unclear if the intricate approximate solvers combined with the technically involved rounding algorithms 
have non-trivial Lipschitz constant.
}

\paragraph{Differential Privacy.} 
{Although Lipschitz constants are used in differential privacy (DP), for example, Lipschitz extensions, 
both the definition of the constants and the analyses we perform in our paper are
fundamentally distinct. 
Indeed, any DP algorithm for 
the problems we study have large known lower bounds in additive error, up to 
$\Omega(n)$, even in expectation. 
Specifically, \cite{dalirrooyfard2024nearly} proposed {a DP} algorithm for
the minimum $s$-$t$ cut problem with an additive error of {$O(n)$},
along with an (almost) matching lower bound.
For matchings,
a recent lower bound 
shows that any \emph{explicit}\footnote{{An explicit algorithm outputs edge subsets, whereas an implicit algorithm may output a different solution concept.}}
DP matching algorithm
requires $\Omega(n)$ additive error, even when restricting to bipartite instances
and allowing constant multiplicative error~\cite{dinitz2025differentially}.
In contrast, we obtain an additive error of $O(\lambda_2)$ for min $s$-$t$ cut (\cref{tab:cut-results}) and purely multiplicative error (\emph{no additive error})
for matchings (\cref{tab:matching-results}).

Lipschitz extensions have been used before in the context of DP graph algorithms (see \eg{},~\cite{kasiviswanathan2013analyzing,raskhodnikova2016lipschitz}).
However, these prior definitions rely on ``worst-case" Lipschitz constants.
That is, the Lipschitz constant is defined as the maximum ratio between the distance of two outputs and the distance of their corresponding inputs, considering all possible input pairs in the function's domain.\footnote{In differential privacy literature, the Lipschitz constant is sometimes referred to as the 
\emph{global sensitivity} of the function.}
}
Such a worst-case guarantee
is often necessary in order to ensure differential privacy since DP mechanisms often require 
worst-case sensitivity.\footnote{Sometimes \emph{coupled sensitivity}~\cite{BGM22} is used, where the sensitivity is the worst-case distance under a coupling.} As such, the notion of 
the pointwise Lipschitz constant we consider, which is defined in terms of the expectation, does not immediately translate to bounds for differential privacy. 

On the other hand,
differential privacy for graph algorithms 
{can be understood as bounding the R\'enyi divergence~\cite{mironov2017renyi} of the output distributions}
under one edge/vertex/unit weight change.
{Hence, it is unclear if there is a generic transformation from privacy (R\'enyi divergence) to Lipschitzness (earth mover's distance)
that does not introduce large errors.}

\paragraph{Perturbation Resilience.} 
Perturbation resilience~\cite{awasthi2012center,angelidakis2017algorithms,balcan2016clustering,bilu2012stable}
is defined (informally) as the property of a problem instance
where the (unique) optimal solution to the problem does not change when the specific instance is slightly perturbed.
Thus, perturbation-resilient input instances are restricted classes of the input space that allow for this property. 
It has been shown that many natural problems, such as various clustering and cut problems~\cite{awasthi2012center,balcan2020k,balcan2016clustering,BDLS13,bilu2012stable,makarychev2014bilu},
are easier to solve when the given instances are perturbation resilient while they are NP-hard to solve in general instances. 

{By definition,
the algorithm that computes the optimal solution on a perturbation resilient instance has a pointwise Lipschitz constant of 0.
However, it is not straightforward to efficiently verify that an instance is perturbation resilient.
We design pointwise Lipschitz continuous algorithms with guarantees over all input instances,
where the pointwise Lipschitz constant may depend on some natural,
efficiently computable properties of the input,
\eg{}, the second smallest eigenvalue of the graph Laplacian.
We emphasize that perturbation resilience is a property of the input instance
whereas pointwise Lipschitz continuity is a property of the algorithm itself.}

\paragraph{Differentiable Algorithms.} 
{The differentiable algorithms literature seeks to incorporate structured optimization problems as a differentiable operator
within deep neural networks~\cite{amos2019differentiable,djolonga2017differentiable,blondel2020fast,berthet2020learning,donti2021learning,agrawal2019differentiable}.
This usually necessitates that the output is fractional and deterministic,
unlike our setting where we also randomly round fractional solutions to discrete solutions.
Although we also use continuous relaxations of discrete problems as an intermediate technique,
our goal is to develop algorithms with approximation and stability guarantees rather than combining them within neural networks.}

\paragraph{Average Sensitivity.} 
\textcite{varma2021average} introduced the notion of {\emph{average}} 
sensitivity, which can be seen as an analogue of pointwise Lipschitz continuity for the unweighted setting.
Namely, the \emph{sensitivity} of a randomized algorithm $\mcal A$ that outputs a set $S$ given the input graph $G=(V,E)$ is defined to be the maximum earth mover's distance between $\mcal A(G)$ and $\mcal A(G-e)$ over $e \in E$, where $G-e$ is the graph obtained from $G$ by deleting $e$, and $\mcal A(G)$ and $\mcal A(G-e)$ represent the output distributions of $\mcal A$ on $G$ and $G-e$, respectively. In that definition, the Hamming distance is used as the underlying metric for calculating the earth mover's distance.
Since its introduction,
{the average} sensitivity of algorithms for various problems have been studied~\cite{li2025average,Peng2020,kumabe_et_al:LIPIcs.ESA.2022.75,Kumabe22,Yoshida2021}.

{All of our algorithms translate to the average sensitivity setting.}
Indeed,
the topology of an unweighted graph can be described by a $\{0,1\}$-valued vector on a complete graph, 
so we can bound the {average} sensitivity of an algorithm by taking the supremum of the pointwise Lipschitz constant over weight vectors between the two $\{0,1\}$-valued vectors corresponding to $G$ and $G-e$ (see \Cref{sec:finite bounds}).
It was shown in~\cite{varma2021average} that there is a $(1,O(n^{2/3}))$-approximation algorithm for the minimum vertex $s$-$t$ cut problem with sensitivity $O(n^{2/3})$.
The algorithm obtained from our pointwise Lipschitz continuous algorithm (Theorem~\ref{thm:s-t-min-cut-additive}) has additive error $O(\lambda_2)$ and {average}
sensitivity $O(\lambda_2^{-1}n)$, which outperforms the previous algorithm when $\lambda_2 \in (n^{1/3},n^{2/3})$.

For the maximum matching problem, although a polynomial-time $2$-approximation algorithm with {average} 
sensitivity $O(1)$ is known~\cite{varma2021average}, the algorithm heavily relies on the fact that the graph is unweighted and it is unclear whether we can use it to design pointwise Lipschitz continuous algorithms.

\paragraph{Replicability.} 
The notion of pointwise Lipschitz continuous algorithms
is related to the notion of replicable learning algorithms that operate on data
drawn from some distribution, 
as introduced by \textcite{impagliazzo2022reproducibility}.
This definition requires that when an algorithm $\mcal A$ is executed on two \iid{} $n$-sample datasets from some distribution $\mathcal{D}$,
under \emph{shared internal randomness}, the outputs of the algorithm
are exactly the same, with high probability.
We can interpret replicable learning algorithms as being stable under resampling of the dataset,
where stability is measured in terms of total variation distance.
Since inception,
replicability has been studied in a variety of settings~\cite{singh2025sensitivity,esfandiari2023replicable,bun2023stability,esfandiari2023replicableb,
kalavasis2023statistical,chase2023replicability,EatonHKS23,karbasi2023replicability,dixon2023list,chase2023local,komiyama2024replicability,kalavasis2024replicable,hopkins2024replicability,kalavasis2024computational}.

On the other hand,
the stability of pointwise Lipschitz algorithms is
measured in terms of the earth mover's distance.

\section{Preliminaries}

In this section, we define the notations and definitions we use in our paper. 

\subsection{Notation}
We first define the notation we use throughout our work.
\begin{itemize}[noitemsep]
    \item ${\bm w}^{(j)}, {\bm w}^j$ for a vector indexed by $j$.

    \item $w_i$ for the $i$-th coordinate of $\bm w$.

    \item $\tilde{\bm w}$ for a vector obtained by perturbing $\bm w$ by some perturbation vector {$\bm \delta$} with $\ell_1$-norm $\sum_i \abs{\delta_i} > 0$.

    \item $w_{\min}, w_{\max}$ for the min/max element of vector $\bm w$.

    \item $\mcal A$ for a (randomized) algorithm 
    and $\mcal A_\pi$ for the deterministic algorithm obtained from $\mcal A$ by fixing its internal randomness $\pi$.

    \item $A \Delta B \coloneqq (A\cup B)\setminus (A\cap B)$ for the symmetric difference.

    \item $\ones_S$ for the indicator variable of set $S$. 
\end{itemize}
Moreover, we say that $\mcal A$ gives an $(\zeta,\alpha)$-approximation guarantee to a minimization problem if its output has value at most $\zeta\cdot \OPT + \alpha$,
or for a maximization problem,
at least $\nicefrac1\zeta\cdot \OPT - \alpha$. 
When there is no additive error $\alpha$, we simply say that $\mcal A$ gives an $\zeta$-approximation guarantee.

\subsection{Graph-Theoretic Notions}
Let $G=(V,E)$ be a graph and $\bm w \in \mathbb{R}_{\geq 0}^E$ be a vector that represents the edge weights.
For a vertex $v \in V$, 
let $d_{G,\bm w}(v) \coloneqq \sum_{e \in E: v \in e}w_e$ be the \emph{weighted degree} of $v$ and $d_G^u(v) \coloneqq \card{\set{w\in V: uw\in E}}$ be the \emph{unweighted degree} of $v$.
For an edge set $F \subseteq E$, let $\bm w(F) \coloneqq \sum_{e \in F}w_e$.
The cut weight of $S \subseteq V$ is defined to be $\mathrm{cut}_{G,\bm w}(S) \coloneqq \bm w(\partial_G S)$,
where $\partial_G S \coloneqq E(S, V\setminus S)$ 
(\ie{}, the edges going between $S$ and $V \setminus S$).
We may drop subscripts if they are clear from the context.

\subsection{Pointwise Lipschitz Continuity}\label{sec:pointwise-lipschitz}

We write $\mcal A(G, \bm w)$ to denote the (random) output of a randomized algorithm $\mcal A$ when it is given as input a graph $G$ and a weight vector $\bm w \in \mathbb{R}_{\geq 0}^E$.
{Roughly speaking,
we say that $\mcal A$ has pointwise Lipschitz constant at most $c_{\bm w}$ if
\[
    \frac{D(\mcal A(G, \bm w), \mcal A(G, \tilde{\bm w}))}{\norm{\bm w-\tilde{\bm w}}_1}
    \leq c_{\bm w}
\]
for all $\tilde{\bm w}$ within a neighborhood of $\bm w$.
Here,
$D$ is some metric over output distributions.
The formal definition (\Cref{def:pointwise lipschitz}) follows from an appropriate choice of $D$.
}

Let $d: 2^V\times 2^V\to \R_+$ denote a metric on the output space of $\mcal A$.
For example,
if $\mcal A$ outputs subsets of vertices,
then one choice is to take $d$ to be the $\ell_1$-metric
\[
    d(V_1, V_2)
    = \card{V_1\Delta V_2}
    = \norm{\ones_{V_1} - \ones_{V_2}}_1\,.
\]
Unless otherwise stated, we take $d$ to be the $\ell_1$ metric above.
We write $(\pi, \tilde \pi)\sim \mcal D$ to denote a coupling\footnote{Recall $\mcal D$ is a joint distribution whose marginals on the first and second coordinates are $\pi, \tilde\pi$, respectively.}
between the random variables $\pi, \tilde \pi$.
Recall the \emph{earth mover's distance}\footnote{This metric is also known as the \emph{Wasserstein 1-distance} \cite{villani2009optimal},
\emph{Kantorovich-Rubinstein metric} \cite{kantorovich1960mathematical},
or \emph{Mallows' 1-distance} \cite{mallows1972note,levina2001earth}.}
\cite{rubner1998metric} on the output distributions of $\mcal A(G, \bm w)$ is defined to be
\[
    \EMD_d(\mcal A(G, \bm w), \mcal A(G, \tilde{\bm w}))
    \coloneqq \inf_{\mcal D} \E_{(\pi, \tilde \pi)\sim \mcal D} d\left( \mcal A_\pi(G, \bm w), \mcal A_{\tilde \pi}(G, \tilde{\bm w}) \right)\,.
\]
{We are now ready to define pointwise Lipschitz continuity formally.}

\begin{definition}[\cite{kumabe2022lipschitz}]\label{def:pointwise lipschitz}
    Suppose the (randomized) algorithm $\mcal A(G, \bm w)$ outputs subsets of some universe $\mathfrak S$.
    Let $d: 2^{\mfrak S}\times 2^{\mfrak S}\to \R_+$ be a metric on $2^{\mfrak S}$.
    We say that $\mcal A$ has \emph{pointwise Lipschitz constant $c_{\bm w}$ with respect to $d$} if
    \[
        \limsup_{\tilde{\bm w}\to \bm w}  \frac{\EMD_d(\mcal A(G, \bm w), \mcal A(G, \tilde{\bm w}))}{\norm{\bm w-\tilde{\bm w}}_1}
        \leq c_{\bm w}.
    \]
    Note here that we identify subsets $S\sset \mfrak S$ with their indicator variables $\ones_S\in \set{0, 1}^{\mfrak S}$.
\end{definition}

\subsubsection{Remarks on \texorpdfstring{\Cref{def:pointwise lipschitz}}{Definition}}\label{sec:def-remarks}
{As mentioned in \Cref{sec:related-work},
\textcite{kumabe2022lipschitz} defined two notions of Lipschitz continuity:
weighted Lipschitz continuity and unweighted pointwise Lipschitz continuity.
They mainly studied the weighted definition while we focus exclusively on the unweighted definition.
Thus from hereonforth,
\textbf{we write Lipschitz continuity/constant to mean pointwise Lipschitz continuity/constant} for simplicity.}

{Some remarks are due about the definition of Lipschitz continuity.
First, we note that an algorithm's Lipschitz constant cannot be scale-invariant,
\ie{}, it must depend on the input weight vector.
To see this, let $w,w'\in \mathbb{R}_{\geq 0}^E$ be arbitrary weight vectors
and suppose that for any constant $c>0$, we have
\[
    \frac{d(\mathcal{A}(G,w/c),\mathcal{A}(G,w'/c))}{
    \|w/c-w'/c\|_1} = \frac{c\cdot d(\mathcal{A}(G,w),\mathcal{A}(G,w))}{
    \|w-w'\|_1}.
\]
However,
this implies that the Lipschitz constant is unbounded!
This explains why the Lipschitz constants of our algorithms depend on the input weights (\eg{}, $\lambda_2$).}

{From the convex analysis perspective,
one might attempt to measure the change in solution using the dual $\ell_\infty$ norm under $\ell_1$-perturbations
rather than the $\ell_1$-norm.
Under this formulation,
the numerator of the ``Lipschitz constant'' obtained would be equivalent to the total variation distance between output distributions under perturbations.
It may be feasible to design algorithms under this notion using techniques from the differential privacy literature.
However,
as mentioned in \Cref{sec:related-work},
measuring the change in solution using the $\ell_1$ norm under $\ell_1$-perturbations
is more natural for combinatorial problems since this usually determines the cost of changing a solution.
Under this definition,
as witnessed by our tight upper and lower bounds for minimum vertex $S$-$T$ cut,
{our framework can obtain optimum bounds.}

{At first glance,
it may seem difficult to compute the Lipschitz constant {for perturbations of more than one entry,
as necessitated by \Cref{def:pointwise lipschitz}}.
\textcite{kumabe2022lipschitz} showed the following lemma,
which essentially says that it suffices to analyze the Lipschitz constant
under the perturbation of a single entry of the weight vector.}
\begin{lemma}[See also Lemma 7.1 of~\cite{kumabe2022lipschitz}]\label{lem:one perturbation}
    Let $\delta > 0$
    and $\tilde{\bm w}$ be obtained from $\bm w$ by choosing $\delta_{e_i}\in \R$ for each $i\in [m]$
    such that $\sum_{i\in [m]} \abs{\delta_{e_i}} = \delta$
    and setting $\tilde{w}_{e_i} \coloneqq w_{e_i} + \delta_{e_i}$.
    Let $\bm w^{(i)} \coloneqq \bm w + \sum_{j=1}^{i}\delta_{e_j}\ones_{e_j}$
    so that $\bm w^{(0)} = \bm w$ and $\bm w^{(m)} = \tilde{\bm w}$.
    Suppose for every $i\in [m]$,
    \[
        \EMD_d\left( \mcal A( G, \bm w^{(i-1)} ), \mcal A( G, \bm w^{(i)} ) \right)
        \leq c_{\bm w} \abs{\delta_{e_i}},
    \]
    where $c_{\bm w}$ is a function of only $\bm w$.
    Then
    \[
        \frac{\EMD_d(\mcal A(G, \bm w), \mcal A(G, \tilde{\bm w}))}{\norm{\bm w-\tilde{\bm w}}_1}
        \leq \sum_{i\in [m]} \frac{c_{\bm w} \abs{\delta_{e_i}}}{\norm{\bm w-\tilde{\bm w}}_1}
        \leq c_{\bm w}.
    \]
\end{lemma}

As the intermediary step of our algorithms,
we consider algorithms with \emph{fractional} pointwise Lipschitz constant.
That is,
$\mcal A(G, \bm w)$ outputs a (deterministic) vector in $\R^{\mfrak S}$
such that the following quotient is bounded in the limit:
\[
    \frac{d(\mcal A(G, \bm w), \mcal A(G, \tilde{\bm w}))}{\norm{\bm w-\tilde{\bm w}}_1} \,.
\]
Here, $d: \R^{\mfrak S}\times \R^{\mfrak S}\to \R_+$ is a metric over the vector space $\R^{\mfrak S}$.

We can interpret such an output vector as a point in the convex hull of indicator vectors of $2^{\mfrak S}$.
Note that in the case of fractional Lipschitz algorithms,
it suffices to consider deterministic algorithms.

\section{Organization of the Rest of the Paper}
We organize the rest of the paper as follows.
\begin{itemize}
    \item \Cref{sec:framework} develops the PGTA framework,
    our main tool for analyzing the Lipschitz constant of regularized convex programs.
    \item \Cref{sec:min-cut-through-pgm} leverages PGTA along with two novel rounding techniques to design Lipschitz algorithms for minimum $S$-$T$ cut.
    \item \Cref{sec:b-matching} presents our Lipschitz algorithms for bipartite matching and $\bm b$-matching,
    both of which use PGTA,
    whereas the latter also requires a brand new rounding scheme.
    \item \Cref{sec:packing IP} discusses another application of PGTA alongside simple independent rounding towards designing a Lipschitz algorithm for packing integer programs.
    \item \Cref{sec:shared-randomness} explains how to extend our algorithms to the more practical shared-randomness setting.
\end{itemize}

\section{Proximal Gradient Trajectory Analysis}\label{sec:framework}

In this section, we consider a convex program
\[
    \min_{\bm x\in K} f_{\bm w}(\bm x),
\]
where $f_{\bm w}:\R^n \to \R$ and $K\sset \R^n$ are both convex
and $f_{\bm w}$ is parameterized by a vector $\bm w\in \R^m$. 
We make no further assumption on $f_{\bm w}(\bm x)$ other than it being real-valued and convex\footnote{We can alternatively assume that $f_{\bm w}$ is lower semi-continuous, proper, convex, and takes value in the extended reals.}.

Our goal is to establish a fractional pointwise Lipschitz-continuous algorithm for the problem.
We can think of $f_{\bm w}$ as the objective function for some convex relaxation of a combinatorial optimization problem such as the minimum $s$-$t$ cut problem, and $\bm w$ represents the weights on edges.

Note that, without any additional assumptions on the objective function $f_{\bm w}$, even small perturbations in the weights $\bm w$ can lead to drastically different optimal solutions.
In order to address this,
we introduce a smooth, strongly convex regularizer $g_{\bm w}(\bm x)$ and instead solve the program
\[
    \min_{\bm x\in K}h_{\bm w}(\bm x),
    \qquad \text{where $h_{\bm w}(\bm x) \coloneqq f_{\bm w}(\bm x) + g_{\bm w}(\bm x)$}.
\]
Here for $L,\sigma > 0$, we say that a differentiable function $g:K \to \mathbb{R}$ is \emph{$\sigma$-strongly convex} and \emph{$L$-smooth} if for any $x,y\in K$,
\begin{align*}
    g(y) & \geq g(x) + \langle \grad g(x), y-x\rangle + \frac{\sigma}{2}\|y-x\|^2, \\
    g(y) & \leq g(x) + \langle \grad g(x), y-x\rangle + \frac{L}{2}\|y-x\|^2
\end{align*}
hold, respectively.

Let $\bm w, \tilde{\bm w}$ be the original and perturbed weight parameters
and let $\bm x^*, \tilde{\bm x}^*$ be the unique optimal solutions with respect to $h_{\bm w}, h_{\tilde{\bm w}}$, respectively.
In order to analyze $\norm{\bm x^* - \tilde{\bm x}^*}_2$,
we can instead analyze a sequence of iterates $(\bm x^k)_{k\geq 0}, (\tilde{\bm x}^k)_{k\geq 0}$
such that $\bm x^k \to \bm x^*, \tilde{\bm x}^k\to \tilde{\bm x}^*$
and analyze
\[
    \lim_{k\to \infty} \norm{\bm x^k - \tilde{\bm x}^k}_2.
\]

To generate $(\bm x^k)$ and $(\tilde{\bm x}^k)$, we use the \emph{proximal gradient method} (PGM).
We emphasize that \emph{PGM is only used in the stability analysis}
and the algorithm designer is free to use their favorite optimization algorithm,
provided the produced solution is sufficiently close to optimality.
First, we recall the \emph{proximal operator} that
is used to minimize a composite function $h(\bm x) \coloneqq f(\bm x) + g(\bm x)$, which
is denoted by
$\prox_{L, f}:\R^n \to K$ and defined as
\[
    \prox_{L, f}(\bm x) \coloneqq \argmin_{\bm z\in K} \frac{L}2 \norm{\bm z-\bm x}_2^2 + f(\bm z).
\]
Then, given an initial point $\bm x^0 \in \R^n$, the PGM with step size $\nicefrac1L$ iteratively applies the following update rule:
\begin{align*}
    \bm x^{t+1}
    = \prox_{L, f}\left(\bm x^t - \frac1L \grad g(\bm x^t) \right).
\end{align*}
It is known that PGM converges at an exponential rate even if we only compute an inexact proximal step \cite{schmidt2011convergence}. 
Moreover,
it is known that the proximal operator is \emph{strictly non-expansive} \cite{rockafellar1976monotone}, 
meaning that
\[
    \norm{\prox_{L, f}(\bm y) - \prox_{L, f}(\bm z)}_2^2
    \leq \iprod{\prox_{L, f}(\bm y) - \prox_{L, f}(\bm z), \bm y-\bm z}\,.
\]
By an application of the Cauchy-Schwartz inequality,
this immediately implies that 
\begin{align}
    \norm{\prox_{L, f}(\bm y) - \prox_{L, f}(\bm z)}_2\leq \norm{\bm y-\bm z}_2. \label[inequality]{eq:pgm-contraction}
\end{align}
We underline again
that any algorithm which solves the convex
program to some sufficiently small error
enjoys the stability properties
we prove next
and PGM is only used to establish the stability of the
unique optimal solution through an analysis
of how its trajectory differs on two distinct
executions under $\bm w, \tilde{\bm{w}}.$

To bound the fractional pointwise Lipschitz constant, we require the following assumptions:
\begin{assumption}\label{asm:PGM}
    Let $0\neq \delta\in \R$ be a real number
    with sufficiently small absolute value and suppose that we obtain $\tilde{\bm w}$ from $\bm w$ by perturbing a single coordinate of $\bm w$ by $\delta$. 
    \begin{enumerate}
        \item $f_{\bm w}, f_{\tilde{\bm w}}$ are convex and continuous over $K$ \label[condition]{item:convex};
        \item Assuming $\abs{\delta}\leq \delta_{\bm w}$ for some $\delta_{\bm w}$ depending only on $\bm w$,
        both $g_{\bm w}$ and $g_{\tilde{\bm w}}$ are $\sigma_{\bm w}$-strongly convex over $K$ for the same $\sigma_{\bm w}$
        that depends only on $\bm w$; \label[condition]{item:strong-convex}
        \item Assuming $\abs{\delta}\leq \delta_{\bm w}$ for some $\delta_{\bm w}$ depending only on $\bm w$, 
        both $g_{\bm w}$ and $g_{\tilde{\bm w}}$ are $L_{\bm w}$-smooth over $K$ for the same $L_{\bm w} > \sigma_{\bm w}$
        that depends only on $\bm w$;
        \item for every $\bm x\in K$, $\frac1{L_{\bm w}} \norm{\grad g_{\bm w}(\bm x) - \grad g_{\tilde{\bm w}}(\bm x)}_2 \leq C_{\bm w} \abs{\delta}$ for some constant $C_{\bm w}$ that depends only on $\bm w$; \label[condition]{item:c_1} 
        \item for every $\bm x\in K$, $\norm{\prox_{L_{\bm w}, f_{\bm w}}(\bm x) - \prox_{L_{\bm w}, f_{\tilde{\bm w}}}(\bm x)}_2 \leq D_{\bm w} \abs{\delta}$ 
        for some $D_{\bm w}$ that depends only on $\bm w$. \label[condition]{item:c_2}
    \end{enumerate}    
\end{assumption}

{Some remarks are in order regarding \Cref{item:c_1} and \Cref{item:c_2},
which may not come across as natural at first glance.
Consider a regularizer of the form $g_{\bm w}(\bm x) = \bm w^\top F(\bm x)$ where $F: K\to \R^m$ is sufficiently ``well-behaved''.
All of our regularizers fall under this form.
Then $\grad g(\bm x) = \grad F(\bm x)^\top w \in \R^n$ so that
\[
    \norm{\grad g_{\bm w}(\bm x) - \grad g_{\tilde{\bm w}}(\bm x)}_2
    = \norm*{\grad F(x)^\top (\bm w-\tilde{\bm w})}_2
    \leq \norm*{\grad F(\bm x)^\top}\cdot \norm{\bm w-\tilde{\bm w}}_2
    = \norm*{\grad F(\bm x)^\top}\cdot \delta,
\]
where $\norm{\grad F(\bm x)^\top}$ denotes the operator norm of the matrix $\grad F(\bm x)^\top$ with respect to the Euclidean norm.
When $K$ is compact and $\grad F$ is continuous,
we can always obtain a trivial bound to satisfy \Cref{item:c_1}.
However,
an important detail in our applications is that this bound can be improved significantly using problem-specific structure to yield better guarantees.

On the other hand,
$\prox_{L_{\bm w}, f_{\bm w}}(\bm x)$ is the optimal solution of a strongly convex objective,
meaning we can translate the closeness of the objective value to a proximity in the variable space.
If $f_{\bm w}, f_{\tilde{\bm w}}$ are sufficiently close in function value,
we would thus expect their optimal solutions to be close.
In our applications,
extracting a bound for \Cref{item:c_2} that depends only on $\bm w$ requires exploiting the specific structure of $f_{\bm w}$.
}

\begin{theorem}\label{thm:PGM}
    Suppose Assumption~\ref{asm:PGM} holds.
    Then the optimal solutions $\bm x^*, \tilde{\bm x}^*$ corresponding to $h_{\bm w}$ and $h_{\tilde{\bm w}}$ satisfy
    \[
        \norm{\bm x^* - \tilde{\bm x}^*}_2
        \leq \frac{L_{\bm w}(C_{\bm w}+D_{\bm w})}{\sigma_{\bm w}} \abs{\delta}.
    \]
\end{theorem}

To prove \Cref{thm:PGM}, we use the following technical lemma, 
whose proof is deferred to \Cref{apx:nonexpansive gradient}. 
Before we state it formally, let us define
the notion of $\eta$-expansiveness. 
\begin{definition}\label{def:eta-expansive}
    Let $\eta> 0$, 
    $K$ be some set endowed with some norm $\norm{\cdot}$,
    and $\mcal T: K \to K$.
    We say that $\mcal T$ is an $\eta$-expansive operator with respect to $\norm{\cdot}$
    if $\forall \bm x,\bm y \in K$ it holds that
    $\norm{\mcal T(\bm x) - \mcal T(\bm y)} \leq \eta \cdot \norm{\bm x - \bm y}$.
\end{definition}

\begin{restatable}{lemma}{nonexpansiveGradient}\label{lem:nonexpansive gradient}
    The gradient step of a $L$-smooth function $g$
    that is $\sigma$-strongly convex over its convex feasible region $K$
    with step size $\eta \leq \nicefrac1{L}$ is $(1-\eta\sigma)$-expansive
    with respect to the Euclidean norm.
\end{restatable}
For completeness, 
we include a proof of \Cref{lem:nonexpansive gradient} in \Cref{apx:nonexpansive gradient}.

{The proof of \Cref{thm:PGM} presented below takes inspiration from the notion of algorithmic stability introduced by \textcite{hardt2016train}. 
Specifically, the authors analyze the stochastic gradient descent trajectory of an arbitrary objective function $h$ that is $L$-smooth and $\sigma$-strongly convex
(which includes our composite objectives as a special case)
and derive an $\ell_2$ solution perturbation bound of roughly $O(\abs{\delta}\cdot \frac{B^2}{\sigma})$,
where $B$ is a bound on the maximum norm of $\grad h$.
Note that the bound depends on $B$ rather than $L$.
While \cite{hardt2016train} observed that the proximal operator is a ``stability-inducing'' operation,
they do not explicitly analyze composite objectives nor the PGM.
On the other hand,
analyzing the trajectory of PGM is a natural choice for composite objectives $h = f+g$ such as our own,
and we develop a more fine-grained bound of the form $O(\abs{\delta}\cdot \frac{L(C+D)}\sigma)$,
where $C, D$ bound deviations of $\grad g$ and $\prox_{L, f}$ (see \Cref{item:c_1,item:c_2}).
For our settings,
$\frac{B^2}\sigma$ can be $\Omega(n^2)$ while $\frac{L(C+D)}\sigma = o(n)$.
Thus, 
we develop stronger perturbation bounds for our particular setting by analyzing the trajectory of a more suitable process (PGM)
and measuring the perturbation using carefully chosen quantities of the objective (smoothness $L$, gradient/proximal deviations $C, D$).}

We now present the proof of \Cref{thm:PGM}.
\begin{proof}[Proof of \Cref{thm:PGM}]
    We assume that the neighborhood around $\bm w$ is small enough so that the conditions in \Cref{asm:PGM} hold.
    Let
    \begin{itemize}
        \item $\bm x^t, \tilde{\bm x}^t$ denote the adjacent runs of PGM,
        \item $L = L_{\bm w}$ the common smoothness parameter,
        \item $\sigma = \sigma_{\bm w}$ the common strong convexity parameter,
        \item and $G(\bm x^t) \coloneqq \bm x^t - \frac1L \grad g_{\bm w}(\bm x^t)$ the gradient update ($\tilde G(\tilde{\bm x}^t)$ similarly).
    \end{itemize}
    Consider a single iteration:
    \begin{align*}
        &\norm{\bm x^{t+1} - \tilde{\bm x}^{t+1}}_2 \\
        &= \norm{\prox_{L, f_{\bm w}}(G(\bm x^t)) - \prox_{L, f_{\tilde{\bm w}}}(\tilde G(\tilde{\bm x}^t))}_2 \\
        &\leq \norm{\prox_{L, f_{\bm w}}(G(\bm x^t)) - \prox_{L, f_{\bm w}}(\tilde G(\tilde{\bm x}^t))} _2 
        + \norm{\prox_{L, f_{\bm w}}(\tilde G(\tilde{\bm x}^t)) - \prox_{L, f_{\tilde{\bm w}}}(\tilde G(\tilde{\bm x}^t))}_2 \\
        &\leq \norm{G(\bm x^t) - \tilde G(\tilde{\bm x}^t)}_2
        + \norm{\prox_{L, f_{\bm w}}(\tilde G(\tilde{\bm x}^t)) 
        - \prox_{L, f_{\tilde{\bm w}}}(\tilde G(\tilde{\bm x}^t))}_2 \tag{by \Cref{eq:pgm-contraction}} \\
        &\leq \norm{G(\bm x^t) - G(\tilde{\bm x}^t)}_2 
        + \norm{\tilde G(\tilde{\bm x}^t) - G(\tilde{\bm x}^t)}_2 
        + \norm{\prox_{L, f_{\bm w}}(\tilde G(\tilde{\bm x}^t)) - \prox_{L, f_{\tilde{\bm w}}}(\tilde G(\tilde{\bm x}^t))}_2 \\
        &\leq \left( 1 - \frac{\sigma}L \right) \norm{\bm x^t - \tilde{\bm x}^t}_2
        + \frac1L \norm{\grad g_{\bm w}(\tilde{\bm x}^t) - \grad g_{\tilde{\bm w}}(\tilde{\bm x}^t)}_2
        + \norm{\prox_{L, f_{\bm w}}(\tilde G(\tilde{\bm x}^t)) - \prox_{L, f_{\tilde{\bm w}}}(\tilde G(\tilde{\bm x}^t))}_2 \tag{by \Cref{lem:nonexpansive gradient} }\\
        &\leq \left( 1 - \frac{\sigma}L \right) \norm{\bm x^t - \tilde{\bm x}^t}_2 + C_{\bm w} \abs{\delta} + D_{\bm w} \abs{\delta} \tag{by \Cref{item:c_1}, \Cref{item:c_2}} \\
        &= \left(1-\frac{\sigma}{L}\right)^{t+1} \norm{\bm x^0 - \tilde{\bm x}^0}_2 + \sum_{\tau=1}^{t+1} \left( 1 - \frac{\sigma}L \right)^{t+1-\tau} (C_{\bm w} + D_{\bm w}) \abs{\delta} \\
        & = \sum_{\tau=1}^{t+1} \left( 1 - \frac{\sigma}L \right)^{t+1-\tau} (C_{\bm w} + D_{\bm w}) \abs{\delta}.
    \end{align*}
    Taking the limit as $t\to \infty$ yields the following bound on $\ell_2$-distance between optimal solutions $\bm x^*, \tilde{\bm x}^*$ respectively.
    \[
        \norm{\bm x^* - \tilde{\bm x}^*}_2
        \leq \frac{L (C_{\bm w} + D_{\bm w})}{\sigma} \abs{\delta}.
        \qedhere
    \]
\end{proof}

\section{Minimum Vertex \texorpdfstring{$S$-$T$}{S-T} Cut}\label{sec:min-cut-through-pgm}
In this section, we provide two pointwise Lipschitz continuous algorithms for the minimum vertex $S$-$T$ cut problem. 

\begin{theorem}\label{thm:s-t-min-cut-additive}
    For any $\gamma > 0$,
    there is a randomized polynomial-time approximation algorithm for undirected minimum vertex $S$-$T$ cut
    that outputs a vertex subset $A\sset V$ with the following guarantees.
    \begin{enumerate}[(a)]
      \item $A$ is a $\left( 1+O\left( \frac1{\gamma\sqrt{n}} \right), O\left( \frac{\lambda_2\log(\nicefrac1\gamma)}{\gamma} \right) \right)$-approximate $S$-$T$ cut
        with probability at least $1-\gamma$.
      \item The algorithm has a pointwise Lipschitz constant of
        $
            O\left( \nicefrac{n}{\lambda_2} \right).
        $
    \end{enumerate}
    Here $\lambda_2$ denotes the algebraic connectivity of $(G,\bm w)$.
\end{theorem}

\begin{theorem}\label{thm:s-t-cut-vertex-set}
    For any $\beta\in (0, \nicefrac12)$ and $\gamma\in (0, 1)$, 
    there exists a randomized polynomial-time approximation algorithm for the minimum vertex $S$-$T$ cut problem
    that outputs a vertex subset $A\sset V$ with the following guarantees.
    \begin{enumerate}[(a)]
      \item $\E[\cut_{G,\bm w}(A)] = O(\beta^{-1}\OPT_\beta)$.
      \item $A$ induces an $S$-$T$ cut with probability $1-\gamma$.
      \item The algorithm has a pointwise Lipschitz constant of
        $
            O\left( \frac{ \sqrt{n} \log(\nicefrac1\gamma)}{\beta^2 \lambda_2} \right).
        $
    \end{enumerate}
    Here $\OPT_\beta$ is the minimum weight of a $\beta$-balanced cut 
    and $\lambda_2$ is the algebraic connectivity of the weighted graph $(G,\bm w)$.
\end{theorem}

We also provide a lower bound whichs shows that our Lipschitz-approximation tradeoff in \Cref{thm:s-t-min-cut-additive} is tight.
\begin{theorem}\label{thm:s-t-cut-additive-lower-bound}
    For any sufficiently large integer $n$ and any $\lambda < n/c$ for a sufficiently large constant $c>0$, there exists a weighted graph $(G,w)$ with algebraic connectivity at least $\lambda$ such that any algorithm that outputs a $(O(1),\lambda)$-approximate minimum $s$-$t$ cut on $(G,w)$ 
    with probability at least $0.99$
    has pointwise Lipschitz constant $\Omega(\nicefrac{n}\lambda)$.
\end{theorem}

In \Cref{subsec:s-t-cut-fractional},
we describe our approach to obtain fractional pointwise Lipschitz continuous solutions.
This part is identical between the two algorithms.
We make the distinction at the rounding steps (\Crefrange{subsec:exp-mech-vertex-S-T-cut}{subsec:vertex-S-T-cut}). 
Finally, we prove the lower bound in \Cref{subsec:s-t-cut-lower-bound}.

\subsection{Obtaining a Fractional Lipschitz Solution}\label{subsec:s-t-cut-fractional}
The first step in our algorithm is to obtain a fractional Lipschitz algorithm.
Fix any $s_0\in S, t_0\in T$, and consider the following convex optimization problem.
\begin{align}
    \begin{array}{llr}
        \text{minimize} & \displaystyle \sum_{uv\in E} w_{uv} \abs{y_u - y_v} & \\
        \text{subject to} & y_{t_0}-y_{s_0} = 1 & \\
        &y_s = y_{s_0} & \forall s\in S \\
        &y_t = y_{t_0} & \forall t\in T \\
        & \iprod{\ones, y} = 0 & \\
        & y\in [-1, 1]^V\cap [y_{t_0}, y_{s_0}]^V &
    \end{array}
    \label[LP]{eq:s-t-min-cut-lp}        
\end{align}
To see that this is a valid relaxation of the minimum $S$-$T$ cut problem,
let $\ones_A$ denote the indicator of any vertex set $A\sset V$ with $S\sset A$ and $A\cap T = \varnothing$.
Then $\ones_A - \frac{\card A}{n}\ones$ is a feasible solution to the program above
with the exact same objective as $\ones_A$.

The goal of this section is to show the following:
\begin{theorem}\label{thm:min-s-t-cut-fractional}
    For any $\eps > 0$, there exists a polynomial-time algorithm that, given a graph $G=(V,E)$ and a weight vector $\bm w \in \mathbb{R}_{\geq 0}^E$, outputs a feasible solution $\bm y$ to \Cref{eq:s-t-min-cut-lp} with the following properties:
    \begin{itemize}
        \item $f_{\bm w}(\bm y) \leq (1+\eps) f_{\bm w}(\bm y^*)$, where $\bm y^*$ is an optimal solution to \Cref{eq:s-t-min-cut-lp}.
        \item The fractional pointwise Lipschitz constant with respect to the $\ell_2$ norm is $O(\eps^{-1}\lambda_2^{-1})$. 
    \end{itemize}
\end{theorem}

To guarantee fractional pointwise Lipschitz continuity, we add a smooth and strongly convex regularizer as follows.
\begin{align*}
    \begin{array}{llr}
        \text{minimize}& \displaystyle \sum_{uv\in E} w_{uv} \abs{y_u - y_v} + \frac\eps2 \sum_{uv\in E} w_{uv} (y_u - y_v)^2 & \\
        \text{subject to} & y_{t_0}-y_{s_0} = 1 & \\
        &y_s = y_{s_0} & \forall s\in S \\
        &y_t = y_{t_0} & \forall t\in T \\
        & \iprod{\ones, y} = 0 & \\
        & y\in [-1, 1]^V\cap [y_{t_0}, y_{s_0}]^V &
        \end{array}
\end{align*}
This ensures that the optimal objective has a multiplicative error of $1+\eps$ with respect to the unregularized problem.

Let $\mcal L_{\bm w}$ denote the unnormalized Laplacian matrix with respect to the edge weights $\bm w$, \ie{}, 
\[
    \mcal L_{\bm w} = BWB^\top,
\]
where $B\in \R^{V\times E}$ is the signed edge-vertex incidence matrix (under some arbitrary orientation of the edges) and $W \in \R^{E\times E}$ is the matrix of edge weights.
Let $f_{\bm w}(\bm y) = \sum_{uv\in E} w_{uv} \abs{y_u - y_v}$
and $g_{\bm w}(\bm y) \coloneqq \frac\eps2 \bm y^\top \mcal L_{\bm w} \bm y$.
Then we can rewrite the problem as
\begin{align*}
    \begin{array}{llr}
        \text{minimize} & f_{\bm w}(\bm y) + g_{\bm w}(\bm y) & \\
        \text{subject to} & y_{t_0}-y_{s_0} = 1 \\
        &y_s = y_{s_0} & \forall s\in S \\
        &y_t = y_{t_0} & \forall t\in T \\
        & \iprod{\ones, y} = 0 & \\
        & y\in [-1, 1]^V\cap [y_{t_0}, y_{s_0}]^V &
    \end{array}
\end{align*}

Let $\lambda_1 \leq \dots \leq \lambda_n$ denote the eigenvalues of $\mcal L_{\bm w}$.
It is known that the first eigenvalue $\lambda_1 = 0$ corresponds to the all-ones vector $\ones$,
which can break the strong convexity of the
objective function if we optimize over
vectors in $[-1, 1]^V.$
However,
since we restrict ourselves to work in the subspace $\iprod{\ones, \bm y}=0$,
we see that $g_{\bm w}(\bm y) = \frac\eps2 \bm y^\top\mcal L_{\bm w} \bm y$ is $\eps \lambda_{n}$-smooth and $\eps \lambda_2$-strongly convex.

We now prove \Cref{thm:min-s-t-cut-fractional}.

\begin{remark}
    We assume without loss of generality that
    $\lambda_2 > 0.$ Indeed, if $\lambda_2 = 0$
    then our bounds trivially hold under the convention $\nicefrac10 = \infty$.
\end{remark}

We wish to apply \Cref{thm:PGM} and \Cref{lem:one perturbation} over a small neighborhood of $\bm w$.
Recall by \Cref{lem:one perturbation} that it does not suffice to apply \Cref{thm:PGM} just for $\bm w, \tilde{\bm w}$.
We must perform the analysis for all $\hat{\bm w}$ in a neighborhood of $\bm w$
and $\tilde{\bm w}$ obtained from $\hat{\bm w}$ by perturbing a single coordinate.
To do so,
we first check that $f_{\bm w}$ and $g_{\bm w}$ satisfy \Cref{asm:PGM}.
Fix $\hat{\bm w}, \tilde{\bm w}$ in a small neighborhood of $w$.
Let $0\neq \delta\in \R$ be a real number arbitrarily in small absolute value
and suppose that we obtain $\tilde{\bm w}$ from $\hat{\bm w}$ by perturbing a single coordinate of $\hat{\bm w}$. 
We have the following properties:
\begin{enumerate}
    \item $f_{\hat{\bm w}}, f_{\tilde{\bm w}}$ are convex and continuous over $K$;
    \item $g_{\hat{\bm w}}, f_{\tilde{\bm w}}$ are $\sigma_{\bm w}$-strongly convex over $K$ for $\sigma_{\bm w} = \eps \lambda_2(\mcal L_{\bm w})/2$,
    that depends only on $\bm w$ (assuming $\hat{\bm w}, \tilde{\bm w}$ lies in the neighborhood of $\bm w$);
    \item for sufficiently small $\delta$,
    $g_{\hat{\bm w}}, g_{\tilde{\bm w}}$ is $L_{\bm w}$-smooth over $K$ for $L_{\bm w} = 2 \lambda_n(\mcal L_{\bm w})$ (assuming $\hat{\bm w}, \tilde{\bm w}$ lies in the neighborhood of $w$); 
    \item for every $x\in K$, $\frac1{L_{\bm w}} \norm{\grad g_{\hat{\bm w}}(\bm x) - \grad g_{\tilde{\bm w}}(\bm x)}_2 \leq C_{\bm w} \abs{\delta}$ 
    for $C_{\bm w} = O(\nicefrac\eps{L_{\bm w}})$; (\Cref{lem:min cut gradient})
    \item for every $\bm x\in K$, $\norm{\prox_{L_{\bm w}, f_{\hat{\bm w}}}(\bm x) - \prox_{L_{\bm w}, f_{\tilde{\bm w}}}(\bm x)}_2 \leq D_{\bm w} \abs{\delta}$ 
    for $D_{\bm w} = O(\nicefrac1{L_{\bm w}})$. (\Cref{lem:min cut prox})
\end{enumerate}    

\begin{lemma}\label{lem:min cut gradient}
    Suppose we obtain $\tilde{\bm w}$ from $\hat{\bm w}$ by perturbing the $uv$-th coordinate by some $\delta \neq 0$.
    Then
    \begin{align*}
        \frac1{L_{\bm w}} \norm{\grad g_{\hat{\bm w}}(\bm y) - \grad g_{\tilde{\bm w}} (\bm y)}_2
        = O\left( \frac{\abs{\delta}}{L_{\bm w}} \right).
    \end{align*}
\end{lemma}

\begin{proof}
    We have
    \begin{align*}
        \frac1{L_{\bm w}} \norm{\grad g_{\hat{\bm w}}(\bm y) - \grad g_{\tilde{\bm w}} (\bm y)}_2
        = \frac\eps{L_{\bm w}} \norm{(\mcal L_{\hat{w}} - \mcal L_{\tilde{w}}) y}_2 
        = \frac\eps {L_{\bm w}} \sqrt{(\delta y_u)^2 + (\delta y_v)^2} 
        \leq \frac{\sqrt2 \eps  \abs{\delta}}{L_{\bm w}},
    \end{align*}
    where we used the fact $\bm y\in [-1, 1]^V$ in the inequality.
\end{proof}

\begin{lemma}\label{lem:min cut prox}
    Suppose we obtain $\tilde{\bm w}$ from $\hat{\bm w}$ by perturbing the $uv$-th coordinate by some $\delta \neq 0$.
    Then for every $\bm x\in K$,
    \[
        \norm{\prox_{L_{\bm w}, f_{\hat{\bm w}}}(\bm x) - \prox_{L_{\bm w}, f_{\tilde{\bm w}}}(\bm x)}_2
        =  O\left( \frac{\abs{\delta}}{L_{\bm w}} \right).
    \]
\end{lemma}

\begin{proof}
    Let $\hat{\bm z} = \prox_{L_{\bm w}, f_{\hat{\bm w}}}(\bm y)$ and $\tilde{\bm z} = \prox_{L_{\bm w}, f_{\tilde{\bm w}}}(\bm y)$.
    By the $L_{\bm w}$-strong convexity of the proximal objective $\bar f_{\hat{\bm w}, \bm y}(\cdot) \coloneqq \frac{ L_{\bm w}}2 \norm{\cdot -\bm y}_2^2 + f_{\hat{\bm w}}(\cdot)$, 
    we see that
    \begin{align*}
        \frac{L_{\bm w}}2 \norm{\hat{\bm z}-\tilde{\bm z}}_2^2 
        &\leq f_{\hat{\bm w}}(\tilde{\bm z}) + \frac{L_{\bm w}}2 \norm{\tilde{\bm z}-\bm y}_2 - f_{\hat{\bm w}}(\hat{\bm z}) - \frac{L_{\bm w}}2 \norm{\hat{\bm z}-\bm y}_2 - \underbrace{\iprod{\grad \bar f_{\hat{\bm w}, y}(\hat{\bm z}), \tilde{\bm z}-\hat{\bm z}}}_{\geq 0} \\
        &\leq f_{\hat{\bm w}}(\tilde{\bm z}) + \frac{L_{\bm w}}2 \norm{\tilde{\bm z}-\bm y}_2 - f_{\hat{\bm w}}(\hat{\bm z}) - \frac{L_{\bm w}}2 \norm{\hat{\bm z}-\bm y}_2
    \end{align*}
    and vice versa for $\frac{L_{\bm w}}2 \norm{\cdot -\bm y}_2^2 + f_{\tilde{\bm w}}(\cdot)$.
    The second inequality here follows from the fact that the directional derivatives in every feasible direction is non-negative
    for any (local) minima of $\bar f_{\hat{\bm w}, \bm y}$.
    Adding these two inequalities leads to the regularization terms canceling out, leaving
    \begin{align*}
        \frac{L_{\bm w}+L_{\bm w}}2 \norm{\hat{\bm z} - \tilde{\bm z}}_2^2
        &\leq f_{\hat{\bm w}}(\tilde{\bm z}) - f_{\hat{\bm w}}(\hat{\bm z}) + f_{\tilde{\bm w}}(\hat{\bm z}) - f_{\tilde{\bm w}}(\tilde{\bm z}) \\
        &= [f_{\hat{\bm w}}(\tilde{\bm z}) - f_{\tilde{\bm w}}(\tilde{\bm z})] - [f_{\hat{\bm w}}(\hat{\bm z}) - f_{\tilde{\bm w}}(\hat{\bm z})] \\
        &= (\hat{w}_{uv} - \tilde{w}_{uv})\abs{\tilde{z}_u - \tilde{z}_v} -  (\hat{w}_{uv} - \tilde{w}_{uv})\abs{\hat{z}_u - \hat{z}_v} \\
        &\leq \abs{\delta}\cdot \abs*{\abs{\tilde{z}_u - \tilde{z}_v} - \abs{\hat{z}_u - \hat{z}_v}} \\
        &\leq \abs{\delta}\cdot \abs{\tilde{z}_u - \tilde{z}_v - \hat{z}_u + \hat{z}_v} \tag{reverse triangle inequality} \\
        &= \abs{\delta} \sqrt{(\tilde{z}_u - \tilde{z}_v - \hat{z}_u + \hat{z}_v)^2} \\
        &= \abs{\delta} \sqrt{[(\tilde{z}_u - \hat{z}_u) + (\hat{z}_v - \tilde{z}_v)]^2} \\
        &\leq \abs{\delta} \sqrt{2(\tilde{z}_u - \hat{z}_u)^2 + 2(\hat{z}_v - \tilde{z}_v)^2} \tag{convexity of $x^2$} \\
        &\leq \sqrt2 \delta \norm{\hat{\bm z} - \tilde{\bm z}}_2.
    \end{align*}
    It follows that
    \[
        \norm{\hat{\bm z}-\tilde{\bm z}}_2
        \leq \frac{\sqrt2 \abs{\delta}}{L_{\bm w}}.
        \qedhere
    \]
\end{proof}

Equipped with the previous results,
we see that we can apply \Cref{thm:PGM}.
This gives the following guarantee
regarding the distance of the optimal
solutions under two different executions.
\begin{corollary}\label{cor:pgm-min-st-cut-lipschitz-fractional}
    Let $G=(V,E)$ be a graph and $\bm w \in \R_{\geq 0}^E$ be a weight vector.
    Let $\hat{\bm w}, \tilde{\bm w} \in \R_{\geq 0}^E$ be two weight vectors in a sufficiently
    small neighborhood of $\bm w$, where $\tilde{\bm w} = \hat{\bm w} + \delta \ones_e,$ for some sufficiently 
    small $\delta > 0$ and some edge $e \in E$. Let 
    $\hat{\bm y}^*, \tilde{\bm y}^*$ be the optimal
    solutions of the regularized
    objectives $h_{\hat{\bm w}}, h_{\tilde{\bm w}},$ respectively. Then,
    \[
    \norm{\hat{\bm y}^* - \tilde{\bm y}^*}_2
    \leq \frac{L_{\bm w}(C_{\bm w}+D_{\bm w})}{\sigma_{\bm w}} \abs{\delta}
    = O\left( \frac{\abs{\delta} }{\eps \lambda_2} \right).
\]  
\end{corollary}

By~\Cref{cor:pgm-min-st-cut-lipschitz-fractional} and~\Cref{lem:one perturbation}, we obtain~\Cref{thm:min-s-t-cut-fractional}.

To extract a cut from a fractional solution $\bm y$,
we must round the fractional solution to an integral one.
Consider the standard threshold rounding scheme described below:
Sample $\tau\sim U[-1,1]$ and take the threshold cut $A_\tau \coloneqq \set{v\in V: y_v\leq \tau}$.
It is not hard to see that the following holds:
\begin{lemma}[Threshold Rounding]\label{lem:threshold-rounding}
    We have $\E_\tau [\cut_{G,\bm w}(A_\tau)] = \frac12 \sum_{uv \in E} w_{uv} \abs{y_u - y_v} = \frac12 f_{\bm w}(\bm y)$.
    Moreover, we have $\E_{\tau} |A_\tau \triangle \tilde A_\tau| = O(\|\bm y - \tilde{\bm y}\|_1)$.
\end{lemma}
We say that a set $A \subseteq V$ is \emph{feasible} if $S \sset A \sset V\setminus T$ and \emph{infeasible} otherwise.
Then, the threshold rounding may output an empty cut or the whole vertex set, which are infeasible, with probability (at most) $\nicefrac12$.
We develop various techniques for overcoming this issue in the next subsections.

\subsection{Rounding via Exponential Mechanism}\label{subsec:exp-mech-vertex-S-T-cut}
We can round a fractional solution through the \emph{exponential mechanism} \cite{mcsherry2007mechanism}.
We use a slightly modified version of its guarantees due to \cite{varma2021average}.

\begin{theorem}[Exponential Mechanism; Lemma 2.1 in \cite{varma2021average}]\label{thm:exponential-mechanism}
    Let $\eta > 0$ and $\mcal A_\eta$ be the algorithm that, given a vector $\bm x\in \R^k$, samples an index $i\in [k]$ with probability proportional to $\exp(-\eta x_i)$.
    Then, we have
    \[
        \Pr_{i\sim \mcal A_\eta(x)} \left[ x_i \geq \min_{j\in [k]} x_j + \frac{\log k}\eta + \frac{t}\eta \right]
        \leq \exp(-t).
    \]
    Moreover for $\tilde{\bm x} \in \R^k$,
    \[
        \TV(\mcal A_\eta(\bm x), \mcal A_\eta(\tilde{\bm x}))
        = O(\eta \norm{\bm x-\tilde{\bm x}}_1).
    \]
\end{theorem}

Let $\gamma\in (0, 1)$ be a parameter such that $1/\gamma$ is an integer. 
If we know that the optimal $S$-$T$ cut $A$ has cardinality $[\gamma(i-1)n,\gamma i n]$ for an integer $i \in [1/\gamma]$, then we have $\ones_A- \frac{|A|}{n} \ones \in [-\gamma i, 1-(i-1)\gamma ]^V$.
Hence, we consider the following (regularized) convex relaxation
of the problem.
\begin{align}
    \begin{array}{llr}
        \text{minimize} & f_{\bm w}(\bm y) + g_{\bm w}(\bm y) & \\
        \text{subject to} & y_{s_0}-y_{t_0} = 1 \\
        &y_s = y_{s_0} & \forall s\in S \\
        &y_t = y_{t_0} & \forall t\in T \\
        & \iprod{\ones, y} = 0 & \\
        & \bm y\in [-\gamma i, 1-(i-1)\gamma]^{V}\cap [y_{t_0}, y_{s_0}]^{V} &
    \end{array}
    \label[LP]{eq:lp-min-s-t-cut-restricted-domain}
\end{align}
We solve this LP for each $i \in [1/\gamma]$ and then apply the exponential mechanism to select a solution.
The algorithm is summarized as~\Cref{alg:min-s-t-cut-additive}.
\begin{algorithm}[htp!]
    \caption{Minimum Vertex $S$-$T$ Cut via Exponential Mechanism}\label{alg:min-s-t-cut-additive}
    \KwIn{A graph $G=(V,E)$, a weight vector $\bm w \in \mathbb{R}_{\geq 0}^E$, and $\gamma\in (0, 1)$}
    \For{$i \in [\gamma^{-1}]$}{
        Solve~\Cref{eq:lp-min-s-t-cut-restricted-domain} for $i$ using \Cref{thm:min-s-t-cut-fractional} with $\eps \coloneqq \nicefrac1{\sqrt{n}}$\;
        Let $\theta^{(i)} \in \mathbb{R}$ and $\bm y^{(i)} \in \mathbb{R}^V$ be the obtained objective value and solution, respectively.
    }
    Sample $\Lambda_2 \sim U[\nicefrac{\lambda_2}{2}, \lambda_2]$\;
    Apply the exponential mechanism on $\bm \theta=(\theta^{(1)},\ldots,\theta^{(1/\gamma)})$ with $\eta = \gamma  \eps^{-1}\Lambda_2^{-1} /\sqrt{n}$\;
    Let $i^* \in [n]$ be the chosen index\;
    Apply the threshold rounding on $\bm y^{(i^*)}$, where we sample the threshold $\tau$ uniformly from $[-\gamma i^*,1-(i^*-1)\gamma]$\;
\end{algorithm}

This leads to a proof of \Cref{thm:s-t-min-cut-additive},
our first main result in this section.
\begin{proof}[Proof of \Cref{thm:s-t-min-cut-additive}]
    We first analyze the solution quality.
    By combining \Cref{thm:min-s-t-cut-fractional} and \Cref{thm:exponential-mechanism} with $t=\log \nicefrac1\gamma$, 
    we see that with probability at least $1-\gamma$, 
    the selected fractional solution has objective at most
    \[
        (1+\eps) \OPT + O\left(\frac{\log \gamma^{-1}}{\eta}\right)
        = (1+\eps) \OPT + O\left(\frac{\eps \lambda_2 \sqrt{n} \log \gamma^{-1}}{\gamma}\right)
        = \left(1+\frac{1}{\sqrt{n}}\right) \OPT + O\left(\frac{\lambda_2 \log \gamma^{-1}}{\gamma}\right).
    \]
    Threshold rounding yields a feasible $S$-$T$ cut with probability at least $1-\gamma$
    and conditioned on feasibility,
    the expected cost is the same as the objective of the fractional solution.
    An application of Markov's inequality on the non-negative random variable $\cut_{G, \bm w}(A) - \OPT$ yields the desired approximation guarantees.

    Next, we analyze the pointwise Lipschitz constant.
    By Weyl's inequality (\Cref{thm:weyl-inequality}),
    we have $\abs{\lambda_2 - \tilde \lambda_2}\leq 2\delta$.
    But then we can verify with a straightforward computation that
    \[
        \TV(\Lambda_2, \tilde \Lambda_2)
        = \TV(U[\nicefrac{\lambda_2}2, \lambda_2], U[\nicefrac{\tilde \lambda_2}2, \tilde \lambda_2])
        = O(\nicefrac\delta{\lambda_2})
    \]
    for a sufficiently small $\delta$.
    Let $\bm \theta\in \R^{\gamma^{-1}}$ denote the vector consisting of optimal fractional values of the $\gamma^{-1}$ different convex relaxations we consider.
    Noting that \Cref{thm:min-s-t-cut-fractional} yields 
    \[
        \norm{\bm y - \tilde{\bm y}}_1
        = O\left( \frac{\sqrt{n}\cdot \abs{\delta}}{\eps \lambda_2} \right),
    \]
    and that for every index $i \in [1/\gamma]$
    we have that $\norm{\bm y^{(i)} - \tilde{\bm y}^{(i)}}_1  = O\left( \frac{\sqrt{n}\cdot \abs{\delta}}{\eps \lambda_2} \right)$,
    the pointwise Lipschitz constant is bounded as
    \begin{align*}
        & [1-\TV(\Lambda_2, \tilde \Lambda_2)-\TV(\mcal A(\bm \theta), \mcal A(\tilde{\bm \theta}))]\cdot O\left( \frac{\sqrt{n}}{\eps \lambda_2} \right) +  [\TV(\Lambda_2, \tilde \Lambda_2) + \TV(\mcal A(\bm \theta), \mcal A(\tilde{\bm \theta}))]\cdot n \\
        &= O\left( \frac{\sqrt{n}}{\eps \lambda_2} \right) + O\left(\frac{\eta}{\gamma} + \frac1{\lambda_2} \right)\cdot n 
        = O\left( \frac{n}{\lambda_2} \right) + O\left( \frac1{\lambda_2} \right)\cdot n
        = O\left( \frac{n}{\lambda_2} \right).
    \end{align*}
    The second last inequality follows from the choices $\eps = \nicefrac1{\sqrt{n}}$
    and $\eta =  \gamma \eps^{-1}\Lambda_2^{-1} /\sqrt{n} = \Theta(\nicefrac\gamma{\lambda_2})$.
\end{proof}

\subsection{Rounding via \texorpdfstring{$k$}{k}-Way Submodularity}\label{subsec:vertex-S-T-cut}

Recall that the cut function $\mathrm{cut}_{G, \bm w}: 2^V\to \R$
is a \emph{submodular} function,
\ie{}, for every $A, B\sset V$,
\[
    \mathrm{cut}_{G, \bm w}(A) + \mathrm{cut}_{G, \bm w}(B)
    \geq \mathrm{cut}_{G, \bm w}(A\cup B) + \mathrm{cut}_{G, \bm w}(A\cap B).
\]
To boost the probability of outputting feasible sets, we use the following result regarding submodular functions due to \cite{harvey2006capacity}.
\begin{theorem}[$k$-way submodularity; Theorem 2 in \cite{harvey2006capacity}]\label{thm:k-way submodularity}
    Let $f:2^{\mcal X}\to \R$ be a submodular function
    and $A_1, \dots, A_k\sset \mcal X$ be a collection of $k$ sets.
    For $r\in [k]$,
    define
    \[
        B_{r, k}
        \coloneqq \bigcup_{\set{i_1, \dots, i_r}\sset \binom{[k]}{r}} \bigcap_{j=1}^r A_{i_j}.
    \]
    Then, we have
    \[
        \sum_{r=1}^k f(A_r)
        \geq \sum_{r=1}^k f(B_{r, k}).
    \]
    Note that the case of $k=2$ is simply the definition of submodularity.
\end{theorem}

First,
we consider the following convex relaxation for the \emph{$\beta$-balanced} $S$-$T$ cut problem
for some $\beta\in [0, \nicefrac12]$.
\begin{align}
    \begin{array}{llr}
        \text{minimize} & f_{\bm w}(\bm y) + g_{\bm w}(\bm y) & \\
        \text{subject to} & y_{s_0}-y_{t_0} = 1 \\
        &y_s = y_{s_0} & \forall s\in S \\
        &y_t = y_{t_0} & \forall t\in T \\
        & \iprod{\ones, y} = 0 & \\
        & y\in [-1+\beta, 1-\beta]^{V}\cap [y_{t_0}, y_{s_0}]^{V} &
    \end{array}
    \label[LP]{lp:min-balanced-s-t-cut}
\end{align}

\begin{algorithm}[htp!]
    \caption{Rounding for Minimum Vertex $S$-$T$ Cut via $k$-Way Submodularity}\label{alg:conditional-cut-rounding}
    \KwIn{A graph $G=(V,E)$, a vector $\bm y \in \mathbb{R}_{\geq 0}^V$, $\beta \in [0, \nicefrac12]$, and $\gamma\in (0, 1)$}
    \For{$i=1$ to $k \coloneqq \Theta(\beta^{-2}\log(\nicefrac1\gamma))$}{    
        Sample $\tau_i \in [-1+\beta, 1-\beta]$ uniformly at random\;
        $A_i \gets \{v \in V : y_v \leq \tau_i\}$\;
    }
    Sample an integer $r$ from $\{k/2,\ldots,(1/2+\beta/4)k\}$ uniformly at random\;
    $B_{r,k} \gets \bigcup_{\{i_1,\ldots,i_r\} \in \binom{[k]}{r}} \bigcap_{j=1}^r A_{i_j}$\;
    \Return $B_{r, k}$\;
\end{algorithm}

For $\beta \in [0,1/2]$, we say that a set $A\subseteq V$ is \emph{$\beta$-balanced}
if $\min\{|A|,|V\setminus A|\} \geq \beta n$, or in other words, $\beta n \leq |A| \leq (1 - \beta)n$.
Let $\OPT_\beta$ denote the minimum weight of a $\beta$-balanced $S$-$T$ cut.
Using $k$-way submodularity, we provide an algorithm with small pointwise Lipschitz constant that outputs a cut of weight comparable to $\OPT_\beta$.
Note that, for a $\beta$-balanced set $A$, we have $\ones_A - \frac{\card A}{n}\ones \in [-1+\beta,1-\beta]^V$.
Then, we consider \Cref{alg:conditional-cut-rounding}, which first applies threshold rounding multiple times, where the threshold is sampled from a smaller range than $[-1,1]$, and then applies $k$-way submodularity to aggregate the obtained sets.
\begin{lemma}\label{lem:conditional-cut-feasible}
    Suppose that a vector $\bm y \in \mathbb{R}_{\geq 0}^V$ is a feasible solution to \Cref{lp:min-balanced-s-t-cut}.
    Then, the set $B_{r,k}$ in~\Cref{alg:conditional-cut-rounding} is feasible for $r \in \{k/2,\ldots,(1/2+\beta/4)k\}$ with probability at least $1-\gamma$.
\end{lemma}

\begin{proof}
    Because $y_{s_0} - y_{t_0} = 1$, for each $i \in [k]$, the set $A_i$ is feasible with probability at least
    \[
      p
      \coloneqq \frac{1}{(1-\beta)-(-1+\beta)}
      = \frac{1}{2-2\beta} \,.
    \]
    Let $X$ be the number of $i$'s such that $A_i$ is feasible. 
    By a multiplicative Chernoff bound, we have
    \[
        \Pr\left[X \leq (1-\phi)p k \right] \leq \exp\left(- \frac{\phi^2 pk}{2}\right)
    \]
    for any $\phi\in (0, 1)$.
    By setting $\phi = \beta/3$, we have
    \begin{align*}
      (1-\phi)p
      &= \frac12 \left( 1-\frac\beta3 \right)\left( \frac1{1-\beta} \right) \\
      &\geq \frac12 \left( 1-\frac\beta3 \right)(1+\beta) \tag{$1+x\leq \frac1{1-x}, x\in [-1, 1)$} \\
      &= \frac12 \left( 1+\frac{2\beta}3 - \frac{\beta^2}3 \right) \\
      &\geq \frac12 + \frac\beta4. \tag{$\beta\in [0, \nicefrac12]$}
    \end{align*}
    It follows that
    \[
		    \Pr\left[X \leq \left(\frac{1}{2} + \frac{\beta}{4} \right)k \right]
        \leq \exp\left(- \Omega(\beta^2 k ) \right).
    \]
    By choosing the hidden constant in $k = \Theta(\beta^{-2}\log(\nicefrac1\gamma))$ to be large enough, we have
    \[
        \Pr\left[X \leq \left(\frac{1}{2} + \frac{\beta}{4} \right)k \right] \leq \gamma.
    \]

    Suppose the above event does not occur so that at least $(\nicefrac12 + \nicefrac\beta4)k + 1$ sets are feasible.
    Reorder the $A_i$'s so that $A_1,\ldots,A_{(1/2+\beta/4)k+1}$ are feasible, and let $k/2 \leq r \leq (1/2+\beta/4)k$.
    Then, the set $\bigcap_{j=1}^r A_j$ is feasible and in particular, non-empty.
    {This is because
    intersections of feasible cuts
    are feasible, and under the good event
    we have conditioned on, the sets $A_1,\ldots,A_r$
    are feasible.
    Similarly, it is not hard to see
    that every intersection  $\bigcap_{j=1}^r A_{i_j}$ contains at least one feasible set, so
    this intersection does not contain any elements
    from $T.$ Thus, these two properties
    show that by taking a union over these
    intersections $\bigcup_{\{i_1,\ldots,i_r\} \in \binom{[k]}{r}} \bigcap_{j=1}^r A_{i_j}$
    the final set will contain all elements
    of $S$ and no elements from $T.$
    }
\end{proof}

We first analyze the approximation guarantee of~\Cref{alg:conditional-cut-rounding}.
\begin{lemma}\label{lem:conditional-cut-approximation}
    The output $A \subseteq V$ of~\Cref{alg:conditional-cut-rounding} satisfies
    \[
	    \E[\mathrm{cut}_{G, \bm w}(A)] = O(\beta^{-1} f_{\bm w}(\bm y) ).
    \]
\end{lemma}
\begin{proof}
    By the $k$-way submodularity of $\mathrm{cut}_{G, \bm w}$ (\Cref{thm:k-way submodularity}), we have
    \[
        \sum_{r=1}^k \mathrm{cut}_{G, \bm w}(A_r) 
        \geq \sum_{r=1}^k \mathrm{cut}_{G, \bm w}(B_{r,k})
        \geq \sum_{r=k/2}^{(1/2+\beta/4)k} \mathrm{cut}_{G, \bm w}(B_{r,k}).
    \]
    Hence, we have
    \[
    	\E[\mathrm{cut}_{G, \bm w}(A)] = \frac{4}{\beta k}\sum_{r=k/2}^{(1/2+\beta/4)k} \E[\mathrm{cut}_{G, \bm w}(B_{r,k})]
      \leq \frac{4}{\beta k} \sum_{r=1}^k \E[\mathrm{cut}_{G, \bm w}(A_r)] 
    	= O(\beta^{-1} f_{\bm w}(\bm y) ).
     \qedhere
    \]
\end{proof}

We use the following simple fact to analyze the pointwise Lipschitz constant of~\Cref{alg:conditional-cut-rounding}.
\begin{proposition}\label{pro:sets}
    Let $\mcal I$ be an arbitrary indexing set
    and $A \Delta B$ denote the symmetric difference between sets $A$ and $B$.
    Then for any collections of sets indexed by $\mcal I$, $\set{A_\alpha}_{\alpha\in \mcal I}, \set{B_\alpha}_{\alpha\in \mcal I}$,
    the following hold:
    \begin{enumerate}[(a)]
        \item $\left( \bigcup_{\alpha\in \mcal I} A_\alpha \right) \Delta \left( \bigcup_{\alpha\in \mcal I} B_\alpha \right) \sset \bigcup_{\alpha\in \mcal I} (A_\alpha\Delta B_\alpha)$ \label[part]{pro:sets-union}
        \item $\left( \bigcap_{\alpha\in \mcal I} A_\alpha \right) \Delta \left( \bigcap_{\alpha\in \mcal I} B_\alpha \right) \sset \bigcup_{\alpha\in \mcal I} (A_\alpha\Delta B_\alpha)$.
    \end{enumerate}
\end{proposition}

\begin{lemma}\label{lem:conditional-cut-lipschitz}
    Let $A$ and $\tilde A$ be the outputs of~\Cref{alg:conditional-cut-rounding} for $\bm y$ and $\tilde{\bm y}$, respectively.
    Then, we have
    \[
        \E[|A \triangle \tilde A|] \leq O(\beta^{-2}\log(\nicefrac1\gamma) \cdot \|\bm y - \tilde{\bm y}\|_1 ).
    \]
\end{lemma}
\begin{proof}
    For any $r \in [k]$, we have
    \begin{align*}
        B_{r,k} \Delta \tilde B_{r,k} 
        & = \left( \bigcup_{I\in \binom{[k]}{r}} \bigcap_{i\in I} A_i \right)\Delta \left( \bigcup_{I\in \binom{[k]}{r}} \bigcap_{i\in I} \tilde A_i \right) \\
        &\sset \bigcup_{I\in \binom{[k]}{r}} \left( \left( \bigcap_{i\in I} A_i \right) \Delta \left( \bigcap_{i\in I} \tilde A_i \right) \right)  \tag{by \Cref{pro:sets} (a)} \\
        &\sset \bigcup_{I\in \binom{[k]}{r}} \bigcup_{i\in I} (A_i\Delta \tilde A_i)\tag{by \Cref{pro:sets} (b)}  \\
        &\sset \bigcup_{j=1}^k (A_j\Delta \tilde A_j).
    \end{align*}
    For each $i \in [k]$, we have $\E [|A_i \triangle \tilde A_i|] \leq \|\bm y^* - \tilde{\bm y}^* \|_1$.
    But then by the calculations above, we have $\E [|A \triangle \tilde A|] \leq k \cdot \|\bm y^* - \tilde{\bm y}^*\|_1 = O(\beta^{-2}\log(\nicefrac1\gamma) \cdot \|\bm y^* - \tilde{\bm y}^*\|_1 )$.
\end{proof}

We are now ready to prove \Cref{thm:s-t-cut-vertex-set}
\begin{proof}[Proof of \Cref{thm:s-t-cut-vertex-set}]
    The approximation guarantee is clear from \Cref{thm:min-s-t-cut-fractional,lem:conditional-cut-approximation}.
    \Cref{lem:conditional-cut-feasible} gives the feasibility guarantee.
    To analyze the pointwise Lipschitz constant, first note that \Cref{thm:min-s-t-cut-fractional} yields 
    \[
        \norm{\bm y - \tilde{\bm y}}_1
        = O\left( \frac{\sqrt{n}\cdot \abs{\delta}}{\lambda_2} \right).
    \]
    Combined with \Cref{lem:conditional-cut-approximation,lem:conditional-cut-lipschitz}, we obtain the claimed bound.
\end{proof}

\subsection{A Lower Bound}\label{subsec:s-t-cut-lower-bound}
In this section, we prove \Cref{thm:s-t-cut-additive-lower-bound},
a lower bound on the pointwise Lipschitz constant for the minimum vertex $S$-$T$ cut problem.
This shows that our Lipschitz-approximation tradeoff in \Cref{thm:s-t-min-cut-additive} is tight.

\begin{proof}[Proof of \Cref{thm:s-t-cut-additive-lower-bound}]
    Let $\mathcal{A}$ be an arbitrary $(C, f(n))$-approximation algorithm for the minimum $s$-$t$ cut problem for $C = O(1)$ and let $ f(n) < n/2(C+2)$ be an increasing function of $n$.
    Let $G=(U\cup R, E)$ be a complete bipartite graph, where $|U| =  (C+2) f(n)$ and $|R| = n - (C+2) f(n)$.
    Let $s,t$ be two distinct vertices in $R$.
    We define two weight vectors ${\bm w}, \tilde{\bm w}\in [0,1]^E$ as follows:
    \[
        w_e = \begin{cases}
            1/4(C+2) & \text{if } s \in e, \\
            1 & \text{otherwise},
        \end{cases}
        \qquad
        \tilde w_e = \begin{cases}
            1/4(C+2) & \text{if } t \in e, \\
            1 & \text{otherwise}.
        \end{cases}    
    \]
    The minimum $s$-$t$ cuts of $(G,\bm w)$ and $(G,\tilde{\bm w})$ are $\{s\}$ and $V \setminus \{t\}$, respectively, and both have cut weights of $\OPT = |U|/4(C+2) = f(n)/4$.
    Moreover, for both weighted graphs, any other $s$-$t$ cut has weight at least $|U| = (C+2)f(n) > C\cdot \OPT + f(n)$.
    Thus $\mcal A$ outputs $\set{s}$ with probability at least $0.99$ given $G, \bm w$
    and $V\setminus \set{t}$ with probability at least $0.99$ given $G, \tilde{\bm w}$.
    Hence, the earth mover's distance between the output distributions of $\mathcal{A}$ on $\bm w$ and that on $\tilde{\bm w}$ is at least 
    \[
        0.98 \cdot (n-2) = \Omega(n)
    \]
    
    Consider a transition of the weight vector from $\bm w$ to $\tilde{\bm w}$, where the edge weights never go below $1/4(C+2)$.
    It is well known that $\lambda_2(G,\ones) = \min\{|U|,|R|\} = |U| = (C+2) f(n)$ because $|U|=(C+2) f(n)< n/2$.
    Since $\lambda_2$ scales proportionally with weights, $\lambda_2(G,\ones/4(C+2)) = |U|/4(C+2) = f(n)/4$.
    As $\lambda_2$ is increasing in the edge weight of an any edge, the algebraic connectivity of every graph in the transition is at least $f(n)/4$.
    Hence by \Cref{thm:finite perturbation}, there exists $\bm w^*$ such that the pointwise Lipschitz constant at $\bm w^*$ is at least 
    \[
      \Omega\left( \frac{n}{\norm{\bm w - \tilde{\bm w}}_1} \right)
      = \Omega\left(\frac{n}{|U|}\right) 
      = \Omega\left(\frac{n}{f(n)}\right) 
      = \Omega\left(\frac{n}{\lambda_2(G,\bm w^*)}\right).
        \qedhere
    \]    
\end{proof}

\section{Maximum Weight Bipartite \texorpdfstring{$\bm b$}{b}-Matching}\label{sec:b-matching}
In this section,
we design a bipartite matching algorithm which cleanly extends to handle $\bm b$-matchings.
Our goal is to prove the following results.

\begin{theorem}\label{thm:mwm-bipartite}
    For any bipartite graph $G = (V = U \cup R, E)$, weight vector $\bm w \in \R_{> 0}$, and
    approximation parameter
    $\eps > 0,$ \Cref{alg:maximum bipartite matching} runs in polynomial time and
    returns a $2(1+\eps)$-approximate maximum weighted matching 
    in expectation,
    with pointwise Lipschitz constant $O\left(\frac{\sqrt{m}}{\eps w_{\min}}\right).$ 
\end{theorem}

\begin{theorem}\label{thm:maximum bipartite b-matching}
    Fix $\eps > 0$.
    \Cref{alg:maximum b-matching} terminates in 
    $\poly(n)$ time
    and outputs a $2(1+\eps)e/(e-1)$-approximate maximum weight $\bm b$-matching in expectation
    for bipartite graphs.
    Moreover,
    its pointwise Lipschitz constant is
    $
        O\left(\frac{\sqrt{m}}{\eps w_{\min}}\right).
    $
\end{theorem}

In \Cref{sec:bipartite},
we design and analyze our bipartite matching algorithm.
Having established the general analysis of PGTA,
the analysis of the fractional Lipschitz solution is relatively painless.
We directly use the \cite{kumabe2022lipschitz} rounding algorithm for bipartite matchings.

In \Cref{sec:bmatching:fractional},
we describe how to obtain fractional Lipschitz solutions for bipartite $\bm b$-matchings.
We can no longer directly use the \cite{kumabe2022lipschitz} rounding scheme in this setting
and we design a more involved rounding scheme inspired by their rounding algorithm in \Cref{sec:bmatching:rounding}.

\subsection{Warmup for Bipartite $1$-Matching}\label{sec:bipartite}
The LP that obtains a bipartite maximum (unweighted and weighted) matching is as follows.
Let $U$ and $R$ be the sets of the vertices on the left hand side and right hand sides of the bipartite graph and $\bm w=\{w_{uv}\}_{uv \in E}$ is the 
vector representing the weights of every edge. 
We regularize the standard LP relaxation using the simple strongly convex regularizer of $\sum_{uv \in E} w_{uv} x_{uv}^2$.

\begin{align}\label[LP]{lp:relaxed-bipartite-matching}
    \begin{array}{lll}
        \text{minimize} & \displaystyle \sum_{uv\in E} -w_{uv}x_{uv} + \frac{\eps}{2} \cdot \sum_{uv \in E} w_{uv} x_{uv}^2\\
        \text{subject to} & \displaystyle \sum_{v \in R} x_{uv} \leq 1 & \forall u \in U\\
        & \displaystyle \sum_{u \in L} x_{uv} \leq 1 & \forall v \in R\\
        & 0 \leq x_{uv} \leq 1 & \forall uv\in E.
    \end{array}
\end{align}

We now proceed with the framework introduced in~\cref{sec:framework} and the notation used within that section.
{ The next result
shows that the conditions in \Cref{asm:PGM} are satisfied.

\begin{lemma}\label{lem:bipartite-matching-pgm}
    The convex program in \Cref{lp:relaxed-bipartite-matching}
    satisfies the conditions in \Cref{asm:PGM} with $L_{\bm w} = 2\eps w_{\max}, \sigma_{\bm w} = \eps w_{\min}/2, C_{\bm w} = 1/(2 w_{\max}), D_{\bm w} = 1/(2 \eps w_{\max})$
\end{lemma}

\begin{proof}
    Let $\bm w$ be a weight vector and let 
    $\tilde{\bm w} = \bm w + \delta \cdot \ones_e,$ for some $e \in E$ and 
    $0 \neq \delta \in \R$ that is sufficiently small.
    We set %
$f_{\bm w}(\bm x) = -\iprod{\bm w, \bm x}$
and $g_{\bm w}(\bm x) = \eps \cdot \iprod{\bm w, \bm x^{\circ 2}}$,
where $\bm x^{\circ 2}$ denotes the vector obtained from $\bm x$ 
by multiplying each entry with itself.
Let $L = 2\eps w_{\max}, \sigma = \eps w_{\min}/2$. Then,
for $|\delta|$ that is sufficiently small, 
we see that $g_{\bm w}(\cdot)$
and $g_{\tilde{\bm w}}(\cdot)$
are both $L$-smooth, $\sigma$-strongly convex.

Furthermore, for every $\bm y\in K$,
since $K\sset [0, 1]^E$,
\[
    \frac1L \norm{\grad g_{\bm w}(\bm y) - \grad g_{\tilde{\bm w}}(\bm y)}_2
    \leq \frac{\varepsilon}{L} |\delta|.
\]
Thus our functions satisfy \Cref{item:c_1} with $C_{\bm w} = \nicefrac{\varepsilon}{L}$. 
    
Finally, for every $\bm y\in K$, 
let $\bm z = \prox_{L, f_{\bm w}}(\bm y)$ and $\tilde{\bm z} = \prox_{L, f_{\tilde{\bm w}}}(\bm y)$. 
A similar calculation we made in the proof of~\cref{lem:min cut prox} using the $L$-strong convexity of the proximal objective value yields the following, where $e \in E$ is the edge that differs in weight,
\begin{align*}
    L \norm{\bm z - \tilde{\bm z}}_2^2
    &\leq f_{\bm w}(\tilde{\bm z}) - f_{\bm w}(\bm z) + f_{\tilde{\bm w}}(\bm z) - f_{\tilde{\bm w}}(\tilde{\bm z}) 
    = |\delta| (\tilde z_{e} - z_{e}) 
    \leq |\delta| \norm{\bm z-\tilde{\bm z}}_2.
\end{align*}
In other words,
\[
    \norm{\prox_{L, f_{\bm w}}(\bm y) - \prox_{L, f_{\tilde{\bm w}}}(\bm y)}_2 \leq \frac1{L} |\delta|
\]
and we satisfy \Cref{item:c_2} with $D_{\bm w} = \nicefrac1{L}$.
\end{proof}

Since the regularized objective satisfies the conditions of
\Cref{asm:PGM},
we get the following result as a corollary.

\begin{corollary}\label{cor:regularized-matching-distance-opt}
    Let $\bm x^*$ be the optimal solution of \Cref{lp:relaxed-bipartite-matching} under weight vector $\bm w$ and $\tilde{\bm x}^*$
    be the optimal solution of \Cref{lp:relaxed-bipartite-matching}
    under weight vector $\tilde{\bm w} = \bm w + \delta\cdot \ones_e,$ for some $e \in E$ and 
    $0 \neq \delta \in \R$ that has sufficiently small absolute value.
    Then, $\norm{\bm x^* - \tilde{\bm x}^*}_1 \leq \frac{2\sqrt{m}}{w_{\min}} (1 + \nicefrac{1}{\varepsilon}) |\delta|.$
    
\end{corollary}

\begin{proof}
    \sloppy
    Since \Cref{lp:relaxed-bipartite-matching} satisfies
    \Cref{asm:PGM} with $L_{\bm w} = 2\eps w_{\max}, \sigma_{\bm w} = \eps w_{\min}/2, C_{\bm w} = 1/(2 w_{\max}), D_{\bm w} = 1/(2 \eps w_{\max})$, \Cref{thm:PGM} shows that 
    for the optimal solutions $\bm x^*, \tilde{\bm x}^*$,
\begin{align*}
    \norm{\bm x^* - \tilde{\bm x}^*}_1 &\leq \sqrt{m} \norm{\bm x^* - \tilde{\bm x}^*}_2\\
    &\leq \sqrt{m} \frac{L_{\bm w}(C_{\bm w} + D_{\bm w})}{\sigma_{\bm w}}|\delta|\\
    &= \frac{2\sqrt{m}}{w_{\min}} (1 + \nicefrac1{\varepsilon}) |\delta|.
    \qedhere
\end{align*}
    
\end{proof}

We use the above calculations to show the following lemma on the Lipschitz constant and approximation 
factor of our algorithm.

\begin{lemma}\label{lem:fractional-matching-approx-lipschitz}
   For any bipartite 
   graph $G=(V = U \cup R,E)$ and weight vector $\bm w \in R_{> 0}^E$
   the  
   optimal solution of \Cref{lp:relaxed-bipartite-matching} is a
   $(1+\nicefrac\eps2)$-approximation to the fractional maximum weight bipartite 
   matching of $G, \bm w.$
   Moreover, let $\bm x^*$ be the optimal solution of \Cref{lp:relaxed-bipartite-matching} under weight vector $\bm w$ and $\tilde{\bm x}^*$
    be the optimal solution of \Cref{lp:relaxed-bipartite-matching}
    under weight vector $\tilde{\bm w} = \bm w + \delta\cdot \ones_e,$ for some $e \in E$ and 
    $0 \neq \delta \in \R$ that has sufficiently small absolute value.
    Then, $\norm{\bm x^* - \tilde{\bm x}^*}_1 \leq \frac{2\sqrt{m}}{w_{\min}} (1 + \nicefrac{1}{\varepsilon}) |\delta|.$
\end{lemma}

\begin{proof}
    The Lipschitz constant directly follows from \Cref{cor:regularized-matching-distance-opt}. To get our approximation factor, notice that 
    $\sum_{uv \in E} w_{uv} x_{uv}^2 \leq \sum_{uv \in E}w_{uv} x_{uv}$ for any $x_{uv} \in [0, 1]$ and $uv \in E$. Then, 
    it holds that 
    \begin{align*}
        \sum_{uv} w_{uv} x_{uv}
        &\geq \sum_{uv} w_{uv} x_{uv} - \frac\eps2 \sum_{uv} w_{uv} x_{uv}^2 \\
        &\geq \sum_{uv} w_{uv} x_{uv}^* - \frac\eps2 \sum_{uv} w_{uv} (x_{uv}^*)^2 \\
        &\geq \left(1-\frac\eps2\right) \sum_{uv} w_{uv} x_{uv}^*.
    \end{align*}
    which gives a $(1+\nicefrac\eps2)$-approximation of the maximum fractional matching.
\end{proof}

We use the rounding scheme given by~\cite{kumabe2022lipschitz} (see Algorithm~6, Lines 4-13 in \cite{kumabe2022lipschitz})
to round our fractional solution to an integral solution. Specifically, their paper implies
the following approximation and Lipschitz constant guarantees about rounding
a fractional matching solution to an integral matching. 
}
\begin{theorem}[Sections 7.2, 8.3.5 in \cite{kumabe2022lipschitz}]\label{thm:round-bipartite}
    Given any feasible fractional matching $\bm x = \{x_{ij}\}_{i \in U, j \in V}$ satisfying the
    bipartite matching LP, 
    there exists a rounding scheme $\mcal A_\pi$ that produces a matching $\mathbf{M}_\pi = \mcal A_\pi(\bm x)$ such that %
    \begin{align*}
        \mathbb{E}\left[\sum_{e \in \mathbf{M}_\pi} w_e\right] \geq \frac{1}{2}  \sum_{ij \in E} w_{i j} x_{ij},
    \end{align*}
    where $w_e$ is the weight of edge $e$. 
    {Furthermore, it holds that under the optimal coupling $\pi, \tilde \pi\sim \mcal D$,
    \begin{align*}
        \expect_{(\pi, \tilde \pi)\sim \mcal D} \left[\|\ones_{\mathbf{M}_\pi}-\ones_{\tilde{\mathbf{M}}_{\tilde \pi}}\|_1\right]
        = O\left( \|\bm x - \tilde{\bm x}\|_1 \right). 
    \end{align*}}
\end{theorem}

Our approach is summarized in \Cref{alg:maximum bipartite matching}.

\begin{algorithm}[htp!]
    \caption{Maximum Bipartite Matching}\label{alg:maximum bipartite matching}
    \KwIn{A graph $G=(V=U\cup R,E)$, a weight vector $\bm w \in \mathbb{R}_{> 0}^E$, an approximation parameter $\eps > 0.$}
    Solve \Cref{lp:relaxed-bipartite-matching} and let $\bm x^*\in \mathbb{R}^E$ be the obtained solution \\
    $\hat{\bm x} \gets $ rounding scheme from \cite{kumabe2022lipschitz} \Comment{See  Algorithm~6, Lines 4-13 in \cite{kumabe2022lipschitz} }\\
    Return $\hat{\bm x}$
\end{algorithm}

Combining the above theorem with~\cref{lem:fractional-matching-approx-lipschitz} yields the proof of \Cref{thm:mwm-bipartite}.

\subsection{Obtaining a Fractional Lipschitz Solution}\label{sec:bmatching:fractional}
We solve maximum weight bipartite $\bm b$-matching in a similar way to our solution for maximum weight bipartite matching, with the main difference being the rounding schemes
we use for these two problems.
The LP for maximum matching with
the corresponding regularizer is given below, where $b_v$ is the constraint on the maximum number of matched neighbors 
to vertex $v$.

\begin{align}
    \begin{array}{lll}
        \text{minimize} & \displaystyle \sum_{uv\in E} -w_{uv}x_{uv} + \frac\eps2 \cdot \sum_{uv \in E} w_{uv} x_{uv}^2 \\
        \text{subject to} & \displaystyle \sum_{v \in R} x_{uv} \leq b_u & \forall u \in U\\
        & \displaystyle \sum_{u \in U} x_{uv} \leq b_v & \forall v \in R \\
        & 0 \leq x_{uv} \leq 1 & \forall u,v \in E.
    \end{array}
    \label[LP]{lp:maximum bipartite b-matching}
\end{align}

As in the maximum weight bipartite matching case, it is classically known that 
the optimum solution to the $\bm b$-matching LP equals the maximum 
solution of the integer program. Since we use the same regularizer as in~\cref{sec:bipartite} 
{and we are still optimizing
over a subset of $[0,1]^E$},
we obtain the same Lipschitz constant using the PGM framework summarized below. The 
proof of the approximation factor also follows from the proof of~\cref{lem:fractional-matching-approx-lipschitz}.

\begin{lemma}\label{lem:fractional b-matching}
For any $\varepsilon > 0,$
    the optimal solution of \Cref{lp:maximum bipartite b-matching}
    gives a $(1+\nicefrac\eps2)$-approximate maximum fractional $\bm b$-matching.
    Moreover, let $\bm x^*, \tilde{\bm x}^*$ be the optimal solutions
    of \Cref{lp:maximum bipartite b-matching} under $\bm w, \tilde{\bm w} = \bm w + \delta \ones_e,$ for some $e \in E$ and $0 \neq \delta \in \R,$ with sufficiently small absolute value.
 Then, $\norm{\bm x^* - \tilde{\bm x}^*}_1 = O\left(\frac{\sqrt{m}}{\eps w_{\min}}\right)$.
\end{lemma}

\subsection{Rounding via Multi-Item Cooperative Auction}\label{sec:bmatching:rounding}

We {extend} the fractional bipartite matching rounding scheme from \cite[Section 7.2]{kumabe2022lipschitz} (\Cref{thm:round-bipartite})
for the case of fractional bipartite $\bm b$ matchings.
We can view the following rounding algorithm as a \emph{contention resolution scheme}~\cite{vondrak2011submodular}.

\begin{algorithm}[htp!]
    \caption{Maximum Bipartite $\bm b$-Matching}\label{alg:maximum b-matching}
    \KwIn{A graph $G=(V=U\cup R,E)$, a weight vector $\bm w \in \mathbb{R}_{> 0}^E$, 
    approximation error $\eps$}
    Solve \Cref{lp:maximum bipartite b-matching} and let $\bm x^*\in \mathbb{R}^E$ be the obtained solution \\
    $p(u, v_j) \gets 0$ for all $u\in U, v\in R$ and $j\in [b_v]$
    \Comment{Bidding Phase} \\
    \For{each buyer $u\in U$} {
        Sample sellers $v(u, 1), \dots, v(u, b_u)$ \iid{} from
        \[
            \begin{cases}
                v\in N(u)\cup \set{\perp}, &\text{with probability $\frac{x_{uv}^*}{b_u}$} \\
                \perp, &\text{with remaining probability}
            \end{cases}
        \] \\
        \For{each unique seller $\perp\neq v\in \set{v(u, i), \dots, v(u, b_u)}$} {
            $u$ bids on an item $v_j\in \set{v_1, \dots, v_{b_v}}$ uniformly at random \\
            $p(u, v_j) \gets 1$ \\
        }
    }
    $q(u, v_j) \gets 0$ for all $u\in U, v\in R$ and $j\in [b_v]$
    \Comment{Accepting Phase} \\
    \For{each seller $v\in R$ and $j\in [b_v]$} {
        $v$ accepts a bid for $v_j$ uniformly at randomly from $\set{u\in U: p(u, v_j)=1}$ \\
        $q(u, v_j) \gets 1$ \\
    }
    Return $M \gets \set{uv: \exists j\in [b_v], q(u, v_j) = 1}$
\end{algorithm}

We interpret the rounding algorithm as a cooperative auction between buyers $u\in U$ and sellers $v\in R$,
where each buyer $u$ can buy at most $b_u$ unique items
from its neighboring sellers $N(u)$
and each seller $v$ has $b_v$ identical items to sell.
Each transaction between $u, v$ yields a gain of $w_{uv}$ 
and buyers and sellers attempt to maximize the total gain.
In the bidding phase,
each buyer $u\in U$ selects a set of at most $b_u$ sellers
and bids on one item from each selected seller.
In the accepting phase,
each seller $v\in R$ accepts a bid uniformly
at random for each item $v_j, j\in [b_v]$
whose set of bidders is non-empty.

\begin{lemma}\label{lem:maximum bipartite b-matching bidding}
    Let edge weights $\tilde{\bm w}\in \R_{>0}^E$ be obtained by perturbing $\bm w\in \R_{>0}^E$ at a single coordinate by some sufficiently
    small $\delta > 0$.
    Let $\bm x^*, \tilde{\bm x}^*\in [0, 1]^E$ be the optimal solutions to \Cref{lp:maximum bipartite b-matching}
    resulting from $\bm w, \tilde{\bm w}$, respectively.
    Let $p(u, v_j) \coloneqq \ones\set{\text{$u$ bids on $v_j$}}$ be the indicator variable of bids obtained when \Cref{alg:maximum b-matching} is executed on $\bm w$
    and similarly $\tilde p(u, v_j)$ when it is executed on $\tilde{\bm w}$.
    Then the following holds:
    \begin{enumerate}[(a)]
        \item $\E\left[ \sum_{u\in U, v\in N(u), j\in [b_v]} w_{uv} p(u, v_j) \right] \geq (1-\nicefrac1e ) \OPTLP$
        \item $\EMD_d(\bm p, \tilde{\bm p}) = \norm{\bm x^* - \tilde{\bm x}^*}_1$
    \end{enumerate}
\end{lemma}

\begin{pf}[\Cref{lem:maximum bipartite b-matching bidding}]
    We prove each statement separately.

    \paragraph{Proof of (a).}
    It suffices to compute the marginal probability $\E[p(u, v_j)]$ of a buyer $u\in U$ bidding on the item $v_j$ from a seller $v\in N(u)$
    and lower bound it by $(1-\nicefrac1e) \frac{x_{uv}^*}{b_v}$.
    Then the result follows from the linearity of expectation.
    In fact,
    we show that the probability that $u$ bids on some item from $v\in N(v)$ is at least $(1-\nicefrac1e) x_{uv}^*$.
    Since the specific item is chosen uniformly at random,
    the result follows.
    
    Consider the equivalent experiment for generating a set of at most $b_u$ sellers as follows:
    For $i=1, 2, \dots, b_u$,
    we sample a seller $v\in N(u)$ with probability $\frac{x_{uv}^*}{b_v}$ and do nothing with the remaining probability.
    {Notice that since $\sum_{v \in N(u)} x_{uv}^* \leq b_u$, it is indeed the case
    that $\sum_{v\in N(u)} \frac{x_{uv}^*}{b_u} \leq 1.$}
    If we have not yet chosen $v$ in a prior iteration $i' < i$,
    then add $v$ to the set of sellers.
    The event that the neighbor $v\in N(u)$ is drawn at least once is the complement of the event that it is never picked after $b_u$ draws.
    Hence
    \begin{align*}
        \Pr[\text{$u$ selects $v$}]
        &= 1-\left( 1-\frac{x_{uv}^*}{b_u} \right)^{b_u} \\
        &\geq 1-e^{-x_{uv}^*} \tag{$1+z\leq e^z$} \\
        &\geq \left( 1 - \frac1e \right) x_{uv}^*. \tag{concavity}
    \end{align*}
    The last inequality can be verified by noting that we have equality at $x_{uv}^* = 0, 1$
    and that $1-e^{-x_{uv}^*}$ is a concave function of $x_{uv}^*$.

    \paragraph{Proof of (b).}
    It suffices to show that for each $u\in U$,
    under the optimal coupling $(\bm p, \tilde{\bm p})\sim \mcal D$,
    \[
        \E_{(\bm p, \tilde{\bm p})\sim \mcal D}\left[ \sum_{v\in N(u), j\in [b_v]} \abs{p(u, v_j) - \tilde p(u, v_j)} \right]
        \leq \sum_{v\in N(u)} \abs{x_{uv}^* - \tilde x_{uv}^*}.
    \]
    Then the conclusion follows by summing over all $u\in U$ and the linearity of expectation.

    First,
    we analyze
    \begin{align*}
        &\sum_{v\in N(u)} \abs*{\sum_{j\in [b_v]} p(u, v_j) - \sum_{j\in [b_v]} \tilde p(u, v_j)} \\
        &= \sum_{v\in N(u)} \abs*{\ones\set{\text{$u$ selects seller $v$ under weights $w$}} -  \ones\set{\text{$u$ selects seller $v$ under weights $\tilde{\bm w}$}}}.
    \end{align*}
    This is precisely the cardinality of the symmetric difference of the set of sellers selected by $u$.
    For the final set of sellers to differ at some $v\in N(u)$,
    the number of times $v$ was selected in the multiset of selected sellers must have differed.
    The total variation distance of a single selected seller is $\sum_{v\in N(u)}\abs*{\frac{x_{uv}^*}{b_u} - \frac{\tilde x_{uv}^*}{b_u}}$.
    Hence the expected cardinality of the symmetric difference of the multiset within the optimal coupling is $\sum_{v\in N(u)}\abs{x_{uv}^* - \tilde x_{uv}^*}$
    and we have the upper bound
    \begin{align*}
        \sum_{v\in N(u)} \abs*{\sum_{j\in [b_v]} p(u, v_j) - \sum_{j\in [b_v]} \tilde p(u, v_j)}
        &\leq \sum_{v\in N(u)}\abs{x_{uv}^* - \tilde x_{uv}^*}.
    \end{align*}
    
    Now, 
    {if we consider a coupling
    where we are sharing the internal randomness},
    $p(u, v_j), \tilde p(u, v_j)$ can only differ if $u$ was inconsistent in including $v$ in the set of sellers.
    In other words,
    for any $u\in U, v\in N(u)$,
    we have
    \[
        \sum_{j\in [b_v]} \abs{p(u, v_j) - \tilde p(u, v_j)}
        = \abs*{\sum_{j\in [b_v]} p(u, v_j) - \sum_{j\in b_v} \tilde p(u, v_j)}.
    \]
    It follows that
    \begin{align*}
        &\E_{(\bm p, \tilde{\bm p})\sim \mcal D}\left[ \sum_{v\in N(u), j\in [b_v]} \abs{p(u, v_j) - \tilde p(u, v_j)} \right] \\
        &= \E_{(\bm p, \tilde{\bm p}\sim \mcal D)}\left[ \sum_{v\in N(u)} \abs*{\sum_{j\in [b_v]} p(u, v_j) - \sum_{j\in [b_v]} \tilde p(u, v_j)} \right] \\
        &\leq \sum_{v\in N(u)} \abs{x_{uv}^* - \tilde x_{uv}^*}.
        \qedhere
    \end{align*}
\end{pf}

For the accepting phase,
we observe that it is identical to the case of bipartite matching
since each item is sold to its set of bidders uniformly at random.
The approximation and Lipschitz guarantees follow almost immediately
from the guarantees for bipartite matching \cite{kumabe2022lipschitz}.
We write $q(v_j, u)$ to be the indicator variable that item $v_j$ is sold to $u\in N(v)$.

For the approximation guarantees,
an important observation from \cite{kumabe2022lipschitz} is that the number of bids for an item in \Cref{alg:maximum b-matching}
follows a \emph{Poisson Binomial} distribution.
Recall $\PoBin(y_1, \dots, y_n)$ is the number of heads from flipping $n$ independent coins
with biases $y_1, \dots, y_n$.
The following lemma ensures that we lose at most half the bids in expectation within the acceptance phase.

\begin{lemma}[Lemmas 7.10, 7.11, 7.12 in \cite{kumabe2022lipschitz}]\label{lem:maximum bipartite b-matching accepting approximation}
    Suppose $u\in U$ bids on an item $v_j$ for $v\in N(u), j\in [b_v]$
    and the total number of other bids for $v_j$ follows a Poisson Binomial distribution $\PoBin(y_{u'v_j})_{u\neq u'\in N(v)}$
    whose parameters $y_{u'v_j}\in [0, 1], u\neq u' \in N(v)$ satisfy $\sum_{u\neq u'\in N(v)} y_{u'v_j} \leq 1$.
    Then, if $v_j$ is sold to a bidder uniformly at random,
    $v_j$ is sold to $u$ with probability at least $\nicefrac12$.
\end{lemma}

For $p(u, v_j) \coloneqq \ones\set{\text{$u$ bids on $v_j$}}$ 
and $q(u, v_j) \coloneqq \ones\set{\text{$v_j$ is sold to $u$}}$ from \Cref{alg:maximum b-matching},
\Cref{lem:maximum bipartite b-matching accepting approximation} states that
\[
    \Pr[q(u, v_j) = 1 \mid p(u, v_j) = 1]
    \geq \frac12.
\]
{This is because
for each item $v_j$ we flip $\card{N(v)}$
different independent coins in total 
and for every $u\in N(v)$,
the coin corresponding to $u$
has bias at most $x_{uv}/b_v$. Notice
that for any $u\in N(v)$,
\begin{align*}
  \sum_{u\neq u'\in N(v)} \frac{x_{uv}}{b_v} \leq 1,
\end{align*}
so the conditions of \Cref{lem:maximum bipartite b-matching accepting approximation} are satisfied.
}

For the Lipschitz guarantees,
the following lemma reduces bounding the Lipschitz constant of the final $\bm b$-matching
to bounding the Lipschitz constant of the bids themselves.
\begin{lemma}[Lemma 7.18 in \cite{kumabe2022lipschitz}]\label{lem:maximum bipartite b-matching accepting Lipschitz}
    Let $p(u, v_j), \tilde p(u, v_j), u\in U, v\in N(u), j\in [b_v]$ be two realizations of indicator variables
    that a buyer $u\in U$ 
    bids on some item $v_j, j\in [b_v]$
    from a seller $v\in N(u)$.
    Suppose each item is sold uniformly at random to one of its bidders
    (if there are any),
    and let $q(u, v_j), \tilde q(u, v_j), u\in U, v\in N(u), j\in [b_v]$ be the indicator variables that $v_j$ is sold to $u$
    conditioned on $\bm p, \tilde{\bm p}$,
    respectively.
    Then under the conditional optimal coupling $(\bm q, \tilde{\bm q})\sim \mcal D_{p, \tilde p}$,
    \[
        \E_{(\bm q, \tilde{\bm q})\sim \mcal D_{\bm p, \tilde{\bm p}}}\left[ \norm{\bm q-\tilde{\bm q}}_1 \mid \bm p, \tilde{\bm p} \right]
        \leq 2\norm{p-\tilde p}_1.
    \]
\end{lemma}
Let edge weights $\tilde{\bm w}\in \R_{>0}^E$ be obtained by perturbing $\bm w\in \R_{>0}^E$ at a single coordinate by some small $\delta > 0$.
Let $p(u, v_j) \coloneqq \ones\set{\text{$u$ bids on $v_j$}}$ to be the indicator variable that $u$ bids on the item $v_j$ from seller $v\in N(u)$ when \Cref{alg:maximum b-matching} is executed on $\bm w$
and similarly $\tilde p(u, v_j)$ when executed on $\tilde{\bm w}$.
Define $q(u, v_j) \coloneqq \ones\set{\text{$v_j$ is sold to $u$}}$ to be the indicator variable that the item $v_j$ is sold to $u$
when \Cref{alg:maximum b-matching} is executed on $\bm w$
and similarly $\tilde q(u, v_j)$ when executed on $\tilde{\bm w}$.
Fix a realization $\bm p, \tilde{\bm p}$.
\Cref{lem:maximum bipartite b-matching accepting Lipschitz} states that under the conditional optimal coupling $(\bm q, \tilde{\bm q})\sim \mcal D_{\bm p, \tilde{\bm p}}$,
\[
    \E_{(\bm q, \tilde{\bm q})\sim \mcal D_{\bm p, \tilde{\bm p}}}\left[ \norm{\bm q-\tilde{\bm q}}_1 \mid \bm p, \tilde{\bm p} \right]
    \leq 2\norm{\bm p-\tilde{\bm p}}_1.
\]

We are now ready to prove \Cref{thm:maximum bipartite b-matching}.
\begin{pf}[\Cref{thm:maximum bipartite b-matching}]
    We show the approximation and Lipschitz guarantees separately.

    \paragraph{Approximation Guarantee.}
    Let $\bm x$ denote the optimal solution of \Cref{lp:maximum bipartite b-matching}
    and $M$ denote an output $\bm b$-matching.
    By the linearity of expectation,
    \begin{align*}
        &\E\left[ \sum_{e\in E} w_e \ones\set{e\in M} \right] \\
        &= \sum_{u\in U} \sum_{v\in N(v)} \sum_{j\in [b_v]} w_{uv} \Pr[q(u, v_j) = 1\mid p(u, v_j)=1]\cdot \Pr[p(u, v_j)=1] \\
        &\geq \frac12 \sum_{u\in U} \sum_{v\in N(v)} \sum_{j\in [b_v]} w_{uv} \Pr[p(u, v_j)=1] \tag{\Cref{lem:maximum bipartite b-matching accepting approximation}} \\
        &\geq \frac12 \left( 1-\frac1e \right) \OPTLP. \tag{\Cref{lem:maximum bipartite b-matching bidding}}
    \end{align*}
    Since $\OPTLP \geq (1-\varepsilon) \OPT$ by \Cref{lem:fractional b-matching},
    we have shown the approximation guarantees as desired.

    \paragraph{Pointwise Lipschitz Constant.}
    Let edge weights $\tilde{\bm w}\in \R_{>0}^E$ be obtained by perturbing $\bm w\in \R_{>0}^E$ at a single coordinate by some small $\delta > 0$.
    Let $p(u, v_j) \coloneqq \ones\set{\text{$u$ bids on $v_j$}}$, 
    $q(u, v_j) \coloneqq \ones\set{\text{$v_j$ is sold to $u$}}$
    be the indicator variables that $u$ bids on the item $v_j$ from seller $v\in N(u)$ 
    and the item $v_j$ is sold to $u$,
    respectively,
    when \Cref{alg:maximum b-matching} is executed on $\bm w$.
    Similarly,
    define $\tilde p(u, v_j)$, $\tilde q(u, v_j)$ when the algorithm is executed on $\tilde{\bm w}$.

    Consider the coupling $(\bm q, \tilde{\bm q})\sim \mcal D_{\bm p, \tilde{\bm p}}\circ \mcal D$ 
    obtained by sampling from the optimal coupling $(\bm p, \tilde{\bm p})\sim \mcal D$ (cf. \Cref{lem:maximum bipartite b-matching bidding})
    and then sampling from the conditional optimal coupling $(\bm q, \tilde{\bm q})\sim \mcal D_{\bm p, \tilde{\bm p}}$ (cf. \Cref{lem:maximum bipartite b-matching accepting Lipschitz}).
    We have
    \begin{align*}
        &\E_{(\bm q, \tilde{\bm q})\sim \mcal D_{\bm p, \tilde{\bm p}}\circ \mcal D} \left[ \norm{\bm q-\tilde{\bm q}}_1 \right] \\
        &= \E_{(\bm p, \tilde{\bm p})\sim \mcal D} \left[ \E_{(\bm q, \tilde{\bm q})\sim \mcal D_{\bm p, \tilde{\bm p}}} \left[ \norm{\bm q-\tilde{\bm q}}_1 \mid \bm p, \tilde{\bm p} \right] \right] \\
        &\leq \E_{(\bm p, \tilde{\bm p})\sim \mcal D} \left[ 2\norm{\bm p-\tilde{\bm p}}_1 \right] \tag{\Cref{lem:maximum bipartite b-matching accepting Lipschitz}} \\
        &= 2\norm{\bm x-\tilde{\bm x}}_1 \tag{\Cref{lem:maximum bipartite b-matching bidding}} \\
        &= O\left(\frac{\sqrt{m}}{\eps w_{\min}}\right). \tag{\Cref{lem:fractional b-matching}}
    \end{align*}
\end{pf}

\section{Packing Integer Programs}\label{sec:packing IP}
In this section,
we further demonstrate the applicability of our techniques by designing a Lipschitz algorithm for packing integer programs.
Let $\bm w\in \R_{> 0}^m$, 
$A\in [0, 1]^{p\times m}$,
and $\bm b\in \R_{\geq 1}^p$.
Consider a general \emph{packing integer program (PIP)} \cite{raghavan1988probabilistic} of the form
\begin{align*}
    \begin{array}{ll}
        \text{maximize} & \displaystyle \sum_{i\in [m]} w_i x_i \\
        \text{subject to} & A\bm x \leq \bm b \\
        & \bm x \in \set{0, 1}^m
    \end{array}
\end{align*}

This model generalizes many combinatorial optimization problems such as
maximum independent set,
set packing,
knapsack,
maximum $\bm b$-matching,
and \emph{maximum $\bm b$-matching on hypergraphs} \cite{tardos2005algorithm}:
Given a hypergraph $(V, E)$ where $E\sset 2^V$ is a collection of $m$ subsets of $V=[p]$, positive integers $b_j\geq 1$ for each $j\in [p]$, and a non-negative weight $w_i$ for each edge $e_i\in E\;(i\in [m])$,
the goal is to choose a maximum weight collection of edges
so that each element $j\in V$ is contained in at most $b_j$ of the chosen edges.
A popular PIP formulation is as follows.
\begin{align*}
    \begin{array}{lll}
        \text{maximize} & \bm w^\top \bm x \\
        \text{subject to} & \displaystyle \sum_{i: j\in e_i} x_i \leq b_j &\forall j\in [p] \\
        & x_i \in \set{0, 1}^m
        \end{array}
\end{align*}

The guarantees we 
get are summarized in \Cref{thm:PIP}, which is the main result 
of this section.

\begin{theorem}\label{thm:PIP}
    Let $\bm w\in \R_{> 0}^m$, 
    $A\in [0, 1]^{p\times m}$,
    and $\bm b\in \R_{\geq 1}^p$ be an instance of a PIP
    satisfying $B \coloneqq \min_j b_j \geq 1$.
    For any $c  \geq 1$ \Cref{alg:PIP}
    is an expected $O((cp)^{1/B})$-approximation algorithm
    that returns a feasible solution with probability at least $1-\nicefrac{1}{c}$.
    Moreover,
    its pointwise Lipschitz constant is 
    \[
        O\left( \frac{\sqrt{m}}{w_{\min}\cdot (cp)^{1/B}} \right).
    \]
\end{theorem}

{
Following our recipe from previous sections,
we first obtain a fractional Lipschitz solution in \Cref{sec:pip:fractional}.
This can be rounded using a simple independent rounding scheme we describe in \Cref{sec:pip:rounding}.
}

\subsection{Obtaining a Fractional Lipschitz Solution}\label{sec:pip:fractional}
Returning to the case of a general PIP,
we take the straightforward convex relaxation with an additional strongly convex regularizer.
\begin{align}\label[LP]{lp:regularized PIP}
    \begin{array}{ll}
        \text{minimize} &\displaystyle -\sum_{i=1}^m w_ix_i + \frac{1}2\sum_{i=1}^m w_i x_i^2 \\
        \text{subject to} & A\bm x \leq \bm b \\
        & \bm x \in [0, 1]^m
    \end{array}
\end{align}

Our approach to solve this problem in a pointwise Lipschitz continuous
manner is to first solve the regularized \Cref{lp:regularized PIP} and then round the solution $\bm x^* \in [0,1]$ to a solution $\bm y^* \in \{0,1\}$
as
follows: we fix some sufficiently large number $\gamma > 0$ and
for each $i \in [m]$ we set $y^*_i = 1$ with probability
$x^*_i/\gamma$ and $y^*_i = 0$ with probability $1-x^*_i/\gamma$.
This is described in \Cref{alg:PIP}.

\begin{algorithm}[htp!]
    \caption{Packing Integer Program}\label{alg:PIP}
    \KwIn{A weight vector $\bm w \in \R^m_{> 0},$ a constraint
    matrix $A \in [0,1]^{p \cross m},$ a constraint vector $\bm b \in \R^p_{\geq 1}$, an approximation parameter $c \geq 1.$}
    Solve \Cref{lp:regularized PIP} and let $\bm x^*\in [0,1]^m$ be the obtained solution \\
    $B \gets \min_{i \in [p]} b_i$ \\
    $\gamma \gets e (cp)^{1/B}$ \\
    \For{each variable $i \in [m]$} {
        $y^*_i \gets \mathrm{Bernoulli}\left(\nicefrac{x^*_i}{\gamma}\right)$
    }
    Return $\bm y^*$
\end{algorithm}

First, we prove the following guarantees about the fractional solutions
obtained by solving \Cref{lp:regularized PIP}.

\begin{lemma}\label{lem:fractional PIP}
    For any weight vector $\bm w \in \R^m_{\geq 0},$ constraint matrix
    $A \in [0,1]^{p \cross m},$ and constraint vector $\bm b \in \R^{p}_{\geq 1}$,
    the optimal solution $\bm x^*$ to \Cref{lp:regularized PIP} is a $2$-approximation
    to the optimal solution. Moreover, for $\tilde{\bm w} = \bm w + \delta \ones_j,$ for some $j \in [m]$ and $0 \neq \delta \in \R$ with sufficiently
    small absolute value it holds that $\norm{\bm x^* - \tilde{\bm x}^*}_1 \leq O\left( \frac{\sqrt{m}\abs{\delta}}{w_{\min}} \right),$ where $\tilde{\bm x}^*$ is the optimal solution of \Cref{lp:regularized PIP} under $\tilde{\bm w}, A, \bm b.$
\end{lemma}

\begin{proof}
    {
    We will prove the two results separately. We start with the bound
    on $\norm{\bm x^* - \tilde{\bm x}^*}_2.$
    Our analysis will utilize the bound obtained through the PGM framework
    from \Cref{sec:framework} so 
    we need to check that \Cref{asm:PGM} is 
satisfied.
Let $\tilde{\bm w} = \bm w + \delta \ones_j,$ for 
some $j \in [m]$ and $0 \neq \delta \in \R$ with
sufficiently small absolute value. Let also $f_{\bm w}(\bm x) = -\sum_{i \in [m]} w_i x_i$ 
and similarly for $f_{\tilde{\bm w}}(\cdot)$. Moreover,
let $g_{\bm w}(\bm x) = \nicefrac{1}{2} \sum_{i \in [m]} w_i x_i^2$
and similarly for $g_{\tilde{\bm w}}(\cdot).$ It follows
immediately that $f_{{\bm w}}(\cdot),f_{\tilde{\bm w}}(\cdot)$ are both convex and continuous. Moreover,
it can be seen that for sufficiently small $|\delta|$,
both of them are $\Omega(w_{\min})$-strongly convex
and $O(w_{\max})$-smooth, where $w_{\min} = \min_{i \in [m]} w_i, w_{\max} = \max_{i \in [m]} w_i.$ Let $\sigma, L$
denote the strong convexity, smoothness parameters. Since
the feasible region is a subset of $[0,1]^m$ we get that
\[
    \frac{1}{L} \norm{\nabla g_{\bm w}(\bm x) - \nabla g_{\tilde{\bm w}(\bm x)}}_2 \leq O\left(\frac{|\delta|}{L}\right).
\]
Thus, $C_{\bm w} = 1/L$. Lastly, a similar calculation we made in the proof of~\cref{lem:min cut prox} using the $L$-strong convexity of the proximal objective value yields the following, where $j \in [m]$ is the element whose weight differs,
\begin{align*}
    L \norm{\bm z - \tilde{\bm z}}_2^2
    &\leq f_{\bm w}(\tilde{\bm z}) - f_{\bm w}(\bm z) + f_{\tilde{\bm w}}(\bm z) - f_{\tilde{\bm w}}(\tilde{\bm z}) 
    = |\delta| (\tilde z_{j} - z_{j}) 
    \leq |\delta| \norm{\bm z-\tilde{\bm z}}_2.
\end{align*}
In other words,
\[
    \norm{\prox_{L, f_{\bm w}}(\bm y) - \prox_{L, f_{\tilde{\bm w}}}(\bm y)}_2 \leq \frac1{L} |\delta|
\]
and we satisfy \Cref{item:c_2} with $D_{\bm w} = \nicefrac1{L}$.
}
Then, \Cref{thm:PGM} shows that the fractional optimal solutions $\bm x^*, \tilde{\bm x}^*$ satisfy
\[
    \norm{\bm x^*-\tilde{\bm x}^*}_2
    = O\left( \frac{\abs{\delta}}{w_{\min}} \right).
\]
The final bound follows from $\norm{\bm x^*-\tilde{\bm x}^*}_1 \leq \sqrt{m} \norm{\bm x^*-\tilde{\bm x}^*}_2,$ which is a direct
application of the Cauchy-Schwartz inequality.

Next, we shift our attention to the approximation guarantees
of \Cref{lp:regularized PIP}. First, notice that
since every feasible solution satisfies 
$x_i \in [0,1], \forall i \in [m],$ we have that $\sum_{i \in [m]} w_i x_j^2 \leq \sum_{i \in [m]} w_i x_i^2.$ Thus, we have that
\begin{align*}
    - \sum_{i \in [m]} w_i x^*_i &\leq - \sum_{i \in [m]} w_i x^*_i + \frac{1}{2} \sum_{i \in [m]} w_i (x^*_i)^2 \\
    &\leq  - \sum_{i \in [m]} w_i x_i + \frac{1}{2} \sum_{i \in [m]} w_i x^2_i \\
    &\leq  - \sum_{i \in [m]} w_i x_i + \frac{1}{2} \sum_{i \in [m]} w_i x_i \\
    &\leq -\frac{1}{2} \sum_{i \in [m]} w_i x_i .
\end{align*}

\end{proof}

\subsection{Independent Rounding}\label{sec:pip:rounding}
Having established the guarantees for the fractional solution,
we shift our attention to establishing
guarantees about the \emph{integer} solution of \Cref{alg:PIP}.
First, let us elaborate a bit on the idea behind 
the rounding scheme we use. 
Suppose we solve for the optimal solution $\bm x^*$ of \Cref{lp:regularized PIP} and perform the simple rounding scheme where we independently round each $x_i^*$ to $1$ with probability $x_i^*$ and $0$ with probability $1-x_i^*$.
We preserve the objective in expectation
and also satisfy each constraint in expectation.
However, this is not sufficient since we need to satisfy all
the constraints with high probability.
To do that,
we adapt the ideas from \cite{srinivasan1999approximation, raghavan1988probabilistic} to our setting.
Since all the inequalities are of the form $\leq$,
we can \emph{scale} down the variable $\bm x^*$ by some constant $\gamma > 1$ to obtain $x_i' = x_i^*/\gamma$,
and then perform independent rounding. 

We are now ready to prove \Cref{thm:PIP}
\begin{proof}[Proof of \Cref{thm:PIP}]
    We will prove the approximation guarantee, the feasibility guarantee, and
    the Lipschitz constant
    of our algorithm separately. We start with the approximation guarantee.
    Let $\bm x^*$ denote the solution to \Cref{lp:regularized PIP} under
    $\bm w, A, \bm b.$
    Let $B = \min_{i\in [p]} b_i,$ as defined in \Cref{alg:PIP} and notice
    that $B \geq 1.$ By definition, for each $i \in [m]$ it holds
    that $\E[y^*_i] = x^*_i/\gamma.$ Thus, by linearity of 
    expectation we have that $\E[\sum_{i\in[m]} w_i y_i^*] = \nicefrac{1}{\gamma}\E[\sum_{i\in[m]} w_i x_i^*].$ Finally, recall that $\gamma = e(cp)^{1/B}$ and that $\bm x^*$ gives a $2$-approximation to the objective
    function by \Cref{lem:fractional PIP}. This concludes the approximation guarantee of our algorithm.

    Let us now shift to the feasibility guarantee. Note that $\expect[{(A\bm y)_j}] \leq b_j/\gamma$.
By a multiplicative Chernoff bound,
\[
    \Pr[(A\bm y)_j \geq b_j]
    \leq \left( \frac{e}{\gamma} \right)^{b_j}
    \leq \left( \frac{e}{\gamma} \right)^B
    = \frac1{cp}.
\]
By a union bound over the $p$ constraints,
we satisfy all constraints with constant probability $1-\nicefrac1c$.

Lastly, we focus on the pointwise Lipschitz guarantee. Let 
$\tilde{\bm w} = \bm w + \delta \ones_j,$ for some $j \in [m]$ and $0 \neq \delta \in \R$ with sufficiently small absolute value. Let $\tilde{\bm x}^*$ be the solution
to \Cref{lp:regularized PIP} on $\tilde{\bm w}, A, \bm b.$ Then, 
we know by \Cref{lem:fractional PIP} that $\norm{\bm x^* - \tilde{\bm x}^*}_1 \leq O\left(\frac{\sqrt{m}\abs{\delta}}{w_{\min}}\right),$ which implies
that $\norm{\nicefrac{\bm x^*}{\gamma} - \nicefrac{\tilde{\bm x}^*}{\gamma}}_1 \leq O\left(\frac{\sqrt{m}\abs{\delta}}{\gamma w_{\min}}\right).$ Consider
the following coupling of the two executions of the algorithm that is based 
on shared randomness.
We can consider rounding of each variable as drawing an independent $\tau_i\sim U[0, 1]$ and rounding down if $\tau\leq x_i^*/\gamma$
and vice versa.
The coupling between the perturbed instances is simply the identity coupling on $\tau_i$'s. Under this coupling, we have that
\[
    \E\left[\abs{\bm y^* \Delta \tilde{\bm y}^*}\right] = \sum_{i \in [m]} \abs{\nicefrac{x^*_i}{\gamma} - \nicefrac{\tilde{x}_i^*}{\gamma}} = \norm{\nicefrac{\bm x^*}{\gamma} - \nicefrac{\tilde{\bm x}^*}{\gamma}}_1 \leq O\left(\frac{\sqrt{m}\abs{\delta}}{\gamma w_{\min}}\right).
\]
This concludes the proof.
\end{proof}

\section{Pointwise Lipschitz Continuity under Shared Randomness}\label{sec:shared-randomness}

In this section, we discuss pointwise Lipschitz continuity under shared randomness, which is formally defined below.
\begin{definition}[\cite{kumabe2022lipschitz}]\label{def:pointwise lipschitz shared randomness}
    Suppose the (randomized) algorithm $\mcal A(G, \bm w)$ outputs subsets of some universe $\mfrak S$.
    Let $d: 2^{\mfrak S}\times 2^{\mfrak S}\to \R_+$ be a metric on $2^{\mfrak S}$.
    We say that $A$ has \emph{pointwise Lipschitz constant $c_{\bm w}$ under shared randomness with respect to $d$} if
    \[
        \limsup_{\tilde{\bm w}\to \bm w} \E_\pi \left[ \frac{d(\mcal A_\pi(G, \bm w) - \mcal A_\pi(G, \tilde{\bm w}))}{\norm{\bm w-\tilde{\bm w}}_1} \right]
        \leq c_{\bm w}.
    \]
\end{definition}

\cite[Section 8]{kumabe2022lipschitz} demonstrated how to implement a pointwise Lipschitz algorithm using shared randomness at a constant factor blowup of the Lipschitz constant when the original algorithm employs internal randomness to sample from a uniform continuous distribution or a discrete distribution.
Specifically, Let \Call{Sample}{} be a sampling process that takes two real values $a,b \in \mathbb{R}$ and a vector $\bm p \in [0,1]^{\mathbb{Z}_{\geq 0}}$ and outputs a real value.
For $c \geq 1$, we say that \Call{Sample}{} is $c$-\emph{stable} for a pair  of functions $(l,r)$ with $l,r: \mathbb{R}_{\geq 0}^E \to \mathbb{R}$ if (i) for any $\bm w \in \mathbb{R}_{\geq 0}^E$, $\Call{Sample}{l(\bm w),r(\bm w),p}$ is uniformly distributed over $[l(\bm w),r(\bm w)]$ when $\bm p$ follows the uniform distribution over $[0,1]^{\mathbb{Z}_{\geq 0}}$, and (ii) for any $\bm w,\tilde{\bm w} \in \mathbb{R}_{\geq 0}^E$, we have
\begin{align*}
    &\E_{\bm p\sim \mathcal{U}\left([0,1]^{\mathbb{Z}_{\geq 0}}\right)}\left[\TV\left(\Call{Sample}{l(\bm w),r(\bm w),\bm p},\Call{Sample}{l(\tilde{\bm w}),r(\tilde{\bm w}),\bm p}\right)\right]\\
    &\leq c\cdot\TV\left(\mathcal{U}([l(\bm w),r(\bm w)]),\mathcal{U}([l(\tilde{\bm w}),r(\tilde{\bm w})])\right)
\end{align*}
holds.
Then, they gave the following useful sampling processes.
\begin{lemma}[\rm{\cite{kumabe2023lipschitz}}]\label{lem:shared_const}
Let $l$ and $r$ be constant functions. Then, there is a $1$-stable sampling process for $(l,r)$.
\end{lemma}
\begin{lemma}[\rm{\cite{kumabe2023lipschitz}}]\label{lem:shared_ratio}
Let $c > 1$ and suppose $r(w)=cl(w)$ holds for all $w$.
Then, there is a $(1+c)$-stable sampling process for $(l,r)$.
\end{lemma}
Using these sampling processes, it is easy to see that most of the sampling processes used in our rounding schemes can be transformed for the shared randomness setting such that the pointwise Lipschitz constant blows up by a constant factor.
See~\cite{kumabe2023lipschitz,kumabe2022lipschitz} for more details.

The only exception is the exponential mechanism used in \Cref{alg:min-s-t-cut-additive}.
Recall that the exponential mechanism takes a vector $\bm x\in \mathbb{R}^n$, generates a vector $\bm p \in \mathbb{R}^n$ such that $p_i = \exp(-\eta x_i) / \sum_{j \in [n]} \exp(-\eta x_j)$ for $i \in [n]$, and sample an index $i$ with probability $p_i$.
If we can bound $\|\bm p - \tilde{\bm p}\|_1$ by using $\|\bm x - \tilde{\bm x}\|_1$, then we can bound the pointwise Lipschitz constant under shared randomness by using \cite[Lemma 8.4]{kumabe2023lipschitz}.
\begin{lemma}
    We have
    \[
        \|\bm p - \tilde{\bm p}\|_1     
        \leq O(\eta \cdot \|\bm x - \tilde{\bm x}\|_1).     
    \]
\end{lemma}
\begin{proof}
    We assume that there exists $\hat \imath$ and $\delta \neq 0$ such that
    \[
        \tilde x_i = \begin{cases}
            x_i + \delta & \text{if } i = \hat \imath \\
            x_i & \text{otherwise}.
        \end{cases}
    \]
    If we can show the claim with this assumption, then we can show the claim for general $\bm x$ and $\tilde{\bm x}$ by triangle inequality.
    Without loss of generality, we can assume that $\delta > 0$.

    Let $Z = \sum_{i \in [n]} \exp(-\eta x_i)$ and $\tilde Z = \sum_{i \in [n]} \exp(-\eta \tilde x_i)$.
    We note that $Z \geq \tilde{Z}$.
    Then, we have
    \begin{align*}
        & \|\bm p - \tilde{\bm p}\|_1
        = 
        \sum_{i=1}^n |p_i - \tilde p_i|
        = 
        \sum_{i=1}^n \left|\frac{\exp(-\eta x_i)}{Z} - \frac{\exp(-\eta \tilde x_i)}{\tilde Z}\right| \\
        & = 
        \sum_{i=1}^n \left|\frac{\exp(-\eta x_i)}{Z} - \left(1 - \frac{\tilde Z - Z}{\tilde Z}\right) \frac{\exp(-\eta \tilde x_i)}{Z}\right| \\
        & \leq 
        \sum_{i=1}^n \left|\frac{\exp(-\eta x_i)}{Z} - \frac{\exp(-\eta \tilde x_i)}{Z}\right| + \sum_{i=1}^n \left|\left(\frac{\tilde Z - Z}{\tilde Z}\right) \frac{\exp(-\eta \tilde x_i)}{Z}\right| \\
        & \leq 
        \frac{1}{Z} \sum_{i=1}^n \left|\exp(-\eta x_i) - \exp(-\eta \tilde x_i)\right| + \frac{|\tilde Z - Z|}{Z} \\
        & =
        \frac{1}{Z} \sum_{i=1}^n \left(\exp(-\eta x_i) - \exp(-\eta \tilde x_i)\right) + \frac{Z - \tilde{Z}}{Z}.
    \end{align*}
    As the second term is smaller than the first term, we focus on bounding the first term.
    We have
    \begin{align*}
       & \sum_{i=1}^n \left(\exp(-\eta x_i) - \exp(-\eta \tilde x_i)\right)
       = \exp(-\eta x_{\hat \imath}) - \exp(-\eta (x_{\hat \imath}+\delta))
       = \exp(-\eta x_{\hat \imath}) \left(1 - \exp(-\eta \delta)\right) \\
       & \leq \exp(-\eta x_{\hat \imath}) \cdot \eta \delta. \tag{by $1-\exp(-x) \leq x$.}
    \end{align*}
    Hence, we have $\|\bm p - \tilde{\bm p}\|_1 = O(\eta \|\bm x - \tilde{\bm x}\|_1)$.
\end{proof}

\section{Conclusion}\label{sec:conclusion}
In this work, we developed a principled way to design
pointwise Lipschitz continuous graph algorithms.
We demonstrated that our framework,
paired with stable rounding schemes, can give
algorithms with tight trade-offs between their Lipschitz
constant and approximation guarantees, as demonstrated
by our set of results regarding minimum vertex $S$-$T$ cut.
{In addition to its broad applicability,
our framework yields better bounds than possible compared to previous techniques.}

One very interesting future
direction is to obtain tight bounds for the trade-offs 
between Lipschitz constants and approximation guarantees
for the rest of the problems we consider in our work,
including matchings and $\bm b$-matchings.
{Understanding the tight Lipschitz constant for specific problems yields a better comprehension of which problems have (approximately) optimal solutions
that can be tracked in a stable way}.

\section*{Acknowledgments}
We thank Adrian Vladu and Marco Pirazzini for the insightful discussions.
Felix Zhou acknowledges the support of the Natural Sciences and Engineering Research Council of Canada (NSERC).
Quanquan C. Liu and Felix Zhou are supported by a Google
Academic Research Award and NSF Grant \#CCF-2453323.

\begingroup
\sloppy
\printbibliography
\endgroup

\appendix
\addtocontents{toc}{\protect\setcounter{tocdepth}{1}} %

\clearpage
\section{Additional Tools}
\subsection{Nonexpansiveness of Strongly Convex Gradient}\label{apx:nonexpansive gradient}

In this section, we prove \Cref{lem:nonexpansive gradient},
which is restated below for convenience.
\nonexpansiveGradient*
We note that the case of $K=\R^n$ is stated and proven in \cite[Lemma 3.7.3]{hardt2016train}
and the proof also extends to a convex feasible region,
as we show below.

We start with the following simple fact:
\begin{proposition}
 Let $g: K \to \R$ be a $\sigma$-strongly convex function
 over $K$ that is twice continuously differentiable.
 Then, $g - \frac\sigma2 \norm{\cdot}_2^2$ is convex over $K$.
\end{proposition}
\begin{proof}
    For $\bm x, \bm y\in K$,
    \begin{align*}
        &g(\bm y) - g(\bm x) - \frac\sigma2 (\norm{\bm y}_2^2 - \norm{\bm x}_2^2) \\
        &= \iprod{\grad g(\bm x), \bm y-\bm x} + \frac12 (\bm y-\bm x)^\top H(\bm \xi) (\bm y-\bm x) - \frac\sigma2 (\norm{\bm y}_2^2 - \norm{\bm x}_2^2) \tag{$\bm \xi\in [\bm x, \bm y]$} \\
        &\geq \iprod{\grad g(\bm x), \bm y-\bm x} + \frac\sigma2 \norm{\bm y-\bm x}_2^2 - \frac\sigma2 (\norm{\bm y}_2^2 - \norm{\bm x}_2^2) \tag{$H|_K\succeq \sigma I|_K$} \\
        &= \iprod{\grad g(\bm x), \bm y-\bm x} - \sigma \iprod{\bm y, \bm x} + \sigma \norm{\bm x}_2^2 \\
        &= \iprod{\grad g(\bm x) - \sigma \bm x, \bm y-\bm x} \\
        &= \iprod{\grad [g(\bm x) - \frac\sigma2\norm{\bm x}_2^2], \bm y-\bm x}.
    \end{align*}
    Hence $g - \frac\sigma2 \norm{\cdot}_2^2$ is convex over $K$.     
\end{proof}

Recall that convexity and $L$-smoothness implies the \emph{co-coercivity} of gradient:
\[
    \iprod{\grad g(\bm y) - \grad g(\bm x), \bm y-\bm x}
    \geq \frac1L \norm{\grad g(\bm y) - \grad g(\bm x)}_2^2.
\]
Then, we have the following:
\begin{proposition}\label{pro:co-coercivity-of-g}
    Let $g: K\to \R$ be $\sigma$-strongly convex, $L$-smooth, and twice continuously differentiable.
    For any $\bm x\neq \bm y\in K$,
    we have
    \[
        \iprod{\grad g(\bm  y) - \grad g(\bm x), \bm y-\bm x} 
        \geq \frac1{\sigma + L} \norm{\grad g(\bm y) - \grad g(\bm x)}_2^2
        + \frac{\sigma L}{\sigma + L} \norm{\bm y-\bm x}_2^2    
    \]
\end{proposition}
\begin{proof}
    $g - \frac\sigma2 \norm{\cdot}_2^2$ is convex and remains $(L-\sigma)$-smooth 
    so co-coercivity implies 
    \begin{align*}
        &\iprod{\grad g(\bm y) - \grad g(\bm x), \bm y- \bm x} - \sigma \iprod{\bm y- \bm x, \bm y-\bm x} \\
        &\geq \frac1{L-\sigma} \norm{\grad g(\bm y) - \grad g(\bm x) - \sigma (\bm y - \bm x)}_2^2 \\
        &= \frac1{L-\sigma} \norm{\grad g(\bm y) - \grad g(\bm x)}_2^2 
        - 2\frac{\sigma}{L-\sigma} \iprod{\grad g(\bm y) - \grad g(\bm x), \bm y-\bm x}
        + \frac{\sigma^2}{L-\sigma} \norm{\bm y - \bm x}_2^2.
    \end{align*}
    Rearranging yields the desired inequality.
\end{proof}

We are now ready to prove \Cref{lem:nonexpansive gradient}.
\begin{pf}[\Cref{lem:nonexpansive gradient}]
   Let $G_{g,\eta}(\bm x) \coloneqq \bm x - \eta \grad g(\bm x)$ be the gradient update . For $\bm x, \bm y\in K$,
    \begin{align*}
        &\norm{G_{g, \eta}(\bm x) - G_{g, \eta}(\bm y)}_2^2 \\
        &= \norm{\bm x - \bm y}_2^2 - 2\eta \iprod{\grad g(\bm x) - \grad g(\bm y), \bm x- \bm y} + \eta^2 \norm{\grad g(\bm x) - \grad g(\bm y)}_2^2 \\
        &\leq \left( 1 - \frac{2\eta\sigma L}{\sigma + L} \right) \norm{\bm y-\bm x}_2^2 + \eta\left( \eta-\frac2{\sigma + L} \right) \norm{\grad g(\bm y) - \grad g(\bm x)}_2^2 &\tag{by \Cref{pro:co-coercivity-of-g}} \\
        &\leq \left( 1 - \frac{2\eta\sigma L}{\sigma + L} \right) \norm{\bm y-\bm x}_2^2. & \tag{by $\eta\leq \frac2{\sigma+L}$}
    \end{align*}
    Note that $\sigma \leq L$ so that $\nicefrac{2L}{(\sigma+L)}\geq 1$.
    Hence choosing $\eta \leq \nicefrac1L$ implies that $G_{g, \eta}$ is $(1-\eta\sigma)$-expansive.
\end{pf}

\subsection{Weyl's Inequality}
\begin{theorem}[Weyl; Corollary 4.3.15 in \cite{horn2012matrix}]\label{thm:weyl-inequality}
    Let $A, B$ be $n\times n$ Hermitian (real symmetric) matrices.
    Then
    \[
        \lambda_i(A) + \lambda_1(B)
        \leq \lambda_i(A+B)
        \leq \lambda_i(A) + \lambda_n(B),
        \qquad i\in [n].
    \]
\end{theorem}

In particular,
if $A = \mcal L(G, \bm w)$ is the weighted Laplacian matrix of a graph $G$ with edge weights $\bm w$
and $B = \mcal L(uv, \delta)$ is the rank 1 update matrix that increases the edge weight of $uv$ by $\delta > 0$ for some $uv\in E(G)$,
then
\[
    \lambda_i(\mcal L(G, \bm w))
    \leq \lambda_i(\mcal L(G, \bm w + \delta \ones_{uv}))
    \leq \lambda_i(\mcal L(G, \bm w)) + \lambda_n(B)
    \leq \lambda_i(\mcal L(G, \bm w)) + 2\delta.
\]

\subsection{Finite Perturbation Bounds for Pointwise Lipschitz Continuity}\label{sec:finite bounds}
Although the definition of a pointwise Lipschitz algorithm holds under a limit,
we can still obtain finite perturbation bounds, which will be useful to show lower bounds.

\begin{theorem}\label{thm:finite perturbation}
    Let $\bm w, \tilde{\bm w}\in \R_{\geq 0}^E$ be two weight vectors for a graph $G=(V,E)$.
    Let $w: [0, 1]\to \R_{\geq 0}^E$ be a continuous path from $\bm w(0) = \bm w$ to $\bm w(1) = \tilde{\bm w}$ that is {component-wise} monotone,
    \ie{}, $r \leq \tilde r$ implies $\abs{w_i(r) - w_i(0)} \leq \abs{w_i(\tilde r) - w_i(0)}$ for all $i$.
    If for all $r\in [0, 1]$,
    $\mcal A(G, \bm w(r))$ has pointwise Lipschitz constant of $c_r$ with respect to some metric $d$,
    then we have
    \[
        \frac{\EMD_d(\mcal A(G, \bm w), \mcal A(G, \tilde{\bm w}))}{\norm{\bm w-\tilde{\bm w}}_1}
        \leq \sup_{r\in [0, 1]} c_r
        =: c_{\sup}.
    \]
\end{theorem}

Recall that a subset of $S\sset \R^n$ is said to be \emph{compact}
if for any cover of $S$,
say $ \bigcup_{\iota\in I} C_\iota\supseteq S$,
by open subsets $C_\iota\sset \R^n$,
there exists a finite subcover
$\bigcup_{i\in [n]} C_{\iota_i} \supseteq S$ for some $n < \infty$.
Recall also the Heine-Borel theorem \cite[Theorem 2.41]{rudin1964principles} from elementary analysis
that states a subset of a finite-dimensional vector space is compact
if and only if it is closed and bounded.

The intuition of \Cref{thm:finite perturbation} is that there is a local neighborhood at each point along the path from $\bm w\to \tilde{\bm w}$ where the pointwise Lipschitz constant (approximately) holds.
Then by taking small steps forward on the path,
we can translate the pointwise Lipschitz bound to a bound for $\bm w, \tilde{\bm w}$.
However,
in order to ensure a finite number of steps suffice,
we crucially rely on compactness and continuity.
\begin{pf}[\Cref{thm:finite perturbation}]
    Fix $\varepsilon > 0$.
    By the definition of a limit superior within the definition of the pointwise Lipschitz constant (\Cref{def:pointwise lipschitz}),
    for each $\bm w(r), r\in [0, 1]$,
    there is some open ball $B_r\ni \bm w(r)$ such that for every $\bm w(r)\neq \bm w\in B_r$,
    \[
        \frac{\EMD_d( \mcal A(G, \bm w(r)), \mcal A(G, \bm w) )}{\norm{\bm w(r) - \bm w}_1}
        \leq c_r + \varepsilon.
    \]

    As the continuous pre-image of open sets is open \cite[Theorem 4.8]{rudin1964principles},
    $w^{-1}(B_r)$ contains an open interval about $r$.
    Define $I_r \sset w^{-1}(B_r)$ to be such an open interval.
    Now,
    $\bigcup_{r\in [0, 1]} I_r\supseteq [0, 1]$ is an open cover of $[0, 1]$,
    a compact (closed and bounded) set.
    By the definition of compactness,
    we can extract a finite subcover,
    say indexed by the points $0 = r_0 < r_1 < \dots < r_N = 1$ for some $N < \infty$.
    By deleting and then adding points if necessary,
    we may assume without loss of generality that $r_{i+1}\in I_{r_i}$ for all $i=0, 1, \dots, N-1$.

    It follows that
    \begin{align*}
        &\EMD_d( \mcal A(G, \bm w(0)), \mcal A(G, \bm w(1)) ) \\
        &\leq \sum_{i=1}^{N} \EMD_d( \mcal A(G, \bm w(r_{i})), \mcal A(G, \bm w(r_{i-1})) ) \\
        &\leq (c_{\sup} + \varepsilon) \sum_{i=1}^N \norm{\bm w(r_i) - \bm w(r_{i-1})}_1 \\
        &= (c_{\sup} + \varepsilon) \norm{\bm w(0) - \bm w(1)}_1.
    \end{align*}
    Note that the last equality crucially uses the fact that our path is monotone.
    By the arbitrary choice of $\varepsilon > 0$,
    we must have
    \[
        \EMD_d(\mcal A(G, \bm w), \mcal A(G, \tilde{\bm w}))
        \leq c_{\sup} \norm{\bm w - \tilde{\bm w}}_1,
    \]
    concluding the proof of the theorem.
\end{pf}

\section{Na\"ive Algorithm for Minimum Vertex \texorpdfstring{$S$-$T$}{S-T} Cut}\label{sec:naive min s-t cut}
In this section,
we describe a simple warmup for {the} minimum $S$-$T$ cut algorithm using {unweighted regularization} {techniques} from prior works~\cite{kumabe2022lipschitz}
that yields sub-optimal guarantees.
In particular,
we prove the following theorem.
\begin{theorem}\label{thm:naive-min-s-t-cut}
    Let $\Lambda > 0$.
    There is a randomized polynomial-time approximation algorithm for undirected minimum vertex $S$-$T$ cut
    that outputs a vertex subset $A\sset V$ with the following guarantees.
    \begin{enumerate}[(a)]
      \item $A$ is a $(1, O(\Lambda n))$-approximate $S$-$T$ cut in expectation.
      \item The algorithm has a pointwise Lipschitz constant of $O\left( \frac{n}{\Lambda} \right)$.
    \end{enumerate}
\end{theorem}

In comparison with \Cref{thm:s-t-min-cut-additive},
we must set the regularization parameter $\Lambda = \nicefrac{\lambda_2}{n}$ to match the additive error of $O(\lambda_2)$.
However,
this yields a Lipschitz constant of $\nicefrac{n^{2}}{\lambda_2}$,
which is strictly worse than the guarantees of \Cref{thm:s-t-min-cut-additive}.

One technical detail is that in order to correctly set the parameter $\Lambda = \nicefrac{\lambda_2}{n}$,
we must estimate the weight-dependent parameter $\lambda_2$ in a stable manner.
While \cite{kumabe2022lipschitz} demonstrate how to achieve this via a sampling technique,
another advantage of our algorithm obtained through PGTA is avoiding the need to smoothly estimate weight-dependent parameters to obtain a fractional Lipschitz solution.

\subsection{Obtaining a Fractional Lipschitz Solution}
Let
\[
    \mathrm{cut}_{G, w}(y)
    \coloneqq \sum_{uv\in E} w_{uv} \abs{y_u - y_v}
\]
denote the fractional cut objective.
Recall the LP relaxation \Cref{eq:s-t-min-cut-lp} of minimum $S$-$T$ cut used in \Cref{subsec:s-t-cut-fractional}.
We restate a simplified version below.
\begin{align}
    \begin{array}{llr}
        \text{minimize} & \displaystyle \sum_{uv\in E} w_{uv} \abs{y_u - y_v} & \\
        \text{subject to} & y_{s_0} = 0 & \\
        &y_{t_0} = 1 & \\
        &y_s = y_{s_0} & \forall s\in S \\
        &y_t = y_{t_0} & \forall t\in T \\
        &y_v \in [0, 1] & \forall v\in V
    \end{array}
    \label[LP]{eq:naive-min-s-t-cut-lp}        
\end{align}

We show the following:
\begin{lemma}\label{lem:naive-min-s-t-cut-fractional}
    There is an algorithm that outputs a feasible solution to \Cref{eq:naive-min-s-t-cut-lp}
    such that
    \mbox{$\mathrm{cut}_G(x) \leq \OPTLP + \frac{\Lambda n}{2}$}
    with Lipschitz constant $O(n/\Lambda)$.
\end{lemma}

Following the technique of \cite{kumabe2022lipschitz},
we choose an unweighted $\ell_2$-regularizer
with regularization parameter $\Lambda > 0$.
\begin{align}
	\begin{array}{lll}
		\text{minimize} & \displaystyle \mathrm{cut}_G(y) + \frac\Lambda2 \|y\|_2^2 \\
            \text{subject to} & y_{s_0} = 0 & \\
            &y_{t_0} = 1 & \\
            &y_s = y_{s_0} & \forall s\in S \\
            &y_t = y_{t_0} & \forall t\in T \\
            &y_v \in [0, 1] & \forall v\in V
	\end{array}
	\label[LP]{eq:regularized-lp-min-cut}
\end{align}
This is $(\Lambda/n)$-strongly convex with respect to the $\ell_1$-norm.
The proof of \Cref{lem:naive-min-s-t-cut-fractional} follows from the next two lemmas.

\begin{lemma}\label{lem:warmup:regularized-fractional-approx}
    Let $y\in \mathbb{R}^V$ be the minimizer of \Cref{eq:regularized-lp-min-cut}.
    Then, we have \mbox{$\mathrm{cut}_G(y) \leq \OPTLP + \frac\Lambda2 n$}.
\end{lemma}

\begin{proof}
    Let $y^\star \in \mathbb{R}^V$ be the minimizer of \Cref{eq:naive-min-s-t-cut-lp}.
    Then, we have
    \begin{align*}
        \mathrm{cut}_G(x) + \frac\Lambda2 \|y\|_2^2
        \leq \mathrm{cut}_G(y^\star) + \frac\Lambda2 \|y^*\|_2^2
        \leq \mathrm{cut}_G(y^\star) + \frac\Lambda2 n.
    \end{align*}
\end{proof}

{As mentioned in \Cref{sec:overview:pgta},
\cite{kumabe2022lipschitz} incur an additive error proportional to the maximum range of the regularizer
in order to obtain stable fractional solutions.
Moreover,
any $\Lambda$-strongly convex regularizer over $[0, 1]^n$ has range $\Omega(\Lambda n)$.
Thus, this seems like an inherent limitation of the \cite{kumabe2022lipschitz} analysis rather than the choice of regularizer.
As we see in \Cref{sec:min-cut-through-pgm},
PGTA allows us to sidestep this limitation by unlocking the use of weighted regularizers.
On the other hand,
as we demonstrate below,
the \cite{kumabe2022lipschitz} analysis requires unweighted regularizers.}

\begin{lemma}\label{lem:naive-min-s-t-cut-fractional-lipschitz}
	The Lipschitz constant of the regularized LP is $O(n/\Lambda)$.
\end{lemma}

\begin{proof}
	Let $w \in \mathbb{R}^E$ be a weight vector and $w' \in \mathbb{R}^E$ be the vector obtained from $w$ by increasing the $\hat e$-th coordinate by $\delta$, where $\hat e = ab \in E$.
	By \Cref{lem:one perturbation},
        it suffices to analyze this case in order to bound the Lipschitz constant.

	Let $h_w(y) = h_{w, G}(y) \coloneqq \cut_{G, w}(y) + \frac\Lambda2\norm{y}_2^2$ be the objective function of~\eqref{eq:regularized-lp-min-cut}.
	Let $y$ and $y'$ be the minimizers of $h_w$ and $h_{w'}$, respectively.
	Let $D \coloneqq h_w(y') - h_w(y) + h_{y'}(x) - h_{w'}(y')$.
    As the regularization terms cancel out in $D$, we have
	\begin{align*}
		& D = \sum_{uv \in E} w_{uv} \abs{y'_u-y'_v} 
            - \sum_{uv \in E} w_{uv} \abs{y_u - y_v} 
            + \sum_{uv \in E} w'_{uv} \abs{y_u-y_v} 
            - \sum_{uv \in E} w'_{uv} \abs{y_u'-y_v'} \\
		& = (w'_{\hat{e}} - w_{\hat{e}})(\abs{y_a-y_b}- \abs{y_a'-y_b'}) \\
		& = \delta (\abs{y_a-y_b} - \abs{y_a'-y_b'}).
	\end{align*}

	By the $(\Lambda/n)$-strong convexity of $x\mapsto \frac\Lambda2 \|\cdot\|_2^2$ with respect to the $\ell_1$ norm
    and the optimality of $y, y'$ with respect to $h_w, h_{w'}$, 
    respectively,
    we have
	\begin{align*}
		& h_w(y') - h_w(y) 
		\geq \frac\Lambda{n}\left\|y'-y\right\|_1^2. \\
		& h_{w'}(y) - h_{w'}(y') 
		\geq \frac\Lambda{n}\left\|y'-y\right\|_1^2. 
	\end{align*}
	Summing them up, we get 
	\[
		\frac{2\Lambda}{n} \norm{y'-y}_1^2
            \leq D
            = \delta (\abs{y_a-y_b}- \abs{y_a'-y_b'}).
	\]
        An application of the reverse triangle inequality yields
        \begin{align*}
            \frac{2\Lambda}{n} \norm{y'-y}_1^2
            &\leq \delta (\abs{y_a-y_a'} + \abs{y_b-y_b'})
            \leq \delta \norm{y-y'}_1 \\
            \norm{y'-y}_1
            &\leq \frac{\delta n}{2\Lambda}. \qedhere
        \end{align*}
\end{proof}

We remark that in the proof of \Cref{lem:naive-min-s-t-cut-fractional-lipschitz},
the regularization terms must cancel out in the expression $D$.
The analysis of \cite{kumabe2022lipschitz} does not extend beyond this case.
However,
PGTA is able to handle the case of a weighted regularizer where this cancellation does not happen,
by reducing to the case of an unweighted regularizer at each iteration of PGM.

\subsection{Threshold Rounding}
The rounding algorithm is the classic threshold rounding (\Cref{lem:threshold-rounding}),
where we draw a uniform random threshold $\tau\sim [0, 1]$
and output the vertices with fractional solution $S_\tau = \set{v\in V: y_v\leq \tau}$.
Applied to our setting,
it preserves the fractional objective function in expectation
as well as the fractional Lipschitz constant.
Since we study threshold rounding in further detail in \Cref{sec:min-cut-through-pgm} and obtain stronger results,
we omit a formal analysis of the rounding step here.

This concludes the proof of \Cref{thm:naive-min-s-t-cut}.

\end{document}